\documentclass[11pt]{article}

\usepackage{amsmath}
\usepackage{amssymb}
\usepackage{amsthm}
\usepackage{amsfonts}
\usepackage{mathtools}
\usepackage[pagebackref,colorlinks=true,urlcolor=blue,linkcolor=blue]{hyperref}
\usepackage{prettyref}
\usepackage{tikz}
\usetikzlibrary{decorations.pathreplacing}
\usepackage{float}
\usepackage{caption}
\usepackage{subcaption}
\usepackage[margin = 3cm]{geometry}
\usepackage[colorinlistoftodos]{todonotes}

\usepackage[normalem]{ulem}

\newcommand{\newreptheorem}[2]{%
	\newtheorem*{rep@#1}{\rep@title}
	\newenvironment{rep#1}[1]{%
		\def\rep@title{#2 \ref*{##1}}%
		\begin{rep@#1}}%
		{\end{rep@#1}}
}
\newtheorem{theorem}{Theorem}[section]
\newtheorem{definition}[theorem]{Definition}
\newtheorem{corollary}[theorem]{Corollary}
\newtheorem{lemma}[theorem]{Lemma}
\newreptheorem{theorem}{Theorem}
\newreptheorem{lemma}{Lemma}
\newtheorem{claim}[theorem]{Claim}

\theoremstyle{definition}
\newtheorem{algorithm}{Algorithm}

\newcommand{\savehyperref}[2]{\texorpdfstring{\hyperref[#1]{#2}}{#2}}
\newrefformat{eq}{\savehyperref{#1}{\textup{(\ref*{#1})}}}
\newrefformat{lem}{\savehyperref{#1}{Lemma~\ref*{#1}}}
\newrefformat{def}{\savehyperref{#1}{Definition~\ref*{#1}}}
\newrefformat{thm}{\savehyperref{#1}{Theorem~\ref*{#1}}}
\newrefformat{cor}{\savehyperref{#1}{Corollary~\ref*{#1}}}
\newrefformat{sec}{\savehyperref{#1}{Section~\ref*{#1}}}
\newrefformat{app}{\savehyperref{#1}{Appendix~\ref*{#1}}}
\newrefformat{fig}{\savehyperref{#1}{Figure~\ref*{#1}}}
\newrefformat{item}{\savehyperref{#1}{Item~\ref*{#1}}}
\newrefformat{fact}{\savehyperref{#1}{Fact~\ref*{#1}}}
\newrefformat{test}{\savehyperref{#1}{Test~\ref*{#1}}}
\newrefformat{ex}{\savehyperref{#1}{Example~\ref*{#1}}}
\newrefformat{claim}{\savehyperref{#1}{Claim~\ref*{#1}}}
\newrefformat{alg}{\savehyperref{#1}{Algorithm~\ref*{#1}}}

\DeclareMathOperator*{\E}{\mathbb{E}} 
\newcommand{\Paren}[1]{\left(#1\right)} 
\newcommand{\Brac}[1]{\left[#1\right]} 
\newcommand{\abs}[1]{\lvert#1\rvert} 
\newcommand{\Abs}[1]{\left\lvert#1\right\rvert} 
\newcommand{\Norm}[1]{\Vert #1 \Vert} 
\newcommand{\inprod}[2]{\langle #1, #2 \rangle} 
\DeclareMathOperator{\poly}{poly} 
\newcommand{\sett}[2]{\left\{ #1 \left| \; \vphantom{#1 #2} \right. #2  \right\}} 
\newcommand{\set}[1]{\{#1\}} 
\newcommand{\Set}[1]{\left\{#1\right\}} 
\newcommand{\reals}{{\mathbb R}}
\newcommand{\st}{\text{ s.t. }} 

\newcommand{\R}{\mathbb{R}}
\newcommand{\bias}{\mathrm{bias}}
\def\bits{\{0,1\}}
\def\S{{\cal S}}

\newcommand{\D}{\mathcal{D}} 
\newcommand{\dist}{\text{dist}}
\newcommand{\cind}[1]{a(#1)}

\newcommand {\ra} {\right \rangle}
\newcommand {\la} {\left  \langle}
\newcommand{\norm}[1]{{\left\|{#1}\right\|}}

\newcommand\remove[1]{{}}
\def\eps{\varepsilon}
\renewcommand{\epsilon}{\varepsilon}

\def\bip{{\rm bip}}

\begin{document}

\title{List-Decoding with Double Samplers \footnote{A preliminary version of this paper appeared in \emph{Proc.\ $30$th Annual {ACM}-{SIAM} Symp.\ on Discrete
			Algorithms (SODA)}, 2019~\cite{DINURHKNT2019}}}
\author{Irit Dinur\thanks{Weizmann Institute of Science, ISRAEL. email: {\tt irit.dinur@weizmann.ac.il}. Supported by ERC-CoG grant number 772839.}
	\and Prahladh Harsha\thanks{Tata Institute of Fundamental Research, INDIA. email: {\tt prahladh@tifr.res.in}. Research supported by the Department of Atomic Energy,
		Government of India, under project no. 12-R\&D-TFR-5.01-0500 and in part by UGC-ISF grant and the Swarnajayanti Fellowship. Part of the work was done when the author was visiting
		the Weizmann Institute of Science.}
	\and Tali Kaufman\thanks{Bar-Ilan University, ISRAEL. email: {\tt kaufmant@mit.edu}. Supported by a BSF grant and an ERC grant.}
	\and Inbal Livni Navon\thanks{Weizmann Institute of Science, ISRAEL. email: {\tt inbal.livni@weizmann.ac.il}. Supported by Irit Dinur's ERC-CoG grant number 772839.}
	\and Amnon Ta-Shma\thanks{Tel-Aviv University, ISRAEL. email: {\tt amnon@tau.ac.il}.  Supported by ISF grant no. 952/18.}}
\maketitle

\begin{abstract}
We strengthen the notion of \emph{double samplers}, first introduced by Dinur and Kaufman~[{\em Proc. $58$th FOCS}, 2017], which are samplers with additional combinatorial properties, and whose existence we prove using high dimensional expanders.

The ABNNR code construction ~[{\em IEEE Trans. Inform. Theory}, 38(2):509--516] achieves large distance by starting with a base code $C$ with moderate distance, and then amplifying the distance using a sampler. We show that if the sampler is part of a larger {double sampler} then the construction has an {efficient} list-decoding algorithm. Our algorithm works even if the ABNNR construction is not applied to a base code $C$ but to any string. In this case the resulting code is  {\em approximate}-list-decodable, i.e. the output list contains an approximation to the original input.

Our list-decoding algorithm works as follows: it uses a local voting scheme from which it constructs a unique games constraint graph. The constraint graph is an expander, so we can solve unique games efficiently. These solutions are the output of the list-decoder. This is a novel use of a unique games algorithm as a subroutine in a decoding procedure, as opposed to the more common situation in which unique games are used for demonstrating hardness results.

Double samplers and high dimensional expanders are akin to pseudorandom objects in their utility, but they greatly exceed random objects in their combinatorial properties. We believe that these objects hold significant potential for coding theoretic constructions and view this work as demonstrating the power of double samplers in this context.
\end{abstract}

\section{Introduction}
We develop the notion of a \emph{double sampler}, which is an enhanced sampler. An $(\alpha,\beta)$ sampler is a bipartite graph $G=(U,V,E)$ such that for every function $f:V\to[0,1]$ with expectation $\mu=\E_{v\in V} [f(v)]$, one has $|\mu_u-\mu|\leq \alpha$ for all but a $\beta$ fraction of the vertices $u$, where $\mu_u=\E_{v\sim u} [f(v)]$ (see \cite{Zuckerman1997} for more details).

Towards defining double samplers we observe that in every given sampler $G=(U,V,E)$, every $u\in U$ can be identified with the set of its neighbors $\sett{v\in V}{v\sim u}$. In this way $U$ is a collection of subsets of $V$. In the other direction, given a ground set $V$ and a collection of subsets $\set{ S\subset V}$, the graph $G$ pops out as the bipartite inclusion graph with an edge from $v\in V$ to $S$ iff $v \in S$.

A double sampler (see \prettyref{fig:double_samp} for an illustration) consists of a triple $(V_2,V_1,V_0)$, where $V_0$ is the ground set, $V_1$ is a collection of $m_1$-subsets of $V_0$ and $V_2$ is a collection of $m_2$-subsets of $V_0$, where $m_2>m_1$. We say that $(V_2,V_1,V_0)$ is an $(\alpha,\beta),(\alpha_0,\beta_0)$ {\em double sampler} if
\begin{itemize}
	\item The inclusion graphs on $(V_2,V_1)$ and $(V_2,V_0)$ are $(\alpha,\beta)$ samplers, the inclusion graph on $(V_1,V_0)$ is an $(\alpha+\alpha_0,\beta+\beta_0)$ sampler. An inclusion graph is a graph where we connect two subsets by an edge if one contains the other; here a single vertex is also considered to be a singleton subset.
	
	\item For every $T\in V_2$, let $V_1(T)= \{ S\in V_1\;:\; S\subset T\}$ be the sets in $V_1$ that are contained in $T$. Let $G_{|T}$ be the bipartite inclusion graph connecting elements in $T$ (viewed as elements in the ground set $V_0$) to subsets in $V_1(T)$. We require that for every $T\in V_2$, the graph $G_{|T}$ is an $(\alpha_0,\beta_0)$ sampler. We call this property the {\em locality property} of the double sampler.
\end{itemize}

Our definition of double samplers is stronger than the previous definition due to Dinur and Kaufman~\cite{DinurK2017}, that was missing the locality property\footnote{The main result of Dinur and Kaufman~\cite{DinurK2017} was proven directly from high dimensional expanders, and not from double samplers, so this locality property was used implicitly. It is possible that the result of Dinur and Kaufman~\cite{DinurK2017} can be proven directly from our revised definition of double samplers.}. Whereas the definition of Dinur and Kaufman~\cite{DinurK2017} can be obtained e.g. by concatenating two samplers, our definition herein is much stronger and carries properties not known to be obtained by any random construction.
It is quite remarkable that high dimensional expanders~\cite{LubotzkySV2005-exphdx,KaufmanO2018} give rise to an infinite family of double samplers for which $|V_1|,|V_2|= O(|V_0|)$:

\begin{theorem}[Informal, see formal version in \prettyref{thm:doublesampler}]\label{thm:ds exists informal}
	For every $\alpha,\beta$, $\alpha_0$, $\beta_0>0$ there are integers $m_1,m_2 = \poly(\frac{1}{\alpha\beta\alpha_0\beta_0}), D = \exp(\poly(\frac{1}{\alpha\beta\alpha_0\beta_0}))$, such that there is an explicit polynomial time construction of an $(\alpha,\beta),(\alpha_0,\beta_0)$ double sampler on $n$ vertices, for infinitely many $n\in\mathbb{N}$, such that the subsets in $V_i$ are of size $m_i$ and the bipartite inclusion graphs on $V_0,V_1$ and on $V_0,V_2$ have degree at most $D$.
\end{theorem}
\paragraph{On random double samplers.}
To appreciate the remarkableness of double samplers, think of concrete parameters such as $m_1=2,m_2=3$. A random construction amounts to placing $n$ vertices in $V_0$, a linear (in $n$) number of edges in $V_1$ and a linear number of triples in $V_2$. Consider for example the $G(n,p)$-like model, where triangles are chosen independently. In this case two triangles will almost never share an edge. In either case the inclusion graph on $V_1,V_2$ is highly disconnected, and certainly not a sampler\footnote{Observe that for the chosen parameters of $m_1=2$ and $m_2=3$, there are obvious limits on the $(\alpha,\delta)$ parameters of the sampler, since each triple is connected to at most $3$ edges.}.

We elaborate more on the construction of double samplers towards the end of the introduction.
\begin{figure}
	\centering
	\begin{tikzpicture}
	\def\length{10}
	\def\points{0,1,2,3}
	\def\layers{0/7,1/8,2/10}
	\foreach \j/\v in \layers{
		\foreach \i in \points {
			\draw [fill=black] (\i + \length*0.5-\v*0.5 ,\j*2) circle (0.1cm);}
		\draw [fill=black] (0.5*\v + 0.5*\length,\j*2) circle (0.1cm);
		\draw[dotted] (\length*0.5-\v*0.5  + 3 + 0.2,\j*2) -- (0.5*\v + 0.5*\length - 0.2,\j*2);
		\node at (\length+0.5,\j*2){$V_{\j}$};}
	\node[label=$T$](T) at (2,4) {};
	\foreach \i in {0,1,2} { \draw (2,4) -- (\i + \length*0.5-8*0.5,2); }
	\draw [decorate,decoration={brace,amplitude=10pt},xshift=0cm,yshift=0.2cm]
	(\length*0.5-8*0.5,2) -- (2 + \length*0.5-8*0.5,2) node [fill=white,midway,xshift=-0cm,yshift=0.6cm, align=center]
	{$V_1(T)$};
	\node[label=$S$](S) at (0.7,1.9){};
	\node[label=$x$](S) at (1.2,-0.6){};
	\foreach \i in {0,1,2} {
		\draw[blue] (\i + \length*0.5-7*0.5 ,0) -- (\length*0.5-8*0.5 ,2);	
	}
	\foreach \i in {0,2} {
		\draw[blue] (\i + \length*0.5-7*0.5 ,0) -- (1+\length*0.5-8*0.5 ,2);	
	}
	\foreach \i in {1,2} {
		\draw[blue] (\i + \length*0.5-7*0.5 ,0) -- (2+\length*0.5-8*0.5 ,2);	
	}
	\draw [blue,decorate,decoration={brace,amplitude=10pt,mirror},xshift=-0.2cm,yshift=0cm]
	(\length*0.5-8*0.5,2) -- (\length*0.5-8*0.5,0) node [blue,midway,xshift=-0.7cm,yshift=0cm, align=center]
	{$G_{|T}$};
	\draw [decorate,decoration={brace,amplitude=10pt,mirror},xshift=0cm,yshift=-0.2cm]
	(\length*0.5-7*0.5,0) -- (2 + \length*0.5-7*0.5,0) node [midway,xshift=-0cm,yshift=-0.6cm, align=center]
	{$T\subset V_0$};
	\end{tikzpicture}
	\caption{Double sampler. Each vertex $S\in V_1$ is a set containing $m_1$ elements from $V_0$, and each vertex $T\in V_2$ is a set containing $m_2$ elements from $V_0$. The edges in the graph denote inclusion, for example  $x\in S\subset T$.}\label{fig:double_samp}
\end{figure}

\paragraph{Samplers and distance amplification.}
Alon, Bruck, Naor, Naor and Roth~\cite{AlonBNNR1992} showed how to amplify the distance of any code, simply by pushing the symbols along edges of a sampler graph. Let us describe their encoding in a notation consistent with the above. We think of the graph as a sampler $G=(V_1,V_0=[n])$, where $V_1$ is a collection of $m$-sets of $[n]$.
Given an $n$-bit string $x\in\bits^n$, we place $x_i$ on the $i$-th vertex and then each subset $S\in V_1$ ``collects'' all of the symbols of its elements and gets a short string $x|_S :S\to\bits$. The resulting codeword is the sequence $E_G(x) := (x|_S)_{S\in V_1}$ which can be viewed as a string of length $|V_1|$ over the alphabet $\Sigma = \bits^m$. We refer to the mapping $x\mapsto E_G(x)$ as the ABNNR encoding.

If the string $x$ happens to come from an initial code $C\subset\bits^n$ with minimum distance $2\alpha$, then, altogether we get $E_G(C) := \sett{ E_G(x)}{x\in C}$. Assuming $G$ is an $(\alpha,\beta)$ sampler, the minimum distance of $E_G(C)$ is at least $1-\beta$. Of course the length of the words in $E_G(C)$ depends on the size of $\abs{V_1}$, so the shorter the better.

\paragraph{Approximate error correcting code.}
The ABNNR encoding itself has some interesting decoding properties. As an encoding, $E_G$ does not have a noticeable minimal distance, but still, it is an \emph{approximate}-list-decodable error correcting code. Namely, for every $z\in \Sigma^{\abs{V_1}}$
there is a short list $L_z\subset \{0,1\}^n$, such that $L_z$ contains an approximation to every string $x$ whose encoding $E_G(x)$ is close to $z$ (in Hamming distance).

The elegant encoding of ABNNR is very local and easy to compute in the forward direction (from $x$ to $E_G(x)$), and indeed it has been found useful in several coding theory constructions, e.g. \cite{GuruswamiI2005, KoppartyMRS2017}. In this work we study the inverse question, also known as decoding: given a noisy version of $E_G(x)$, find an approximation to $x$ (or $x$ itself, in case $x\in C$). Moreover, we wish to be able to recover from as many errors as possible.

\paragraph{Decoding and list-decoding} A decoding algorithm for $E_G$ gets as input a string $(z_S)_{S\in V_1}$, and needs to find a word $x$ such that $x|_S = z_S$ for as many $S\in V_1$ as possible.
A natural approach is the ``maximum likelihood decoding'' algorithm: assign each $i\in [n]$ the most likely symbol, by looking at the ``vote'' of each of the subsets $S\ni i$,
\[ x'_i := {\rm majority}_{S:S\ni i}[ z_S(i)].\]
The resulting $x'$ is an approximation for $x$, and in the case where $x\in C$ we can run the decoder of $C$ and retrieve $x$.
Assuming $G$ is a good sampler, this approach gives an approximate-decoding algorithm for $E_G$ that recovers $x$ from error rates almost up to $1/2$.

Going beyond the unique-decoding radius, we show using the Johnson bound that the ABNNR encoding on a sampler graph is (combinatorially) approximate-list-decodable, up to an error rate approaching one as the sampler parameters go to zero (see \prettyref{sec:approximate_ecc}).  However, the maximum likelihood decoder stops working in this regime: one cannot rule out the situation where for each vertex $i$, both $0$ and $1$ symbols occur with equal likelihood, and it is not known, in general,\footnote{We remark that when the ABNNR encoding is applied over a base code $C$ with additional special properties it is possible that more can be done (see e.g \cite{GuruswamiI2005}), but our focus is on a generic decoding mechanism.} how to recover $x$.

Thus, it is natural to ask for an algorithm that approximate-list-decodes the ABNNR encoding up to the Johnson bound radius.
Our main result is a list-decoding algorithm that goes beyond the unique-decoding barrier of $1/2$  and works for error rates approaching $1$. The algorithm works whenever the underlying graph $G = (V_1,[n])$ is part of a double sampler, namely where there is a collection $V_2$ of sets of size $m_2>m_1=m$ so that the triple $(V_2,V_1, [n])$ is a double sampler.
\begin{theorem}[Main - informal, see  \prettyref{thm:main} and \prettyref{cor:main-ds-exists}]\label{thm:main-informal}
	For every $\gamma,\eps>0$ there exist $\alpha,\beta,\alpha_0,\beta_0>0$, integers $m_1,D$ and an $(\alpha,\beta),(\alpha_0,\beta_0)$-double sampler $(X=(V_2,V_1,[n]))$ such that $V_1 \subseteq \binom{[n]}{m_1}$ and $|V_1| \le D\cdot n$ and such that the following holds. Let $G$ be the restriction of $X$ to layers $V_1$ and $[n]$, and let $E_G:\bits^n\to(\bits^{m_1})^{V_1}$ be the ABNNR encoding defined above. Then there is a polynomial time algorithm which receives an input $z\in (\bits^{m_1})^{V_1}$, and outputs a list $L_z$ of size $O(\frac{1}{\gamma^2}) $ which includes an $\epsilon$-approximation for every $x\in \bits^n$ such that $\dist(E_G(x),z)\leq 1-\gamma$.
\end{theorem}
We omitted here the conditions on the constants and the parameter requirements on $X$, which appear in \prettyref{thm:main}. We remark that the dependence of $m_1$ and $c$ on $\gamma$ and $\eps$ is quite far from optimal, and is discussed in \prettyref{sec:code}.

Combining our main theorem with a unique-decodable base code $C$, we get a code $E_G(C)$ that is list-decodable, whenever $G=([n],V_1)$ is the first two layers of a double sampler,
\begin{corollary}[Informal, see \prettyref{cor:list-dec-full}]\label{cor:list-dec}
	For every $\gamma,\epsilon>0$, let $(X=(V_2,V_1,[n]))$ be a double sampler as above. Suppose $C\subset \bits^n$ is an error correcting code with a polynomial time unique-decoding algorithm from an $\epsilon$-fraction of errors, then the following holds. Let $G$ be the restriction of $X$ to layers $V_1$ and $[n]$. Then the code $E_G(C) = \sett{E_G(x)}{x\in C}$ has a polynomial time list-decoding algorithm from a $(1-\gamma)$-fraction of errors, with list size $O(\frac{1}{\gamma})$.
\end{corollary}
Notice that the list-decoding algorithm in the corollary outputs a shorter list than the approximate-list-decoding algorithm in our main theorem (length $1/\gamma$ versus $1/\gamma^2$). This is because in the case of error correcting codes we can prune the list more efficiently, as explained in \prettyref{sec:approximate_ecc}.

At this point the reader may be wondering how the double sampler property helps facilitate list-decoding.
Roughly speaking, a double sampler is a collection of (small) subsets that have both large overlaps as well as strong expansion properties. The expansion properties are key for distance amplification, and the large overlaps, again with good sampling properties, are key for the list-decoding algorithm.

\subsection{Related work}
There are several known list-decodable codes with efficient decoding algorithms, with varying parameters.
This includes codes which use algebraic structure, such as Reed Solomon codes \cite{Sudan1997,GuruswamiS1999}, folded Reed Solomon codes \cite{ParvareshV2005, GuruswamiR2008}, multiplicity codes \cite{Kopparty2015}, algebraic-geometric codes \cite{GuruswamiS1999,GuruswamiX2014} and constructions using a more combinatorial approach such as \cite{GuruswamiI2001,GuruswamiI2005, GuruswamiI2003}. Most of these constructions get parameters better than our constructions, and some of them get very close to optimal rate and alphabet size. Some of these constructions have other advantages, such as having linear time encoding and decoding algorithms.

Our construction starts with a basic code with constant distance and amplifies it to a code that can be list-decoded from a distance approaching $1$, and does it in a black box way. Among all the codes listed above only the construction by Guruswami and Indyk~\cite{GuruswamiI2003} does something similar. However, their reduction is recursive and has no underlying double sampler.

It is interesting to compare our construction to the work of Trevisan~\cite{Trevisan2003} who showed that derandomized direct product theorems can be used to transform a code from unique-decodable into list-decodable. The work of Impagliazzo, Jaiswal, Kabanets and Wigderson \cite{ImpagliazzoJKW2010} also uses derandomized direct product to create list-decodable codes. The encoding we use in this work can also be viewed as a derandomized
direct product encoding. The rate of the resulting code depends on the quality of the derandomization. Whereas previous derandomizations had a sub constant rate, double samplers give a much stronger derandomization that results in constant rate.

The ABNNR construction starts with a binary code and constructs a code over a larger alphabet by having each vertex $S \in V_1$ collect the bits from all indices $i \in S$. It is natural to consider the \emph{direct sum} variant, where each $S \in V_1$ XORs these bits together into a single bit. It would not be hard to adapt our algorithm to that setting, though we haven't explicitly done so. A subsequent work \cite{AlevJQST2020} went quite a step further and showed an algorithm for list-decoding the direct sum code for a strictly broader family of samplers. Their algorithm uses the Sum-of-Squares (SOS) semi-definite programming hierarchy to list-decode the direct sum for all samplers that satisfy a ``splittability'' condition which they introduce. These include not only samplers that come from double samplers, but also samplers that come from random walks on expanders, which are not covered by our work.

Even more recently, \cite{JeronimoST2021} presented an algorithm for list-decoding the direct sum code in nearly linear time. Furthermore, their work gives an efficient decoding algorithm for the error correcting code of \cite{Tashma2017}. The code in \cite{Tashma2017} is the first explicit construction of binary codes with distance close to half and nearly optimal rate. Prior to \cite{JeronimoST2021}, it was not known whether the binary code of \cite{Tashma2017} can be efficiently decoded.
\subsection{The list-decoding algorithm}
On input $(z_S)_{S\in V_1}$, our algorithm starts out with a voting step, similar to the maximum likelihood decoder.  Here we vote not on the value of each bit $i\in[n]$ but rather on the value of $x$ restricted to an entire  set $T\in V_2$. Since the graph $X_{|T}$ between $V_1(T)$ and $T$ is a sampler (this is the locality property), we can come up with a short list of popular candidates for  $x|_T$. This is done by looking at $z_S$ for all subsets $S\in V_1$, $S\subset T$. We define
\[\forall T\in V_2,\quad list(T) := \{ \sigma\in\bits^T\;:\; \Pr_{S\subset T, S\in V_1}[z_S = \sigma|_S]>\eps/2 \;\}.   \]
Note that since $T$ has constant size, we are able to search exhaustively over all $\sigma\in \bits^T$ in constant time.

Given a list for each $T$, we now need to stitch these lists together, and here we again use the fact that $(V_2,V_1)$ is a good sampler. Whenever $T_1\cap T_2$ is significantly large, we will match $\sigma_1\in list(T_1)$ with $\sigma_2\in list(T_2)$ iff $\sigma_1|_{T_1\cap T_2}=\sigma_2|_{T_1\cap T_2}$. Moreover, the double sampler property allows us to come up with an {\em expander} graph whose vertex set is $V_2$, and whose edges connect $T_1$ to $T_2$ when they have significant overlap. This guarantees that for almost all edges $(T_1,T_2)$ there is a matching between the list of $T_1$ and the list of $T_2$.

At this point what we are looking at is a {\bf unique games instance}, where the said expander is the constraint graph, and the said matchings are the unique constraints.\footnote{For definitions, please see the Preliminary section.} We now make two important observations. First, a word with noticeable correlation with the received word, corresponds to a solution for the unique games instance with very high value (i.e., satisfying a large fraction of the constraints). Second, we have an efficient algorithm for finding a high-value solution, because the underlying unique games constraint graph is an expander! This is originally due to Arora, Khot, Kolla, Steurer, Tulsiani Vishnoi \cite{AroraKKSTV2008} but we actually use the variant of Makarychev and Makarychev~\cite{MakarychevM2010}, because it has better parameters. A more naive greedy belief propagation algorithm would fail miserably because it takes about $\log n$ steps to reach a typical point in an expander graph, and this accumulates an intolerable $\eps \cdot \log n \gg 1$ amount of error.

If we want an approximate error correcting code we are done. Otherwise, it remains to run the unique-decoding algorithm of $C$ on each of the solutions of the unique games instance, to remove any small errors, and this completes the list-decoding.

The above high level description gives the rough idea for our algorithm, but the implementation brings up some subtle difficulties, which we explain below.

Every set $T$ induces a constant size local view $E_G|_T$ on the code $E_G$, which has no reason to be an error correcting code, and in particular has no distance. Thus, there could be several valid candidates $\sigma\in\bits^T$ that are very close in Hamming distance. Suppose $\sigma,\sigma'\in list(T_1)$ differ only in a single bit, then for most $T_2\cap T_1$, we don't know which element in $list(T_2)$ should be matched to $\sigma$ and which to $\sigma'$. Saying it differently, what we really need to do is to approximately list-decode the local view.  Equivalently, for each $T$ we prune $list(T)$ and enforce minimal distance $r$ between each two list items, while holding   a ``covering'' property - that if $\sigma$ was in the initial list $List(T)$, then there exists some $\sigma'$ in the final list that is $r$-close to $\sigma$.

However, for reasons that become clear in the proof we need the following stronger property: We require that the pruned list covers all the elements in $List(T)$ with radius $r$, while elements in the pruned list are at least $R \gg r$ away from each other (think of $R$ as being $2r$). We show that there is a small set of possible radii $r$, such that for every $T$ at least one radius from the set is good, in a sense that pruning $list(T)$ with $r$ results in a pruned list such that its elements are $R=5r$ far from each other. Thus, the pruning algorithm chooses $r$ dynamically over $T$, see \prettyref{sec:well-sep}).

Given $T_1,T_2$ with $list(T_1),list(T_2)$ and the same radius $r$, we match $\sigma_1\in list(T_1)$ to $\sigma_2\in list(T_2)$ if they are close (with respect to $r$) on $T_1\cap T_2$. If however $T_1,T_2$ have different radii, we don't know how to match these lists correctly. Therefore, our unique games instance is created on a subgraph containing only those vertices $T$ that share the same radius $r$. We show that there exists such a subgraph which is itself an expander.

\subsection{Double samplers and high dimensional expanders}
Let us briefly explain how double samplers are constructed from high dimensional expanders (proving \prettyref{thm:ds exists informal}). A high dimensional expander is a $d$-dimensional simplicial complex $X$, which is just a hypergraph with hyperedges of size $\le d+1$ and a closure property: for every hyperedge in the hypergraph, all of its subsets are also hyperedges in the hypergraph. The hyperedges with $i+1$ elements are denoted $X(i)$, and the complex is said to be an expander if certain spectral conditions are obeyed, see \prettyref{sec:doublesampler}.

Dinur and Kaufman~\cite{DinurK2017} prove that a two-sided spectral high dimensional expander gives rise to a multi-partite graph with interesting spectral expansion properties. Kaufman and Oppenheim \cite{KaufmanO2020} proved a stronger bound which allows using one-sided spectral expander.
The multi-graph has vertices $X(d)\cup  X(d-1)\cup \ldots \cup X(0)$, and we place edges for inclusion. Namely, $S\in X(m_1)$ is connected by an edge to $T\in X(m_2)$ if $S\subset T$. It is shown that the graph induced by focusing on layers $i$ and $j$ has $\lambda(G(X(i),X(j))) \le \frac {i+1} {j+1} +o(1)$. We show, in \prettyref{sec:doublesampler}, that by narrowing our focus to three layers in this graph (namely, $X(m_2-1),X(m_1-1)$ and $X(0)$) we get a double sampler. This is proven by observing that the spectral properties are strong enough to yield a sampler (an expander mixing lemma argument suffices since we are only seeking relatively weak sampling properties).

{\paragraph{Better double samplers?}
	Double samplers with super-linear (polynomial and even exponential) size have appeared implicitly (or somewhat similarly as ``intersection codes'') in the works of impagliazzo, Jaiswal, Kabanets and Wigderson~\cite{ImpagliazzoKW2012,ImpagliazzoJKW2010}. Two concrete constructions were studied,
	\begin{itemize} \item The first where $V_i = \binom{V}{m_i}$, so $|V_i| \approx n^{m_i}$, for $n=\abs{V_0}$.
		\item The second where $V$ is identified with a vector space over some finite field and then $V_i$ consists of all $d_i$-dimensional subspaces of $V$. Here $|V_i| \approx n^{d_i}$.
	\end{itemize}
	These constructions could fit our encoding scheme but the polynomial size of the sampler means that the code rate would approach zero. In addition, our algorithm is only efficient when the sets in $V_2,V_1$ are very small (of constant or at most logarithmic size), so constructions with larger sets, such as restricting multivariate polynomials to lines also don't fit our algorithm.
	
	The current work is the first to construct double samplers with linear size. This raises the question of finding the best possible parameters for these objects.  In particular, for given sampler parameters $\alpha$ and $\delta$, how small can $|V_1|/|V_0|$ be?
	
	Our current construction is based on Ramanujan complexes of Lubotzky, Samuels and Vishne~\cite{LubotzkySV2005-exphdx} that are optimal with respect to the spectrum of certain Laplacian operators, and not necessarily with respect to obtaining best possible double samplers. It is an interesting challenge to meet and possibly improve upon these parameters through other constructions.
	
	Unlike other pseudorandom objects, there is no known random construction of a double sampler. In particular, we cannot use it as a yardstick for the quality of our parameters. It remains to explore what possible parametric limitations there are for these objects.

	\bigskip
	We believe that double samplers capture a powerful feature of high dimensional expanders whose potential merit more study. Previously, Dinur and Kaufman~\cite{DinurK2017} showed that high dimensional expanders give rise to a very efficient derandomization of the direct product code that is nevertheless still testable. Part of the contribution of the current work is a demonstration of the utility of these objects in a new context, namely of list-decoding.

	\subsection{Derandomized direct product and approximate-list-decoding}
	Our list-decoding algorithm can also be viewed in the context of decoding derandomized direct products.
	The \emph{direct product} encoding takes a string $g\in \bits^N$ and encodes it into  $Enc(g)=(g|_S)_{S\in\S}$ where
	$\S = \binom{[N]}{k}$ contains all possible $k$-subsets of $[N]$. An encoding with $|\S|\ll \binom{N}{k}$, as in this paper, is called a \emph{derandomized} direct product encoding.
	
	Direct products and derandomized direct products are important in several contexts, primarily for {\bf hardness amplification}. This type of amplification goes back to Yao's XOR lemma \cite{Yao1982,Levin1987} (which concerns direct sum, not direct product). In hardness amplification using direct product one begins with a string $g\in\bits^N$ that is viewed as a truth table of a function $g:\bits^n\to \bits$ (here $N = 2^n$), and analyzes the hardness of the new function defined by $Enc(g)$. A typical hardness amplification argument proceeds by showing that if no algorithm (in a certain complexity class) computes $g$ on more than $1-\eps_0$ of its inputs, then no algorithm computes $Enc(g)$ on more than $\eps$ of its inputs. Namely, $Enc(g)$ is much harder than $g$.

	Such a statement is proven, as first described in \cite{Trevisan2005,Impagliazzo2003}, through a (list-) decoding argument: given a hypothetical algorithm that computes $Enc(g)$ successfully on at least an $\eps$-fraction of the inputs, the approximate-list-decoder computes $g$ on $(1-\eps_0)$ of its inputs.
	
	Our list-decoding result falls short of being useful for hardness amplification, because it is not local, and hardness amplification requires an additional feature called local list-decoding which we discuss in the next subsection.

	\subsection{Future directions}
	Both this work and the work of Alex, Jeronimo, Quintana, Shashank and Tulsiani \cite{AlevJQST2020} have a global decoding algorithm, which reads the entire codeword before decoding. An interesting direction is to achieve local list-decoding. In local-list-decoding the algorithm receives an index $i$ and has query-access to a noisy version of $E_G(x)$. The algorithm queries the codeword in a few locations and should output $x_i$ with high probability (or a list including $x_i$).
	
	There is a significant technical hurdle that one faces, related to the diameter of the bipartite graph corresponding to $(V_1,V_0)$. In the local-list-decoding constructions analyzed in \cite{ImpagliazzoKW2012,ImpagliazzoJKW2010} (both derandomized and non-derandomized) the diameter is $O(1)$, and this is crucially used in the list-decoding algorithm. The reason is that decoding occurs through querying vertices whose distance from a given $v$ is bounded. In the above small-diameter situations nearly all of the vertices in the graph are in a constant distance from $v$, and they can't all be corrupted by an adversary.
	
	When we move to a linear-size derandomized direct product encoding, as we do in this work, we have a sparse graph and a super-constant diameter. Clearly balls of bounded radii remain very small in this case, and can easily be corrupted. This is what makes the approximate-list-decoding algorithm performed in our work much more challenging (even in the non-local setting), and the algorithm more complicated than the analogous task performed by \cite{ImpagliazzoKW2012,ImpagliazzoJKW2010}.
	
In a preliminary version of this manuscript, we asked if there can be an efficient decoding algorithm for the error correcting codes constructed in \cite{Tashma2017} by the last author. These codes are binary error correcting codes with distance close to half, that achieve nearly optimal rate, and whose encoding algorithm is similar to the one presented here. In a recent work, Jeronimo, Srivastava, and Tulsiani \cite{JeronimoST2021} resolved this question by presenting a nearly linear decoding algorithm for variants of the error correcting codes in \cite{Tashma2017}.

\section{Preliminaries and Notations}
\label{sec:prelim}
For $\sigma,\sigma'\in \Sigma^n$ and $S\subseteq n$ we define
\begin{eqnarray*}
	\dist_S(\sigma,\sigma') &=& \Pr_{i \in S}[\sigma_i \neq \sigma'_i].
\end{eqnarray*}
For $L \subseteq \Sigma^n$ we define
\begin{eqnarray*}
	\dist_S(\sigma,L) &=& \min_{\sigma' \in L}\dist_S(\sigma,\sigma').
\end{eqnarray*}
When $S=[n]$ we omit the subscript $S$.

An error correcting code $C$ is a function $C:\Sigma_0^n\rightarrow\Sigma_1^m$. It has distance $r$ if for every $x\neq y\in \Sigma_0^n$, $dist(C(x),C(y))\ge r$. Sometimes we identify the error correcting code $C$ with its image, i.e. $C\subset \Sigma_1^m$. Furthermore, $C$ is $(\eta,\ell)$ list-decodable if for every $y\in\Sigma_1^m$,
\[ \abs{\sett{x\in C}{\dist(x,y)\leq\eta}}\leq \ell. \]

An algorithm is said to uniquely decode an error correcting code $C\subset\Sigma_1^m$ from an $\epsilon$-fraction of errors, if for every $x\in\Sigma_1^m$ such that $\dist(x,C)\leq\epsilon$, it outputs $y\in C$ such that $\dist(x,y)\leq\epsilon$. An algorithm is said to $(\eta,\ell)$ list-decode an error correcting code $C:\Sigma_0^n\rightarrow\Sigma_1^m$, if for every $z\in \Sigma_1^m$, it outputs a list $L=\sett{x\in \Sigma_0^n}{\dist(C(x),z)\leq \eta}$, $\abs{L}\leq\ell$.\\

We use the following version of a Chernoff tail bound \cite{ChungL2006}. Let $Y_1,\dots Y_n$ be independent random variables, with $\Pr[Y_i=1]=p_i,\Pr[Y_i=0]=1-p_i$, and let $Y=\sum_{i=1}^n w_i Y_i$ for $w_i>0$. We define $\nu = \sum_{i=1}^nw_i^2p_i$, then
\[ \Pr[Y\leq (1-\delta)\E[Y]]\leq e^{-\frac{\delta^2(\E[y])^2}{2\nu}}, \]
\[ \Pr[Y\geq \E[Y]+\eta]\leq e^{-\frac{\eta^2}{2(\nu+w\frac{\eta}{3})}}, \]
for $w=\max_{i\in[n]}\{w_i\}$.
\subsection{Weighted graphs}
\label{sec:vertex weight}

We say $(G,W)$ is a weighted graph if $G=(V,E)$ is an undirected graph, and $W:E \to \reals_{\geq 0}$ is a weight function that associates with each edge $e$ a non-negative weight $w_e$. We have the convention that non-edges have zero weight. Given the edge weights $w_e$ the weight of a vertex is defined as
\begin{eqnarray*}
	w_v & := & \sum_{e:v \in e}w_e.
\end{eqnarray*}

The edge weights induce a distribution on edges (and vertices) which we denote by $\mu_G$ where $\mu_G(e)=\frac{w_{e}}{\sum_{e' \in E}w_{e'}}$ and similarly for vertices. We overload $\mu_G$ to denote both the distribution on edges and on vertices.  When the graph $G$ is clear, we omit the subscript $G$ from the distribution $\mu_G$.
We denote by $v\sim V$ a random vertex in the graph according to the distribution $\mu_G$, and by $e\sim E$ a random edge.
For a vertex $v\in V$, we denote by $u\sim v$ a random neighbor of $v$, according to the edge weights.

\begin{definition}
	\label{def:irregular}
	We say a distribution $\Pi$ over $V$ has \emph{irregularity} at most $D$, for an integer $D \in \mathbb{N}$, if  there is some $p \in (0,1]$ such that for every $v\in V$, $\Pi(v) \in\set{p,2p,\ldots,Dp}$.
	The irregularity of $v \in V$ is defined to be $\frac{\Pi(v)}{p}$.
\end{definition}
Clearly the uniform distribution has irregularity  $1$.

\subsection{Expanders}\label{sec:expanders}

Let $(G=(V,E),W)$ be a weighted graph.
\begin{itemize}
	\item The \emph{edge expansion} of $G$ is
	\begin{eqnarray*}
		h_G & := & \min_{V'\subset V,\mu(V')\leq\frac{1}{2}} \frac{\mu(E(V',V\setminus V'))}{\mu(V')},
	\end{eqnarray*}
	where $E(A,B)$ denotes the set of edges between $A$ and $B$.
	
	\item
	The normalized adjacency matrix $A$ of $G$ is defined by
	\begin{eqnarray*}
		A_{u,v} & = & \frac{w_{u,v}}{\sqrt{w_u w_v}}.
	\end{eqnarray*}
	
	\item
	$\lambda_2(G)$ denotes the second largest eigenvalue  (in absolute value) of $A$.
\end{itemize}	
In the case of bipartite graph $G=(U,V,E)$, let $A^\bip$ be the normalized adjacency matrix of $G$ defined by
$A^\bip_{u,v}=\frac{w_{u,v}}{\sqrt{w_u w_v}}$ for every $u\in U, v\in V$ ($A^\bip$ is not symmetric). Let $\lambda^\bip_2(G)$ be the second largest singular value of $A^\bip$.

\subsection{Samplers}

\begin{definition}(Sampler)
	A weighted bipartite graph $(G=(V_2,V_1,E),W)$ is an $(\alpha,\beta)$ sampler if for every $f:V_1 \to [0,1]$,
	\begin{eqnarray*}	
		\Pr_{v_2 \sim V_2} \left[  \Abs{\E_{v_1 \sim v_2}[ f(v_1)]-\E_{v_1 \sim V_1}[f(v_1)]} \ge \alpha \right] & \le  & \beta.
	\end{eqnarray*}
	When the distribution is uniform, we say $G$ is an $(\alpha,\beta)$ sampler.
\end{definition}

It is possible to convert a weighted sampler into an unweighted one, if the weights satisfy certain conditions.
\begin{definition}[Flattening]\label{def:unweighted-ver}
	Let $(G=(V_2,V_1,E),W)$ be an $(\alpha,\beta)$ sampler, such that $W$ is uniform on $V_1$, has irregularity at most $D$ on $V_2$, and for every $v\in V_2$, all the edges touching $v$ have the same weight.
	Then the flattening of $G$ is the unweighted bipartite graph $G'=(V_2',V_1,E')$, with vertex set $V_2'$ containing every vertex $v \in V_2$ repeated $M(v)$ times, where $M(v) \in [D]$ is the irregularity of $v$ in $G$. The edge set $E'$ contains $(v',w)$ if $v'$ is a duplicate of $v\in V_2$ and $(v,w)\in E$.
	
\end{definition}
\begin{claim}\label{claim:weighted_to_un}
	Let $(G=(V_2,V_1,E),W)$ be an $(\alpha,\beta)$ sampler with weights that satisfy the conditions of \prettyref{def:unweighted-ver}, then $G'=(V_2',V_1,E')$ the flattening of $(G,W)$ is also an $(\alpha,\beta)$ sampler.
\end{claim}
\begin{proof}
	Fix an arbitrary function $f:V_1\rightarrow[0,1]$, let $B\subset V_2$ be
	\[B = \sett{v\in V_2}{\abs{\E_{u\sim v}[f(u)]-\E_{u\sim V_1}[f(u)]}\geq\alpha}.\]
	The graph $(G,W)$ is an $(\alpha,\beta)$ sampler, so $\Pr_{v\sim V_2}[v\in B]\leq \beta$.
	
	Let $B'\subset V_2'$ be
	\[B' = \sett{v\in V_2'}{\abs{\E_{u\sim v}[f(u)]-\E_{u\sim V_1}[f(u)]}\geq\alpha}.\]
	The distribution $W$ is uniform over $V_1$, so $\E_{u\sim V_1}[f(u)]$ is the same in $G$ and in $G'$. For every $v'\in V_2'$ a copy of $v\in V_2$, the distribution over its neighbors is uniform, so $\E_{u\sim v'}[f(u)]=\E_{u\sim v}[f(u)]$. Therefore, the set $B'$ contains exactly all $v'\in V_2'$ which are copies of $v\in B$, and
	\[ \Pr_{v'\sim V_2'}[v'\in B'] =  \Pr_{v\sim V_2}[v\in B] \leq \beta.\]
\end{proof}

\subsection{Every sampler contains an induced expander}
\label{sec:G2}
The two-step walk of a weighted bipartite graph $(G' = (V_2,V_1,E'), W')$, is the weighted graph
$$(G=(V_2,E),W).$$
For every $T_1,T_2\in V_2$ and $S\in V_1$ such that $S$ is a common neighbor of $T_1,T_2$, we connect $T_1,T_2$ by an edge and label the edge $(T_1,T_2)_S$.

The weight of the edge $(T_1,T_2)_S$ corresponds to the probability of picking $S\sim V_1$ and then independent neighbors $T_1,T_2$ of $S$ in $G'$ according to the edge weights $W'$. More explicitly, $W((T_1,T_2)_S) = \mu_{G'}(S) \Pr_{T_1,T_2 \sim S} [ T_1=u_1 \wedge T_2=u_2 ]$.

Notice that the graph $G$ contains parallel edges and self loops.

\begin{claim}\label{claim:two_step}
	Choosing a random edge $e\in G$ according to the weights $W$ and a random vertex $T\in e$, has the same distribution as picking $T\sim V_2$ according to the vertex weights $W'$.
\end{claim}
\begin{proof}
	Fix a vertex $T\in V_2$, the probability to choose $T$ in $G$ equals:
	\[ \Pr_{(T_1,T_2)\sim E}[T_1 = T]  = \sum_{S\in V_1}\Pr_{S_1\sim V_1}[S_1=S]\Pr_{T_1\sim S}[T_1=T]. \]
	We omit the second endpoint of the edge because it's independent.
	
	By Bayes' rule,
	\begin{align*}
	\Pr_{(T_1,T_2)\sim E}[T_1 = T]  =& \sum_{S\in V_1}\Pr_{S_1\sim V_1}[S_1=S]\Pr_{(S_1,T_1)\sim E'}[T_1=T | S_1= S]\\
	=& \sum_{S\in V_1}\Pr_{S_1\sim V_1}[S_1=S]\frac{\Pr_{T_1\sim V_2}[T_1=T ]\Pr_{(S_1,T_1)\sim E'}[S_1=S | T_1=T]}{\Pr_{S_1\sim V_1}[S_1=S]}\\
	=&\Pr_{T_1\sim V_2}[T_1=T ]\sum_{S\in V_1}\Pr_{(S_1,T_1)\sim E'}[S_1=S | T_1=T]\\
	=&\Pr_{T_1\sim V_2}[T_1=T ].
	\end{align*}
\end{proof}

In \prettyref{sec:expanding subset} we show that the two-step graph of a sampler always contains an expander.

\begin{theorem}[Every sampler contains an induced expander]	\label{thm:induced-expander}
	Let
	$$(G' = (V_2,V_1,E'), W')$$
	be an $(\alpha,\beta)$ sampler for some $\alpha,\beta \in (0,1)$.
	Let $(G=(V_2,E),W)$ be the two-step walk graph of $G'$.
	Fix  $\eta  \in (0,1)$ such that
	$\eta > 10\sqrt{\max\set{\alpha,\beta}}$.
	
	Then for every $A \subseteq V_2$ with $\mu_{G}(A) = \eta$ there exists a set $B \subseteq A$ such that:
	\begin{itemize}
		\item $\mu_G(B) \ge \frac{\eta}{4}$.
		
		\item $\lambda_2(G_B) \le \frac{99}{100}$, where $G_B$ is the induced graph of $G$ on $B$ with the same edge weights.
	\end{itemize}
	Furthermore, given $A$, such a set $B$ can be found in time polynomial in $|V|$.
\end{theorem}

In order to prove the above theorem, we use a variant of the expander mixing lemma on the two-step walk graph $G$. Even though $G$ is not an expander, we show that large sets in it expand.
\begin{claim} \label{claim:sampler mixing lemma}
	Let $(G' = (V_2,V_1,E'), W')$ be an $(\alpha,\beta)$ sampler and let $(G=(V_2,E)$, $W)$ be the two-step walk on $G'$.
	
	Then for every $A,B\subset V_2$ satisfying $\mu_G(A) > \alpha,\mu_G(B) > \beta$,
	\begin{align*}
	&\Pr_{(T_1,T_2)\sim E}[T_1\in A,T_2\in B]\geq \Paren{\mu_G(A) - \alpha}\Paren{\mu_G(B) - \beta}\\
	&\Pr_{(T_1,T_2)\sim E}[T_1\in A,T_2\in B]\leq \mu_G(A)\Paren{\mu_G(B) + \beta} + \alpha.
	\end{align*}
\end{claim}
\begin{proof}
	Fix sets $A,B\subset V_2$ that satisfy the conditions of the claim. We define the function	$f:V_1\rightarrow[0,1]$ by,
	\[ \forall v\in V_1, \quad f(v) = \Pr_{T\sim v}[T\in A],\]
	In words, $f(v)$ is the probability of a random neighbor (according to the edge weights) of $v$ in $G'$ to be in $A$. From \prettyref{claim:two_step}, $\E_{v\sim V_1}[f(v)] = \Pr_{T \sim V_2}[T\in A] = \mu_G(A)$.
	
	For every $T\in V_2$, let $p_T$ be the probability of a random neighbour of $T$ in $G$ to be in $A$. Using $f$,
	$p_T = \Pr_{(T_1,T_2)\sim E}[T_2\in A | T_1 = T] = \E_{v\sim T}[f(v)]$. The last equality is because $G$ is a two-step random walk of $G'$.
	
	From the sampling properties of $G'$,
	\[ \Pr_{T\sim V_2}\Brac{\Abs{\E_{v\sim T}[f(v)] - \E_{v\sim V_1}[f(v)] } > \alpha}\leq \beta, \]
	Substituting $\E_{v\sim T}[f(v)]$ by $p_T$, we get that $\Pr_{T\sim V_2}\Brac{\Abs{p_T - \mu_G(A) } > \alpha}\leq \beta$.

	Let $R\subset V_2$ be the set
	\[ R = \sett{T\in V_2}{\Abs{p_T - \mu_G(A)} > \alpha }. \]
	From above, $\mu_G(R)\leq \beta$.
	For every $T_1\notin R$, $p_T\in \Brac{\mu_G(A)-\alpha,\mu_G(A)+\alpha}$.
	
	Therefore,
	\begin{align*}
	\Pr_{(T_1,T_2)\sim E}[T_1\in A,T_2\in B]
	\geq& \Pr_{T_2 \sim V_2} [T_2\in B\setminus  R]\Pr_{(T_1,T_2)\sim E} [T_1\in A | T_2\in B\setminus R]\\
	\geq&  \Paren{\mu_G(B) - \beta}\Paren{\mu_G(A) - \alpha}.
	\end{align*}
	\begin{align*}
	\Pr_{(T_1,T_2)\sim E}[T_1\in A,T_2\in B]
	\leq& \Pr_{T_2\sim V_2}[T_2\in R]\\ &+ \Pr_{T_2\sim V_2} [T_2\in B\setminus  R]\Pr_{(T_1,T_2)\sim E} [T_1\in A | T_2\in B\setminus R]\\
	\leq&  \beta + \mu_G(B)\Paren{\mu_G(A) + \alpha}.
	\end{align*}
\end{proof}

\subsection{Double samplers}\label{sec:d_samplers}
A double sampler is a two-layered graph with some additional properties. It is convenient to view it as an inclusion graph, defined as follows
\begin{definition}(Inclusion graph)
	An \emph{inclusion} graph $X=(V_2,V_1,V_0)$ with cardinalities $m_2 > m_1 >0$  is a tri-partite graph with vertices $V=V_2 \cup V_1 \cup V_0$, where $V_i \subseteq \binom{V_0}{m_i}$ for  $i=1,2$ and $(a,b) \in E$ iff $a \subseteq b$.
\end{definition}
Given a distribution $W_2$ on $V_2$, define a distribution $\Pi$ on $V_2\times V_1 \times V_0$ by sampling $v_2\in V_2$ according to $W_2$, then choosing a random neighbor $v_{1}\in V_{1}$ of $v_2$ and then a random neighbor $v_0\in V_0$ of $v_1$.

We denote by $\Pi_i$ the $i$-th coordinate of $\Pi$, and by $\Pi_{i,j}$ the distribution $\Pi$ restricted to layers $i,j$. Notice that $\Pi_2=W$.
For every $j\neq i\in \{0,1,2\}$ and every $v_i\in V_i$, we denote by $(\Pi_j | \Pi_i=v_i)$ the distribution of $\Pi_j$ conditioned on the $i$-th layer vertex equals $v_i$, explicitly $\Pr_{u_j\sim (\Pi_j | \Pi_i=v_i)}[v_j= T] = \Pr_{(u_2,u_1,u_0)\sim \Pi }[u_j = T | u_i = v_i]$.

By the way we constructed $\Pi$, it satisfies the following property: for every $T\in V_2$, the distribution $(\Pi_1|\Pi_2=T)$ is uniform, and for every $S\in V_1$, the distribution $(\Pi_0|\Pi_1=S)$ is also uniform. Note that other conditional distributions may not be uniform, for example, $(\Pi_2|\Pi_1=S)$ might not be uniform.

Given an inclusion graph $X=(V_2,V_1,V_0)$ and a distribution $\Pi$, we denote by $(X(V_{i+1},V_{i})$, $\Pi_{i+1,i})$  the weighted bipartite graph between $V_i,V_{i+1}$. For every $T \in V_2$, we define the weighted bipartite graph $$(X_{|T}=(U,T,E)\;,\;W_T)$$ where:
\begin{itemize}
	\item
	$U=\set{S \in V_1 ~|~ S \subseteq T}$, and recall $T$ is a set of elements from $V_0$,
	\item
	$(S,i) \in E$ for $S \in U$ and $i \in T$ iff $i \in S$, and,
	\item$W_T=(\Pi_{1,0}|\Pi_2=T)$.
\end{itemize}
The graph $X_{|T}$ is the subgraph of $X$ that contains all the subsets of $T$,
see \prettyref{fig:double_samp} for a graphic representation.

With this notation we define double samplers.

\begin{definition}[Double Sampler]
	\label{def:doublesampler}
	Let $X=(V_2,V_1,V_0)$ be an inclusion graph with distribution $\Pi$ on $V_2\times V_1\times V_0$ defined from $W$ on $V_2$ as above. We say $(X,W)$ is a $((\alpha,\beta),(\alpha_0,\beta_0))$ double sampler, if
	\begin{enumerate}
		\item\label{itm:samp1}
		$(X(V_2,V_1), \Pi_{2,1})$ is an  $(\alpha,\beta)$ sampler.	
		
		\item\label{itm:samp3}
		For every $T\in V_2$, $(X_{|T},W_T)$ is an $(\alpha_0,\beta_0)$ sampler.
	\end{enumerate}
	We say that $X$ has irregularity at most $D$ if
	\begin{itemize}
		\item $\Pi_2,\Pi_0$ are uniform, and for each $T\in V_2$, the graph $X_{|T}$ is bi-regular and $W_T$ is the uniform distribution.
		\item The distribution on $\Pi_1$ has irregularity at most $D$.
	\end{itemize}
	We say that $X$ is perfectly regular if it has irregularity $D=1$.
\end{definition}

Note that by the definition of inclusion graph over a ground set, the bipartite graphs $X(V_2,V_0)$, $X(V_1,V_0)$ and $X_{|T}$ are always left-regular. A vertex $S\in V_1$ has exactly $m_1$ neighbors in $V_0$, which are the $m_1$ elements the set $S$ contains, and the same for $T\in V_2$.

The definition of a double sampler implies that the graph $X$ has more sampling properties, proven in the claim below.
\begin{claim}
	\label{claim:sampler}
	Let $X$ be a double sampler as in \prettyref{def:doublesampler}, then:
	\begin{enumerate}
		\item
		\label{item:20sampler}
		$(X(V_2,V_0),\Pi_{2,0})$ is an $(\alpha,\beta)$ sampler, and,
		
		\item
		\label{item:10sampler}
		$(X(V_1,V_0),\Pi_{1,0})$ is an  $(\alpha+\alpha_0,\beta+\beta_0)$ sampler.
	\end{enumerate}
	\item
\end{claim}

\begin{proof}\mbox{ }
	
	\begin{itemize}
		\item
		For \prettyref{item:20sampler},	fix an arbitrary function $f:V_0\rightarrow[0,1]$, let $f':V_1\rightarrow[0,1]$ be $f'(S)=\E_{x\sim (\Pi_0|\Pi_1=S)}[f(x)]$. From the definition of $\Pi$, the expectation of $f$ over $\Pi_0$ and $f'$ over $\Pi_1$ is the same $\E_{S\sim\Pi_1}[f'(S)]=\E_{x\sim\Pi_0}[f(x)]$.
		
		The graph $(V_2,V_1)$ is a $(\alpha,\beta)$ sampler,
		\begin{align}\label{eq:21sampler}
		\Pr_{T \sim \Pi_2} \left[  \Abs{\E_{S \sim (\Pi_1|\Pi_2=T)}[ f'(S)]-\E_{S \sim \Pi_1}[f'(S)]} \ge \alpha \right] \le \beta.
		\end{align}
		
		The joint distribution $\Pi=\Pi_0,\Pi_1,\Pi_2$ is defined so that for every $T\in V_2$, if we first choose $S\sim\Pi_1|\Pi_2=T$ and then $x\sim\Pi_0|\Pi_1=S$, it is the same as when we choose $x\sim\Pi_0|\Pi_2=T$. Therefore, for every $T\in V_2$, $\E_{x \sim (\Pi_0|\Pi_2=T)}[f(x)]=\E_{S \sim (\Pi_1|\Pi_2=T)}[f'(S)]$.
		
		Substituting the expectations in \prettyref{eq:21sampler} we get
		\begin{align} \label{eq:20sampler}
		\Pr_{T \sim \Pi_2} \left[  \Abs{\E_{x \sim (\Pi_0|\Pi_2=T)}[ f(x)]-\E_{x \sim \Pi_0}[f(x)]} \ge \alpha \right] \le \beta,
		\end{align}
		which finished the proof.
		
		\item
		For  \prettyref{item:10sampler}, fix an arbitrary  $f:V_0\rightarrow[0,1]$, and let $f':V_1\rightarrow[0,1]$ be as in the first item, $f'(S)=\E_{x\sim \Pi_0|\Pi_1=S}[f(x)]$. Again, $\E_{S\sim\Pi_1}[f'(S)]=\E_{x\sim\Pi_0}[f(x)]$ and equations \prettyref{eq:21sampler},\prettyref{eq:20sampler} holds here as well.
		
		For every $T\in V_2$, let $\mu_T$ be the expected value of $f$ on $T$, i.e. $\mu_T = \E_{x \sim (\Pi_0|\Pi_2=T)}[f(x)]$. Let $A$ be the set of vertices $T\in V_2$ in which $\mu_T$ is close to the  expectation of $f$ on $V_0$,
		\[A = \sett{T\in V_2}{\Abs{\mu_T - \E_{x \sim \Pi_0}[f(x)]} \leq\alpha }.\]
		From \prettyref{eq:20sampler}, $\Pr_{T\sim\Pi_2}[T\in A]\geq 1-\beta$.
		
		For every $T$, $\mu_T = \E_{S\sim(\Pi_1 | \Pi_2=T)}[f'(S)]$, see explanation in the first item. The bipartite graph $X_{|T}$ is an $(\alpha_0,\beta_0)$ sampler, so
		\[ \Pr_{S \sim (\Pi_1|\Pi_2=T)}\Brac{\Abs{f'(S) - \mu_T} \geq\alpha_0}\leq\beta_0. \]
		By the triangle inequality,
		\[ \Abs{f'(S)-\E_{x \sim \Pi_0}[f(x)]} \leq \Abs{f'(S)-\mu_T}+\Abs{\mu_T-\E_{x \sim \Pi_0}[f(x)]} . \]
		For $T\in A$, $\Abs{\mu_T-\E_{x \sim \Pi_0}[f(x)]}\leq \alpha$. Therefore for $T\in A$
		\begin{align}\label{eq:exp_on_A}
		\Pr_{S \sim (\Pi_1|\Pi_2=T)}\Brac{\Abs{f'(S)-\E_{x \sim \Pi_0}[f(x)]} \ge \alpha+\alpha_0}\leq \beta_0.
		\end{align}

		We finish the proof by using the fact that the probability of $T\sim\Pi_2$ to be in $A$ is at least $1-\beta$. Let $I$ be the event in which $\Abs{f'(S)-\E_{v \sim \Pi_0}[f(v)]} \ge \alpha+\alpha_0$. Using this notation,
		\begin{align*}
		\Pr_{S \sim \Pi_1}\Brac{I} \leq & \Pr_{T\sim\Pi_2}[T\notin A] +\Pr_{T\sim\Pi_2}[I | T\in A]\\
		\leq&\beta + \beta_0 \tag{by \prettyref{eq:20sampler} and \prettyref{eq:exp_on_A}}.
		\end{align*}
		
	\end{itemize}	
\end{proof}

In \prettyref{sec:doublesampler} we prove that bounded-degree double samplers can be constructed from high dimensional expanders. While the distributions in the explicit construction are not uniform, they are not too far off.  $\Pi_1$ has bounded irregularity and $\Pi_2,\Pi_0$ are uniform. The irregularity stems from the irregularity of the high dimensional expanders. Very recent work \cite{FriedgutI2020} has managed to construct {\em regular} high dimensional expanders, which leads to perfectly regular double samplers. These were not available when an earlier version of this manuscript came out.

\begin{theorem}\label{thm:doublesampler}
	For every $\alpha,\beta,\alpha_0,\beta_0>0$ there exist constants  $m_1,m_2,D\in \mathbb{N}$, $m_1,m_2 = \poly(\frac{1}{\alpha\beta\alpha_0\beta_0}), D = \exp(\poly(\frac{1}{\alpha\beta\alpha_0\beta_0}))$, such that there is a family of explicitly constructible double samplers $(X_n,W_n)$ for infinitely many $n\in \mathbb{N}$ satisfying
	\begin{itemize}
		\item  $X_n=(V_2,V_1,V_0)$ is an inclusion graph, where $|V_0|=n$, $V_i \subseteq \binom{V_0}{m_i}$  for $i=1,2$.
		\item $X_n$ is an $((\alpha,\beta),(\alpha_0,\beta_0))$  double sampler.
		\item $|V_1|,|V_2| \le D \cdot n$.
		\item The distributions $\Pi_0,\Pi_2$ are uniform and the distribution $\Pi_1$ has irregularity at most $D$.
		\item For each $m\in \mathbb{N}$ there is some $n\in [m,Dm]$ such that the complex $X_n$ on $n$ vertices is constructible in time $poly(n)$.
	\end{itemize}
\end{theorem}

\subsection{UG constraint graphs}

\begin{definition}(UG constraint graph)
	Let $(G=(V,E),W)$ be a weighted graph. We say $((G,W),\set{\pi_e}_{e \in E})$ is a UG constraint graph with $\ell$ labels if
	$\pi_e:[\ell]\rightarrow[\ell]$ is a permutation.
	An assignment for $G$ is a function $a:V\rightarrow[\ell]$. We say the assignment $a$ satisfies an edge $e=(u,v)$ if $\pi_e(a(u)) = a(v)$.  The \emph{value} of an assignment is the fraction of satisfied edges. We say $((G,W),\set{\pi_e}_{e \in E})$ is $p$-satisfiable if there exists an assignment with value at least $p$.
\end{definition}
The graph $G$ is undirected, and since the constraints are unique it doesn't matter if $\pi_e$ is represented from $u$ to $v$ or vice versa.
Arora, Khot, Kolla, Steurer, Tulsiani and Vishnoi \cite{AroraKKSTV2008} showed how to solve unique games instances on expander graphs in polynomial time. This result was improved by Makarychev and Makarychev \cite{MakarychevM2010}, who proved,
\begin{theorem}[{\cite[Theorem~10]{MakarychevM2010}}]
	\label{thm:pre MM}
	Let $G$ be a regular graph with edge expansion $h_G$. There exist positive absolute constants $c$ and $C$ and a polynomial time approximation algorithm that
	given a $1-\beta$ satisfiable instance of UG on $G$ with $\frac{\beta}{1-\lambda_2(G)}\leq c$, the algorithm finds a solution of value $1-C\frac{\beta}{h_G}$.
\end{theorem}

We need a version of this theorem with two modifications:

\begin{itemize}
	\item
	The theorem, as stated, refers to unweighted regular graphs. We need the same results for non-regular weighted graphs.
	
	\item
	The theorem finds one assignment with high value. However, we need to get an approximation to \emph{all} assignments with high value.
\end{itemize}

In \prettyref{app:weighted ug} we go over the algorithm in \cite{AroraKKSTV2008,MakarychevM2010}  and show that the same result holds for weighted non-regular graphs. In \prettyref{sec:list ug} we show how to output a list that contains an approximation to all assignments with high value. To do that we run the algorithm
several times, each time peeling off the solution that is found. We prove:

\begin{theorem}	\label{thm:pre UG}
	Let $(G=(V,E),W)$ be a weighted undirected graph with
	$\lambda_2(G)\leq\frac{99}{100}$.
	Let $\set{\pi_e}_{e \in E}$ be unique constraints over the edges of $G$, with $\ell$ labels.
	
	Then there is an absolute constant $c>1$ and a polynomial time algorithm that outputs a list of assignments $L=\set{a^{(1)},\dots, a^{(t)}}$ with $a^{(i)}:V \to [\ell]$. The list satisfies that for every assignment $a:V\rightarrow[\ell]$ that satisfies $1-\eta$ of the constraints for $\eta <c^{-\ell-1}$,
	there exists $a^{(i)}\in L$ that satisfies $\Pr_{v\sim W}[a(v) = a^{(i)}(v)]\geq 1-\eta c^{\ell}$.
\end{theorem}

\subsection{The Johnson bound}
\label{sec:johnson}
Johnson's bound shows that any code with a good distance is also a good list-decodable code. The bound has several versions, and the quantitative bounds are different over large and small alphabets. The bounds are usually stated for codes, but, essentially address the following geometric problem: How many vectors can be all close to the same vector, and far away from each other. As such, the bounds are also useful for approximate error correcting codes, see \prettyref{sec:approximate_ecc}. In this section we repeat the Johnson bound (with this geometric interpretation in mind) and give the proofs for completeness. We follow the exposition in \cite{Sudan2001,GuruswamiRS}.

\subsubsection{Large alphabet}

\begin{theorem}
	\label{thm:Johnson_large_a}
	Let $\Sigma$ be a finite set of cardinality $q$, let $z \in \Sigma^n$ and $L \subseteq \Sigma^n$. If:
	\begin{itemize}
		\item
		For every $x \in L$, $\dist(x,z) \le 1-\gamma$, and,
		
		\item
		For every distinct $x,x' \in L$, $\dist(x,x') \ge 1-\beta$,
	\end{itemize}
	where $\gamma^2 > \beta$, then, $|L| \le \frac{\gamma - \beta}{\gamma^2-\beta}$.
\end{theorem}

\begin{proof}
	Let $z \in \Sigma^n$ and let $L = x_1,\ldots, x_\ell$ be a list satisfying the theorem requirements, $\abs{L}=\ell$. We prove the Johnson bound by double counting. The expression we bound is the probability that two different list elements agree on a random coordinate,
	\[ p = \Pr_{x\neq x'\in L, i\in [n]}[x_i = x'_i]. \]
	
	By the list requirement, $x,x'\in L$ satisfies $\dist(x,x') \ge 1-\beta$, therefore $p \le \beta$.
	Also,
	
	\begin{eqnarray*}
		p ~=~ \Pr_{x\neq x'\in L, i\in [n]}[x_i = x'_i] & = & \E_{i\in[n]}[\Pr_{x\neq x'\in L}[x_i= x'_i]]\\
		&  \ge &  \E_{i\in[n]}[\Pr_{x\neq x'\in L}[x_i=x'_i=z_i]]
	\end{eqnarray*}
	
	For every coordinate $i\in[n]$, let $\ell_i$ be, $\ell_i = \abs{\sett{x\in L}{x_i=z_i}}.$
	Using this notation,
	\[ \Pr_{x\neq x'\in L}[x_i=x'_i=z_i] = \frac{\ell_i}{\ell}\cdot\frac{\ell_i-1}{\ell-1}. \]
	
	The list requirements imply that $\E_{i\in[n]}[\ell_i]\geq \gamma\ell$, which gives the bound on $p$:
	\begin{align*}
	p &\geq \E_{i\in[n]}\Brac{\frac{\ell_i(\ell_i-1)}{\ell(\ell-1)}} \geq \frac{1}{\ell(\ell-1)}\Paren{\E_{i\in[n]}[\ell_i^2]-\E_{i\in[n]}[\ell_i]}\geq \frac{\gamma^2 \ell -\gamma}{\ell-1}.
	\end{align*}
	Hence $\beta(\ell-1) \ge \gamma^2\ell-\gamma$ and  $\ell \leq \frac{\gamma-\beta}{\gamma^2 - \beta}$.
\end{proof}

\subsubsection{Binary alphabet}

For Binary alphabet and distance $\gamma = \frac{1-\alpha}2$ where $\alpha\approx 0$, the condition $\alpha^2>\beta$ is restrictive, as it forces
$\beta$ to be about $\frac{1}{4}$ or below, which corresponds to the distance $1-\beta$ being $\frac{3}{4}$ which is vacuous for Binary codes. We now state a version useful for Binary alphabets.

\begin{lemma}(See \cite[Lem 7.1]{Sudan2001})
	\label{lem:Johnson-core}
	Fix $\Delta, \Gamma>0$. Suppose $x_1,\ldots,x_m \in \R^n$ are such that
	\begin{itemize}
		\item
		$\norm{x_i}^2 \le \Gamma$, and,
		
		\item
		$\la x_i,x_j \ra \le -\Delta$ for every $i \neq j$.
	\end{itemize}
	Then $m \le 1+\frac{\Gamma}{\Delta}$.
\end{lemma}

\begin{proof}
	Set $z=\sum x_i$. Then $\la z,z \ra \ge 0$ and
	\begin{eqnarray*}
		\la z ,z \ra & =& \sum_i \la x_i,x_i \ra +\sum_{i \neq j} \la x_i,x_j \ra \le m \Gamma +m(m-1)(-\Delta).
	\end{eqnarray*}
	Together this implies $\Gamma-(m-1)\Delta \ge 0$ and the bound.
\end{proof}

\begin{theorem}(Following \cite{Sudan2001})
	\label{thm:Johnson_small_a}
	Let $\Sigma=\set{0,1}$, $z \in \Sigma^n$, $L \subseteq \Sigma^n$. If:
	\begin{itemize}
		\item
		For every $w \in L$, $\dist(w,z) \le \frac{1-\alpha}{2}$, and,
		
		\item
		For every different $x,x' \in L$, $\dist(x,x') \ge \frac{1-\beta}{2}$,
	\end{itemize}
	where $\alpha^2 > \beta$. Then, $|L| \le 1+\frac{4}{\alpha^2-\beta}$.
\end{theorem}

\begin{proof}
	We first define an embedding $E: \set{0,1}^n \to \R^{n}$ by letting $E(b_1,\ldots,b_n)=\frac{1}{\sqrt{n}}((-1)^{b_1}$ $,\ldots,(-1)^{b_n})$. Notice that if $\sigma_1,\sigma_2 \in \set{0,1}^n$ then  $\la E(\sigma_1) ,E(\sigma_2) \ra = 1-2\dist(\sigma_1,\sigma_2)$ and in particular $\norm{E(\sigma)}=1$.
	Suppose $L=\set{\sigma_1,\ldots,\sigma_m}$. Denote $x_i=E(\sigma_i)$, $y=E(z)$ and $z_i=x_i -\alpha y$.
	Then:
	
	\begin{eqnarray*}
		\la z_i,z_i \ra & = & \la x_i-\alpha y,x_i - \alpha y \ra = 1-2\alpha \la x_i,y \ra +\alpha^2 \le 1+2\alpha+\alpha^2 =(1+\alpha)^2 \le 4, \mbox{ and, } \\
		\la z_i,z_j \ra & = & \la x_i-\alpha y,x_j - \alpha y \ra = \la x_i,x_j \ra-\alpha[ \la x_i,y \ra+ \la x_j,y \ra] +\alpha^2 \\
		&  = & 1-2\dist(x_i,x_j)-\alpha(1-2\dist(x_i,y)+1-2\dist(x_j,y)]+\alpha^2 \\
		& \le & \beta-2\alpha^2+\alpha^2=\beta-\alpha^2.
	\end{eqnarray*}
	It follows by  \prettyref{lem:Johnson-core} that $m \le 1+\frac{4}{\alpha^2-\beta}$.
\end{proof}

\section{Approximate Error Correcting Codes}\label{sec:approximate_ecc}
In this section we revisit the definition of approximate error correcting codes (ECCs), and the ABNNR encoding \cite{AlonBNNR1992} in both direct product and direct sum form. We also relate approximate ECC to distance amplification.

\begin{definition}(Approximate ECC)  \cite[Def 1.5]{ImpagliazzoJKW2010}
	$E: \Sigma_0^n \to \Sigma_1^{m}$ is an $(r,\eta,\ell)$ approximate-list-decodable error correcting code if for every $z \in \Sigma_1^{m}$ there is a list $L \subseteq \Sigma_0^n$ of cardinality at most $\ell$ such that for every codeword $E(x)$ with $\dist(E(x),z) \le \eta$ we have $\dist(x,L) \le r$.
\end{definition}

There are a few important differences between an approximate-list-decodable ECC and a list-decodable ECC. A code $C: \Sigma_0^n \to \Sigma_1^{m}$ is $(\eta,\ell)$ \emph{list-decodable}, if for every $z\in \Sigma_1^m$, $\abs{L}= \abs{\sett{x\in \Sigma_0^m}{\dist(C(x),z)\leq\eta}}\leq\ell$. In an approximate-list-decodable ECC, we relax the condition and only require the list $L$ to contain an \emph{approximation} for $x$ such that $\dist(E(x),z)\leq\eta$. This requirement is much weaker, and an approximate-list-decodable ECC might not have a noticeable minimal distance.

Another difference between the codes is the uniqueness of the list. In list-decoding, for every code $C$ and input $z$ there is a \emph{unique} list $L$ which satisfies the requirements, $L=\sett{x\in \Sigma_0^m}{\dist(C(x),z)\leq\eta}$. In addition, given an efficient encoding algorithm for $C$ it is easy to check if $x$ is in $L$, by checking if $\dist(C(x),z)\leq\eta$.
In approximate-list-decoding, for every encoding $E$ and input $z$ there can be many possible lists which satisfy the requirements, as each $x$ satisfying $\dist(E(x),z)\leq\eta$ has many strings that $r$-approximate it. Furthermore, given a string $x'\in \Sigma_0^n$, it is not clear how to check efficiently if $x'$ is an $r$-approximation to some $x$ such that $\dist(E(x),z)\leq\eta$, as going over all strings which are $r$-close to $x'$ takes possibly $\exp(n)$ time.

From the above paragraph we understand why the list size in \prettyref{thm:main}, which is an approximate-list-decoding algorithm, is larger than the list size in \prettyref{cor:list-dec}. In the corollary there is a list-decoding algorithm, where it is possible to check and prune the list, see a proof of the corollary in \prettyref{sec:list-dec-c}.
\subsection{The ABNNR direct-product construction}

We restate the ABNNR encoding \cite{AlonBNNR1992}. In the original paper, it is defined as an amplification step applied on an underlying code. Here we state it as an approximate error correcting code.

\begin{definition}(ABNNR encoding\cite{AlonBNNR1992})\label{def:ABNNR}
	Let $G=(V_1,V_0,E)$ be a $d$ left-regular bipartite graph, and let $\Sigma_0$ be a constant size alphabet. The ABBNR encoding
	$$E_{G,\Sigma_0}:\Sigma_0^{\abs{V_0}}\rightarrow\Sigma_1^{\abs{V_1}}$$
	for $\Sigma_1 = \Sigma_0^d$ is defined as follows.
	For every $v\in V_1$ let $S_v\subset V_0$ be the neighbours of $v$. For every $x\in \Sigma_0^{\abs{V_0}}$ the encoding $E_G(x)$ is defined by
	\[ \forall v\in V_1, \quad E_G(x)_v = x_{|S_v}. \]
\end{definition}
In order to simplify the notations, when the alphabet is clear we omit $\Sigma_0$ and use $E_G$.
To clarify, $d$ left-regular means that each $v\in V_1$ has exactly $d$ neighbors.

The key property of this construction is distance amplification.
\begin{claim}\label{claim:sampler_dis_enc}
	For every $x,x'\in \Sigma^{\abs{V_0}}$ such that $\dist(x,x') > \alpha$, \[ \dist(E_G(x),E_G(x'))\geq 1-\beta.\]
\end{claim}

\begin{proof}
	Let $f: V_0 \rightarrow \{0,1\}$ equal $1$ for all $u\in V_1$ such that $x_{|u} \neq x'_{|u}$, and $0$ else. As $\dist(x,x') > \alpha$ we have $\E_{u\in V_0}[f(u)] > \alpha$.
	By the sampler property of $G$, all except a $\beta$ fraction of the vertices in $V_1$ satisfy
	\begin{eqnarray*}
		\E_{u\sim v}[f(u)] & \ge &  \E_{u\in V_0}[f(u)]- \alpha ~>~0.
	\end{eqnarray*}
	Thus, for all these vertices $v$, $E_G(x)_v \ne E_G(x')_v$ which implies $\dist(E_G(x),E_G(x')) \ge 1-\beta$.
\end{proof}

Using this, we show that $E_G$ is an approximate-list-decodable code.
\begin{lemma}\label{lem:ABNNR}
	Let $G=(V_1,V_0,E)$ be a $d$ left-regular $(\alpha,\beta)$ sampler. Let  $\gamma > \sqrt{\beta}$.
	Then $E_G$ is $(\alpha,1-\gamma,
	\ell=\frac{\gamma-\beta}{\gamma^2-\beta})$ approximate-list-decodable ECC.
\end{lemma}

\begin{proof}
	Fix an arbitrary $z \in \Sigma_1^{\abs{V_1}}$, and define
	
	\begin{eqnarray*}
		L_0 &=& \set{x \in \Sigma_0^{\abs{V_0}} ~|~ \dist(E_G(x),z) \leq 1-\gamma}.
	\end{eqnarray*}
	
	The list $L_0$ contains all $x$ such that $E_G(x)$ has $\gamma$ agreement with $z$, but it is possibly very large. We reduce the size of $L_0$ by removing elements which are too close to each other.
	Let $L_\alpha \subset L_0$ be a maximal subset of $L_0$, such that different $x,x'\in L_\alpha$ are at distance more than $\alpha$.
	By definition, for each $x \in L_0$ there exists $x' \in L_\alpha$ such that $\dist(x,x') \le \alpha$, or else it is possible to add $x$ into $L_\alpha$ which contradicts $L_\alpha$ being maximal. It remains to bound $\ell=\abs{L_\alpha}$.
	We bound the cardinality of $L_\alpha$ by noting that:
	
	\begin{itemize}
		\item
		For different $x,x'\in L_\alpha$, $\dist(E_G(x),E_G(x'))\geq 1-\beta$ (by \prettyref{claim:sampler_dis_enc}), and,
		
		\item 
		For every $x \in L_\alpha$, $\dist(E_G(x),z) \le 1-\gamma$ (by definition of $L_0$).
	\end{itemize}
	Hence we get $\ell$ vectors that are all close to one vector $z$, but are pairwise far apart. By a variant of the Johnson bound, \prettyref{thm:Johnson_large_a}, we have $\ell \le \frac{\gamma-\beta}{\gamma^2-\beta}$. We remark that for this proof it is enough to take a disperser instead of a sampler.
\end{proof}

\subsection{A direct sum construction over the binary alphabet}
In this section we take a small detour to analyse the direct sum encoding, which is very related to the direct product discussed in the rest of this paper. We show that it is list-decodable.

We focus on binary alphabet, i.e. $\sigma \in \set{0,1}^n$, we define $	\bias(\sigma)=\Abs{\sum_{i=1}^n (-1)^{\sigma_i}}$. In this case  $\dist(\sigma,\sigma')=\delta$ implies $\bias(\sigma+\sigma')=\abs{1-2\delta}$.

\begin{definition}(A direct sum encoding)
	Let $G=(V_1,V_0,E)$ be a $d$ left-regular bipartite graph. The direct sum encoding
	$$E^\oplus:\set{0,1}^{\abs{V_0}}\rightarrow\set{0,1}^{\abs{V_1}}$$
	is defined as follows:
	For every $v\in V_1$ let $S_v\subset V_0$ be the neighbours of $v$. For every $x \in \set{0,1}^{\abs{V_0}}$, the encoding $E^\oplus(x)$ is defined by
	$$\forall v\in V_1, \quad E^\oplus(x)_v = \sum_{u \in S_v} x_u \bmod 2.$$
\end{definition}

\begin{definition}(Parity sampler)
	Let $G=(V_1,V_0,E)$ be a bipartite graph. Given $f:V_0 \to \set{0,1}$ we define a function $G(f):V_1 \to \set{0,1}$ by
	\begin{eqnarray*}
		\forall v_1\in V_1,\qquad G(f) (v_1) & = & \sum_{v \in S_{v_1}} f(v) \bmod 2.
	\end{eqnarray*}
	We say $G$ is an $(\alpha,\beta)$ parity sampler if for every function $f:V_0 \to \set{0,1}$ with $\bias(f) \le \alpha$ , it holds that $\bias(G(f)) \le \beta$.
\end{definition}

\begin{lemma}\label{lem:ABNNR-Binary}
	Let $G=(V_1,V_0,E)$ be a $d$ left-regular $(\alpha,\beta)$ parity sampler. Fix
	$\gamma > \sqrt{\beta}$.
	Then the direct sum encoding on $G$ is an $(\alpha,\frac{1-\gamma}{2},
	\ell=1+\frac{4}{\gamma^2-\beta})$ approximate-list-decodable ECC.
\end{lemma}

\begin{proof}
	Let
	$E^\oplus$ denote the direct sum encoding on $G$. Fix any $z \in \set{0,1}^{\abs{V_1}}$, and define
	
	\begin{eqnarray*}
		L_0 &=& \set{x \in \set{0,1}^{\abs{V_0}} ~|~ \dist(E^\oplus(x),z) \leq \frac{1-\gamma}{2}}.
	\end{eqnarray*}
	
	The list $L_0$ contains all $x$ such that $E^\oplus(x)$ has $\frac{1+\gamma}{2}$ agreement with $z$, but it might be too large a set. We reduce the size of $L_0$ by removing elements which are too close to each other.
	Let $L_\alpha \subset L_0$ be a maximal subset of $L_0$, such that different $x,x'\in L_\alpha$ have $\bias(x +x') \le \alpha$.
	By definition, for each $x \in L_0$ there exists $x' \in L_\alpha$ such that $\bias(x,x') > \alpha$, or else it is possible to add $x$ into $L_\alpha$ which contradicts $L_\alpha$ being maximal. It remains to bound $\ell=\abs{L_\alpha}$.
	We bound the cardinality of $L_\alpha$ by noting that:
	
	\begin{itemize}
		\item 
		For different $x,x'\in L_\alpha$, we have $\bias(x+x') \le \alpha$, hence by the parity sampler property $\bias(E^\oplus(x+x'))=\bias(E^\oplus(x)+E^\oplus(x')) \le \beta$. In particular
		$\dist(E^\oplus(x),E^\oplus(x')) \ge \frac{1-\beta}{2}$. Also,
		
		\item 
		By definition of $L_0$, for every $x \in L_\alpha$, $\dist(E^\oplus(x),z) \le \frac{1-\gamma}{2}$.
	\end{itemize}
	Hence if $\ell=|L_\alpha|$, we get $\ell$ vectors that are all close to one vector $z$, but are all far apart from each other. By the Johnson bound, \prettyref{thm:Johnson_small_a}, we have $\ell \le 1+\frac{4}{\gamma^2\beta}$.
\end{proof}

\subsection{Well-separated approximate ECC}
\label{sec:well-sep}
In some cases when list-decoding, we also want the elements in the output list $L$ to have at least some distance between them, we define such a list as $r$-separated.
\begin{definition}
	Let $L \subseteq \Sigma^n$ and $r \in [0,1]$. We say $L$ is $r$-separated if for every $\sigma_1,\sigma_2 \in L$, $\dist(\sigma_1,\sigma_2) \ge r$.
\end{definition}
Let $G=(V_1,V_0,E)$ be a $d$ left-regular $(\alpha,\beta)$-sampler, and let $E_G:\Sigma_0^{\abs{V_0}}\rightarrow\Sigma_1^{\abs{V_1}}$ be as in \prettyref{def:ABNNR}. Fix $\gamma > \sqrt{\beta}$, we show an (inefficient) algorithm which $(r,1-\gamma,\ell)$ approximate-list-decodes $E_G$ and outputs a $5r$-separated list, by dynamically adjusting $r$ for every input.

\begin{algorithm}[Well-Separated List Decoding Algorithm]
	\label{alg:well-sep-alg}
	
	The algorithm has a decoding parameter $1-\gamma$.
	The input is a word $z\in\Sigma_1^{\abs{V_1}}$, the output is a list $L\subset\Sigma_0^{\abs{V_0}}$ and a radius $r$.
	\begin{itemize}
		\item Set $i=0, \tau_0=\alpha, L_0=\sett{x\in \Sigma_0^{\abs{V_0}}}{\dist(E_G(x),z)\leq 1-\gamma}$.
		\item At stage $i$, if $L_i$ is $10\tau_i$-separated, output $L_i,r = 2\tau_i$.
		
		Otherwise, say $v_1,v_2 \in L$ are independent if $\dist(v_1,v_2) \ge 10\tau_i$. Set $L_{i+1}$ to be a maximal independent set in $L_i$. Set $\tau_{i+1}=10\tau_i$, $i=i+1$ and repeat the loop.
	\end{itemize}
\end{algorithm}
\begin{lemma}\label{lem:well-sep}
	For every $z\in\Sigma_1^{\abs{V_1}}$ the algorithm above  outputs a list $L$ of size at most $\ell = \lfloor \frac{\gamma-\beta}{\gamma^2-\beta} \rfloor$  and a radius $r=2\alpha 10^i$ for some $i\in \{0,\dots\ell\}$  such that
	\begin{itemize}
		\item For every $x\in \Sigma_0^{\abs{V_0}}$ such that $\dist(E_G(x),z)\leq 1-\gamma$, $\dist(x,L)\leq r$.
		\item $L$ is $5r$-separated.
	\end{itemize}
\end{lemma}

\begin{proof}
	We start by bounding the size of the lists $L_i$ produced by the algorithm.
	\begin{claim}
		For every $\tau>\alpha$, a list $L_i$ which is $\tau$-separated satisfies $\abs{L_i}\leq\ell$.
	\end{claim}
	\begin{proof}
		Every $x,x'\in L_i$ satisfy $\dist(x,x')\geq \tau$. Using the sampler properties of $G$, \prettyref{claim:sampler_dis_enc}, $\dist(E_G(x),E_G(x'))\geq1-\beta$.
		From the definition of $L_0$, every $x\in L_i\subset L_0$ also satisfies $\dist(E_G(x),z)\leq 1-\gamma$. By the Johnson bound, \prettyref{thm:Johnson_large_a}, $\abs{L_i}\leq \ell$.
	\end{proof}
	The algorithm outputs a list $L$ which is at least $10\alpha$-separated, so the claim above proves $\abs{L}\leq\ell$.

	We now bound the number of steps the algorithm performs. Denote by $t$ the index $i$ with which we quit (if we quit), we prove that $t \leq \ell$.
	By the algorithm definition, for every $i$, $L_{i+1}\subsetneq L_i$, as in the case where $L_{i+1}=L_i$ the algorithm stops. The list $L_1$ is $10\alpha$-separated, so from the claim above $\abs{L_1}\leq\ell$, together we get that $t\leq \ell$.
	
	Next we prove the covering property, let $x$ be such that $\dist(E_G(x),z)\leq 1-\gamma$, we show that $\dist(x,L)\leq r$. By the definition of $L_0$, $x\in L_0$. For every $i < t$ and every $x' \in L_i$, $\dist(x',L_{i+1}) \le 10\tau_i=\tau_{i+1}$, for otherwise $x'$ can be added to the independent set $L_{i+1}$ contradicting its maximality. Hence, for $x \in L_0$:
	\begin{eqnarray*}
		\dist(x,L_t)  & \le & \sum_{i=0}^t \tau_{i}=\alpha \sum_{i=0}^t 10^i ~=~  \frac{10^{t+1}-1}{10-1}\alpha ~\le~ \frac{10}{9}\tau_t < r.
	\end{eqnarray*}
	
	Where the last inequality holds since $r=2\tau_t$. The algorithm always outputs a list that is $10\tau_t = 5r$-separated, which finishes the proof.
\end{proof}
The algorithm runs in time $\exp(\abs{V_0}\log \abs{\Sigma_0})\poly(\ell)$, because it goes over all $x\in \Sigma_0^{\abs{V_0}}$ to create the initial list $L_0$, and the loop runs at most $\ell$ times.

\subsection{Approximate-list-decoding implies distance amplification}\label{sec:dist_amp}
In this section we show that composing an approximate-list-decodable code with a an error correcting code results in a list-decodable error correcting code.

\begin{claim}\label{claim:comb-list-dec}
	Let $E: \Sigma_0^n \to \Sigma_1^{m}$ be an $(r,\eta,\ell)$ approximate-list-decodable code, and let $C\subset \Sigma_0^n$ be an error correcting code with distance $2r$. Then the code:
	\[ E(C) = \sett{E(x)}{x\in C}, \]
	is $(\eta,\ell)$ list-decodable.
\end{claim}

\begin{proof}
	Let $z\in \Sigma_1^{m}$ and let $L=List_E(z)$ be the cardinality $\ell$ list guaranteed by the approximate-list-decoding property of $E$.
	Let $L'$ be the list
	\[ L' = \sett{x\in C}{\exists y\in L, \dist(x,y)\leq r}. \]
	We claim that $L'$ is small and contains all codewords close to $z$.

	Suppose $x\in C$ is such that $\dist(E(x),z)\leq \eta$. Then by the definition of approximate-list-decoding, $\exists y\in L$ such that $\dist(x,y)\leq r$, so by definition $x\in L'$.
	
	For every $y\in L$, there exist at most a single $x\in C$ which is $r$-close to $y$, because the distance of $C$ is $2r$. Therefore, $\abs{L'}\leq\ell$.	
\end{proof}

It is easy to see that if $E$ has an approximate-list-decoding algorithm and $C$ has a unique-decoding algorithm, then $E(C)$ has a list-decoding algorithm. To decode $E(C)$ we simply run the decoding algorithm of $E$ getting $L$, then run the decoding algorithm of $C$ on all elements in $L$ and output the result.

\section{The Code and its Approximate List-Decoding Algorithm}
\label{sec:code}
In this section we describe our code and an approximate-list-decoding algorithm for it. The code is the ABNNR encoding over the first two layers of a double sampler, defined as follows.

Let $(X=(V_2,V_1,V_0=[n]),\Pi)$ be a double sampler.  It is recommended to first focus on the case where the $X$ is perfectly regular, namely every pair of layers give rise to a bi-regular graph. In the slightly more general case, $X$ has irregularity at most $D$, i.e. $\Pi_0,\Pi_2$ are uniform and $\Pi_1$ has irregularity at most $D$. Furthermore, we assume that for every $T\in V_2$, the bipartite subgraph $X_{|T}$ containing all subsets of $T$ in $X$ is a bi-regular unweighted graph. \\

We define the encoding $E_X$ as follows. Consider the graph obtained by restricting $X$ to layers $V_0,V_1$, and let $G$ be its flattening on vertex sets $V_0$ and $V'_1$, see \prettyref{def:unweighted-ver} (in the perfectly regular case, there's no need for flattening). Next, we define $E_X=E_G$ to be the ABNNR encoding as given in \prettyref{def:ABNNR}.
Namely, as $G$ is a bipartite graph on vertex sets $[n]$ and $V'_1$, we have
$E_X:\Sigma_0^n\rightarrow\Sigma_1^{|V'_1|}$. The encoding of a string $z\in \Sigma_0^n$ is given by
\[ \forall S\in V_1' \quad (E_X(z))_S := z_{|S}. \]
The code is defined over alphabet $\Sigma_1 = \Sigma_0^{m_1}$, where $m_1$ is the size of subsets in $V'_1$. The blocklength of the code is $|V'_1|$, which is bounded by $D\abs{V_1}$, where $D$ is the irregularity of $X$.

Our main theorem below is an approximate-list-decoding algorithm for $E$. Recall that an algorithm $(\epsilon,1-\gamma,\frac{70}{\gamma^2})$ approximates-list-decodes $E$ if for any input $z\in\Sigma_1^{\abs{V_1'}}$ it outputs a list $L_{out},\abs{L_{out}}\leq\frac{70}{\gamma^2}$ such that  for every $x$ that satisfies $\dist(E(x),z)\leq 1-\gamma$, it holds that $\dist(x,L_{out})\leq\epsilon$.

\begin{theorem}[Main Theorem (formal version of \prettyref{thm:main-informal})]
	\label{thm:main}
	There exists $c>0$ such that the following holds. For every $\gamma,\eps>0$, let $D$ be some constant and let $\alpha,\beta,\alpha_0,\beta_0>0$ satisfy
	\begin{align*}
	\alpha_0 &\leq \frac{\epsilon}{4}10^{-\frac{\gamma}{8}}, &
	\beta_0 & \leq  \frac{\epsilon \gamma}{1000(c^{\frac{8}{\gamma}}+1)},\\
	\alpha &\leq\frac{\epsilon\sqrt{\gamma}}{1000}, & \beta& \leq\frac{\epsilon}{1000(c^{\frac{8}{\gamma}}+1)}.
	\end{align*}
	There exists $c'>0$ such that if $(X=(V_2,V_1,V_0),W)$ is an $((\alpha,\beta), (\alpha_0,\beta_0))$ double sampler, with irregularity at most $D$, the encoding $E=E_X$ has a randomized approximate-list-decoding algorithm that runs in time polynomial in $n$ with success probability $1-e^{-c'n}$ and parameters $(\epsilon,1-\gamma,\frac{70}{\gamma^2})$.
	
	Here $n=|V_0|$ and we assume that $\abs{V_2},\abs{V_1}= \Theta(n)$ and the subsets in $V_1,V_2$ have bounded size.
\end{theorem}

\prettyref{thm:doublesampler} proves the existence of double samplers with the required parameters. Combining it with the above theorem results in the following corollary.
\begin{corollary}\label{cor:main-ds-exists}
	For every $\gamma,\eps>0$ there exist an integer $D$, a constant $c'>0$ and an infinite family of bounded-degree bipartite graphs ${G_n=(A_n,B_n,E_n)}$ such that $E_{G_n}:\Sigma_0^{A_n}\to \Sigma_1^{B_n}$ has rate $ \exp(-\poly(\frac{1}{\epsilon}\exp(-\frac{1}{\gamma})))$ and a randomized polynomial time approximate-list-decoding algorithm with parameters $(\epsilon,1-\gamma,\frac{70}{\gamma^2})$ and success probability $1-e^{-c'|G_n|}$. Moreover, for every large enough $m$, there is some $G_n$ with $m\le |A_n|\le Dm$.
\end{corollary}
\begin{proof}
	Fix $\gamma,\epsilon >0$, and let $\alpha_0 = \frac{\epsilon}{4}10^{-\frac{\gamma}{8}},
	\beta_0  =  \frac{\epsilon \gamma}{1000(c^{\frac{8}{\gamma}}+1)},
	\alpha =\frac{\epsilon\sqrt{\gamma}}{1000},  \beta =\frac{\epsilon}{1000(c^{\frac{8}{\gamma}}+1)}$.
	
	Let ${X_n}$ be an infinite family of $((\alpha,\beta),(\alpha_0,\beta_0))$ double samplers, promised from \prettyref{thm:doublesampler}.  The family is dense, such that for every integer $m\in\mathbb{N}$ there exists $X_n=(V_1,V_1,V_0)$ such that $m\leq\abs{V_0}\leq Dm$. Furthermore, for every $n$ the double sampler $X_n=(V_2,V_1,V_0)$ has irregularity at most $D= \exp(\poly(\frac{1}{\alpha\beta\alpha_0\beta_0}))$ and $\abs{V_2},\abs{V_1}\leq D\abs{V_0}$. The vertices in $V_1$ are $m_1$-sets, for $m_1 = \poly(\frac{1}{\alpha\beta\alpha_0\beta_0})$ and the vertices in $V_2$ are $m_2$-sets for $m_2 = \poly(\frac{1}{\alpha\beta\alpha_0\beta_0})$.
	
	For every ${X_n}=(V_2,V_1,V_0)$, let $G_n=(A_n,B_n,E_n)$ be the flattening of the bipartite graph $X_n(V_1,V_0)$. From \prettyref{claim:weighted_to_un}, $G_n$ is an unweighted bipartite graph with $\abs{A_n} = \abs{V_0}$ and $\abs{B_n}\leq D\abs{V_1}\leq D^2\abs{V_0}$. In addition, the degree of each vertex $b\in B_n$ is $m_1$.
	
	Let $c'$ be the constant from \prettyref{thm:main}. From the theorem, the ABNNR encoding $E_{G_n}$ has a polynomial time approximate-list-decoding algorithm with parameters $(\epsilon,1-\gamma,\frac{70}{\gamma^2})$ and success probability $1-e^{-c'|A_n|}$. The rate of the code is $\frac{\abs{A_n}}{\abs{B_n}}\frac{1}{m_1}\geq \frac{1}{D^2 m_1} = \exp(-\poly(\frac{1}{\alpha\beta\alpha_0\beta_0}))=\exp(-\poly(\frac{1}{\epsilon}\exp(-\frac{1}{\gamma})))$.
\end{proof}

In this section $X$ is always an $((\alpha,\beta), (\alpha_0,\beta_0))$ double sampler and $E$ is the encoding defined above. We show that $E$ is combinatorially approximate-list-decodable in \prettyref{sec:comb_list_dec}, in \prettyref{sec:dec} we present the list-decoding algorithm and in \prettyref{sec:alg-param} we discuss its parameters . We prove \prettyref{cor:list-dec} in \prettyref{sec:list-dec-c}.

\subsection{Combinatorial approximate-list-decoding}\label{sec:comb_list_dec}
We briefly show that the encoding $E$ is combinatorially approximate-list-decodable. Let $G$ be the flattening of $X(V_1,V_0)$, as in the previous section. From \prettyref{claim:sampler}, \prettyref{item:10sampler}, $X(V_0,V_1)$ is an $(\alpha+\alpha_0,\beta+\beta_0)$ sampler. It follows from \prettyref{claim:weighted_to_un} that $G$ is also an $(\alpha+\alpha_0,\beta+\beta_0)$ sampler.

By \prettyref{lem:ABNNR} our encoding $E$ is
$(\alpha+\alpha_0,1-\gamma,l=\frac{\gamma-\beta-\beta_0}{\gamma^2-\beta-\beta_0})$ approximate-list-decodable code, for every $\gamma>\sqrt{\beta+\beta_0}$.
Choosing $\gamma=\sqrt{2\beta+\beta_0}$, we get,
\begin{corollary}\label{cor:list_dec}
	$E$ is an $(\alpha+\alpha_0,1-\sqrt{2(\beta+\beta_0)},l=\frac{2}{\beta+\beta_0})$ approximate-list-decodable code.
\end{corollary}
The above corollary is about combinatorial list-decoding, and does not imply that there is an efficient algorithm that approximate-list-decodes $E$ with these parameters.
\subsection{The list-decoding algorithm}\label{sec:dec}
In this section we describe a polynomial time list-decoding algorithm for $E$.
Denote the input $z'\in\Sigma_1^\abs{V_1'}$. We interpret it as $\{z'_S\}_{S\in V_1'}$.
As a preprocessing step, we create $z\in\Sigma_1^{\abs{V_1}}$ as follows: for every $S\in V_1$, the multi-set $V_1'$ has at most $D$ copies of $S$. The algorithm picks at random one such copy, $S'\in V_1'$, and sets $z_S = z'_{S'}$ (this preprocessing step is not needed in the perfectly regular setup).
\begin{enumerate}
	\item \textbf{Approximate-list-decoding of Local Views}\label{item:app_loc}
	
	For every $T \in V_2$, the graph $X_{|T}$ is an unweighted $(\alpha_0,\beta_0)$ sampler. The restriction of $E$ to $X_{|T}$ is the ABNNR encoding over the sampler $X_{|T}$.
	We apply the well-separated list-decoding algorithm, \prettyref{alg:well-sep-alg}, on $X_{|T}$ with input $\set{z_S}_{S\in V_1,S\subset T} $, and decoding parameter $(1-\frac{\gamma}{2})$. The algorithm outputs a list $L_T\subset \Sigma_0^{T}$ and a radius $r_T$ such that $\abs{L_T}\leq\ell$ for $\ell\leq\frac{8}{\gamma}$
	and $r_T=2\cdot 10^i\alpha_0$ for some $i\in\{0,\dots\ell\}$.
	The list $L_T$ satisfies
	\begin{itemize}
		\item
		Every $\sigma\in\Sigma_0^T$ such that $\Pr_{S\sim(\Pi_1 |\Pi_2=T)}[\sigma_{|S} = z_S]\geq\frac{\gamma}{2}$ is $r_T$-close to one of the elements in $L_T$, and,
		\item
		$L_T$ itself is $R_T=5r_T$-separated, i.e. $\forall\sigma\neq\sigma'\in L_T, \dist(\sigma,\sigma')>R_T$.
	\end{itemize}
	
	W.l.o.g. the list $L_T$ has size exactly $\ell$, otherwise we add dummy strings that obey the distance requirements.
	
	\item \textbf{Creating a UG constraint graph}\label{item:constraint}

	\begin{itemize}
		\item
		We define the constraint graph
		$(G_C = (V_2,E_C),W=\{w_e\}_{e\in E_C})$
		as follows: The vertices are $V_2$ and for every triple $T_1,S,T_2$ such that $S\subset T_1\cap T_2,S\in V_1$, we have an edge $(T_1,T_2)$ labeled by $S$, denoted by $(T_1,T_2)_S$. The weight of $(T_1,T_2)_S$ corresponds to choosing a random $S \sim \Pi_1$ and then $T_1,T_2$ independently from the distribution $ (\Pi_2|\Pi_1=S)$. Thus, the graph contains parallel edges and self loops. This is the two-step walk graph obtained from $X(V_2,V_1)$, see \prettyref{sec:G2}.
		
		\item For every $T$, the label set is $L_T$ (note that $|L_T|=\ell$).
		
		\item
		Given an edge $(T_1,T_2) \in E_C$ with label $S\subset T_1\cap T_2$,
		set the constraint $\pi$ of $(T_1,T_2)$ with label $S$ as follows:
		\begin{enumerate}
			\item For every $\sigma \in L_{T_1}$ if there is an unmatched $\sigma'\in L_{T_2}$ such that
			\[ \dist_S(\sigma,\sigma') \leq 2(r_{T_1}+\alpha_0), \]
			then set $\pi(\sigma)=\sigma'$.
			\item For every unmatched $\sigma\in L_T$, set $\pi(\sigma)$ to an arbitrary unmatched label.
		\end{enumerate}
		Observe that we always output unique constraints, because we only ever set $\pi(\sigma)$ to an unmatched label.
	\end{itemize}

	\item \textbf{Finding a large expanding UG constraint subgraph}
	\label{item:find-exp}
	
	For every $i\in\{0,\dots\ell\}$, let $V^{(i)}\subset V_2$ be all $T$ such that $r_T = 2\alpha_010^i$. For every $i$ such that $\mu_{G_C}(V^{(i)}) \ge \frac{1}{2(\ell+1)}$, we run \prettyref{alg:expanding-subgraph} on the graph $G_C$ with the subset $V^{(i)}$.
	The algorithm finds a set $U^{(i)}\subset V^{(i)}$ such that,
	\begin{itemize}
		\item
		$\Pr_{T \sim \Pi_2} [T \in U^{(i)}] \ge \frac{1}{4}\Pr_{T \sim \Pi_2} [T \in V^{(i)}]$.
		\item
		Denote by $G^{(i)}$ the induced subgraph of $G_C$ on $U^{(i)}$. Then, $\lambda_2(G^{(i)})\leq\frac{99}{100}$.
	\end{itemize}

	\item\label{item:solving-ug} \textbf{Solving the Unique Constraints}
	
	For every $i$ as above, we run the unique games algorithm,  \prettyref{alg:list-ug}, on $G^{(i)}$ and get a list $\mathcal{L}^{(i)}$ of  assignments. For each assignment $b:U^{(i)} \to [\ell]$ in $\mathcal{L}^{(i)}$ we define $x\in \Sigma_0^{V_0}$ as follows. For every $j \in V_0$ we pick a random $T\in U^{(i)}$ such that $j\in T$, according to the vertex weights of $G^{(i)}$. Let $\sigma\in\Sigma_0^T$ be the $b(T)$'th element in $L_T$, then we set $x_j=\sigma_j$.
	We add $x$ to $L_{out}$. The output is a list $L_{out} \subset \Sigma_0^n$.
\end{enumerate}

\subsection{Algorithm Parameters}\label{sec:alg-param}
In this section we analyze the parameters of the decoding algorithm, the size of the output list, the runtime and the randomness.

\paragraph{Output List Size} In the first step of the algorithm (see \prettyref{item:app_loc}) the algorithm creates a list $L_T$ and a radius $r_T=2\alpha_0 10^i$ for every $T\in V_2$. The size of $L_T$ is at most $\ell$ and the radius has one of $\ell+1$ possible values. In the second step, \prettyref{item:constraint}, the algorithm creates a constraint graph $G_C$ with constraints $\pi_e:[\ell]\rightarrow[\ell]$.

In the third step, \prettyref{item:find-exp}, the algorithm creates a subgraph $G^{(i)}$ for every $i\in\{0,\dots,\ell\}$. There are at most $(\ell+1)$ such graphs $G^{(i)}$, one for every possible value of $r_T$. For every graph $G^{(i)}$, the algorithm solves a unique games instance and outputs a list containing at most $\ell$ assignments (see \prettyref{item:solving-ug}). For each such assignment the algorithm adds a string to $L_{out}$, therefore $\abs{L_{out}}\leq (\ell+1)\ell\leq \frac{70}{\gamma^2}$.

The same encoding $E$ is combinatorially approximate-list-decodable with list size $O(\frac{1}{\gamma})$, see \prettyref{sec:comb_list_dec}. Our algorithm outputs a list of size $O(\frac{1}{\gamma^2})$, and we don't know how to shorten the list. If we had a way to check if $x$ should be in $L_{out}$, i.e. if it is an $\epsilon$-approximation to $x'\st \dist(E(x'),z')\leq 1-\gamma$,  we could have reduced the output list size to $O(\frac{1}{\gamma})$. If the encoding $E$ is used to amplify the distance of a uniquely-decodable error correcting code $C$, then $E(C)$ has a decoding algorithm with output list of $O(\frac{1}{\gamma})$, see \prettyref{sec:list-dec-c}.

\paragraph{Runtime}
We bound the runtime of the decoding algorithm by going over each step it performs and calculating its runtime.
The preprocessing step takes linear time in $\abs{V_0}$. In the first step the algorithm performs \prettyref{alg:well-sep-alg} on $X_{|T}$ for every $T\in V_2$. Preforming \prettyref{alg:well-sep-alg} on $X_{|T}$ takes $\ell2^{\abs{T}}$. Since $\abs{T}$ and $\ell$ are constants, the total runtime of this step is linear in $\abs{V_2}$.

In the second step the algorithm creates the constraint graph $G_C$ by performing a two-step walk over the bipartite graph $X(V_2,V_1)$. The size of $G_C$ can be seen to be linear in $|V_2|$. It takes polynomial time in $\abs{X}$ to construct it. Creating the constraints for each edge is a local operation which takes constant time for each edge.

In the third step the algorithm goes over all $i\in\{0,\dots,\ell\}$ and finds an expander subgraph $G^{(i)}$ of $G_C$ by running \prettyref{alg:expanding-subgraph}. \prettyref{alg:expanding-subgraph} runs in time polynomial in the original graph size, which means that the entire step takes polynomial time in $\abs{X}$.

In the last step, the algorithm performs \prettyref{alg:list-ug} on all of the graphs $G^{(i)}$ (there are at most $\ell+1$ such graphs). By \prettyref{thm:pre UG}, \prettyref{alg:list-ug} runs in time $\poly(\abs{G^{(i)}})$, so the entire step also takes polynomial time in $\abs{X}$.

\paragraph{Randomness}
The algorithm uses randomness twice directly, and one more time when running the unique games algorithm of Makarychev and Makarychev \cite{MakarychevM2010}. The unique games algorithm of \cite{MakarychevM2010} can be derandomized, as they explain in their paper. The randomness directly inside our algorithm occurs in the following places.
\begin{itemize}
	\item Preprocessing step. This step is avoided when $X$ is a perfectly regular double sampler. In the more general irregular case, the preprocessing step chooses for each $S\in V_1$ a copy of $S$ in $V_1'$. To derandomize, one can reuse the same randomness for every $S\in V_1$. This way, the number of random bits needed is $\log D$, and we can cycle through all of these easily. We can enumerate over all possible random strings in $\{0,1\}^{\log D} $ and generate an output list for each random string. The output of the derandomized algorithm is the union of the output lists of the random strings. Our randomized algorithm succeeds with high probability, so there must be a random string which succeeds.
	\item Final step (step 4) of the algorithm. In this step for each assignment $b:U^{(i)}\rightarrow [\ell]$, and each $j\in[n]$ the algorithm picks a random $T\in U^{(i)}$ which contains $j$ and uses it to define $x(j)$. There are $k=O(1)$ possible choices \footnote{By our assumption, the distribution $\Pi_0$ is uniform, meaning that each $j$ participates in the same number of sets $T\in V_2$, which must be a constant number.} for a set $T$ such that $j\in T$
	and we can instead proceed as follows. First define $x(j)$ according to the first choice for all $j$. This leads to a list of possible codewords. Next, define $x(j)$ according to the second choice for all $j$, and so forth. Finally combine the lists from all possible choices.
	
\end{itemize}
The derandomized algorithm may output a somewhat larger list than the randomized one, as it outputs a list for every choice used. In approximate-list-decoding, we are not able to prune the output list and shorten it (see discussion in the beginning of \prettyref{sec:approximate_ecc}). Luckily, if the encoding $E$ is used to amplify the distance of an ECC $C$, the list size does not increase, see the next section for more details.

\subsection{Proof of \prettyref{cor:list-dec}} \label{sec:list-dec-c}
In this section we show how our main theorem implies a list-decoding algorithm for the code $E(C) = \sett{E(x)}{x\in C}$.
\begin{corollary}\label{cor:list-dec-full}
	For all $\gamma,\epsilon>0$, let $(X=(V_2,V_1,V_0=[n]),W)$ be a double sampler with parameters $\alpha,\beta,\alpha_0,\beta_0$ as in \prettyref{thm:main}, and irregularity at most $D$. Let $C\subset\Sigma_0^n$ be an error correcting code with a polynomial time unique-decoding algorithm from an $\epsilon$-fraction of errors. Then the error correcting code $E(C)$ has a randomized polynomial time $(1-\gamma,l=\frac{\gamma -\beta-\beta_0}{\gamma^2-\beta-\beta_0})$ list-decoding algorithm.
\end{corollary}
Plugging in the double samplers from \prettyref{thm:doublesampler}, as done in \prettyref{cor:main-ds-exists}, results in an error correcting code with rate $\exp(-\poly(\frac{1}{\epsilon}\exp(-\frac{1}{\gamma}))) \cdot\text{Rate}(C)$.

\begin{proof}
	Fix an input $z\in\Sigma_1^{\abs{V_1'}}$, the list-decoding algorithm of $E(C)$ on $z$ proceeds as follows.
	\begin{itemize}
		\item Run the approximate-list-decoding algorithm for $E$ on input $z$, receive a list $L_{out}$ of size at most $\frac{70}{\gamma^2}$.
		\item For each $x\in L_{out}$, run the decoding algorithm of $C$ on $x$.
		\begin{itemize}
			\item If failed, do nothing.
			\item If it outputs $x'\in C$, insert $x'$ into $L'$ only if $\dist(E(x'),z)\leq 1-\gamma$.
		\end{itemize}
	\end{itemize}
	
	The above algorithm is polynomial time since the decoding algorithms of $E,C$ are polynomial, and $\abs{L_{out}}$ is constant.
	
	We prove the correctness of the algorithm. Fix $z'\in\Sigma_1^{\abs{V_1'}}$, and let $y\in C$ be a string satisfying $\dist(z,E(y))\leq 1-\gamma$. From \prettyref{thm:main}, there is $x\in L_{out}$ such that $\dist(x,y)\leq\epsilon$. The unique-decoding algorithm of $C$ on input $x$ should return $y$, and the algorithm inserts $y$ into $L'$ as $\dist(E(y),z)\leq 1-\gamma$.
	
	We are left with bounding the list size. The encoding $E$ is $(1-\gamma,\alpha+\alpha_0,l=\frac{\gamma-\beta-\beta_0}{\gamma^2-\beta-\beta_0})$ combinatorially approximate-list-decodable (see \prettyref{sec:comb_list_dec}). The error correcting code $C$ has distance at least $2\epsilon$, and $\epsilon>\alpha+\alpha_0$ (by the conditions of \prettyref{thm:main}). \prettyref{claim:comb-list-dec} proves that $E(C)$ is $(\alpha+\alpha_0,l=\frac{\gamma-\beta-\beta_0}{\gamma^2-\beta-\beta_0})$ combinatorially list-decodable, that is, that $L = \sett{x\in C}{\dist(E(x),z)}$ has size at most $l$. The output list of the algorithm $L'$ satisfies $L'\subseteq L$, because the algorithm inserts $x'\in C$ into $L'$ only if $\dist(E(x'),z)\leq 1-\gamma$, so $\abs{L'}\leq\abs{L}\leq l$.
\end{proof}

\section{Proof of Correctness}\label{sec:correctness}

Fix $g:V_0 \rightarrow \Sigma_1$ such that $\Pr_{S \in V_1'}[ g_{|S}=z'_S] \geq \gamma$. It suffices to show that with high probability, $\dist(g,L_{out})\leq\epsilon$.  Whenever we say ``correct'' in this section we always mean correct with respect to the fixed function $g$.

Let $a:V_2\rightarrow[\ell]$ be an assignment. For $T\in V_2$ we overload notation and denote by $a(T)$ the $a(T)$'th element in the list $L_T$ (formally this is $L_T(a(T))$). Similarly we treat a constraint $\pi_{(T_1,T_2)_S}:[\ell]\rightarrow[\ell]$ as $\pi_{(T_1,T_2)_S}:L_{T_1}\rightarrow L_{T_2}$.

\paragraph{Preprocessing Step}

We first prove that $z\in \Sigma_1^{V_1}$ constructed from the input $z'\in\Sigma_1^{V_1'}$ is noticeably correlated with $g$.
\begin{claim}
	With probability at least $1- e^{-c'n}$, for some constant $c'$,
	\[\Pr_{S\sim \Pi_1}[z_S = g_{|S}]\geq \frac{3}{4}\gamma. \]
\end{claim}
\begin{proof}
	Each $S\in V_1$ has possibly several copies in the multiset $V_1'$. The algorithm chooses $S'$ to be a random copy, and sets $z_S = z'_{S'}$.
	It remains to use a tail bound to show that it is highly unlikely that the correlation drops below $3\gamma/4$ after moving to $\set {z_S}$.
	
	For every $S\in V_1$, let $I_S$ be the random variable that indicates the event $z_S = g_{|S}$. 	
	The input $z'$ is $\gamma$-close to $g$:
	\[ \E_{S\sim \Pi_1}[I_S] = \Pr_{S'\in V_1'}[z'_{S'} = g_{|S'}] \geq \gamma. \]
	We use a Chernoff tail bound (see \prettyref{sec:prelim}) on the independent random variables $\{I_S\}_{S\in V_1}$. We define the random variable $I=\sum_{S\in V_1}\Pi_1(S)I_S$, then $\E[I]=\E_{S\sim\Pi_1}[I_S]$ and define $\nu = \sum_{S\in V_1}\Pi_1^2(S)\E[I_S]$. The distribution $\Pi_1$ has irregularity at most $D$, which lets us bound $\nu$: $\nu\leq \frac{D}{\abs{V_1}}\sum_{S\in V_1}\Pi_1(S)\E[I_S]\leq \frac{D}{\abs{V_1}}\E[I]$.
	\[ \Pr\Brac{\E[I] < \frac{3}{4}\gamma}\leq e^{-\frac{1}{32D}\abs{V_1}\gamma}. \]
	Since $D,\gamma$ are constants and $\abs{V_1}=\Theta(n)$,we pick $c'= \frac{1}{32D}\gamma\frac{n}{\abs{V_1}}$ and finish the proof.
\end{proof}
In the rest of the proof, we assume that $z$ is such that $\Pr_{S\sim \Pi_1}[z_S = g_{|S}]\geq \frac{3}{4}\gamma$.

\subsection{Approximate-list-decoding of local views}\label{sec:local-views}
The decoding algorithm takes each $T\in V_2$, and applies the well-separated list-decoding algorithm, \prettyref{alg:well-sep-alg}, to the graph $X_{|T}$ on input $\set{z_S}_{S\in V_1,S\subset T}$ and closeness parameter $(1-\frac{\gamma}{2})$. The graph $X_{|T}$ is an $(\alpha_0,\beta_0)$ unweighted sampler, and $\gamma > 10\sqrt{\beta_0}$.
\prettyref{alg:well-sep-alg} outputs a list $L_T$ and radius a $r_T$. From \prettyref{lem:well-sep}, the list and radius satisfy
\begin{eqnarray*}
	\abs{L_T}  = \ell  \le  \frac{\frac{\gamma}{2}-\beta_0}{\frac{\gamma^2}{4}-\beta_0} \le\frac{8}{\gamma},
\end{eqnarray*}
and $r_T = 2\alpha_0 10^i$ for some $i\in \{0,\dots \ell\}$. The list $L_T$ is $R_T = 5r_T$ separated.

\begin{definition}(Correct vertex)
	A vertex $T \in V_2$ is \emph{correct} if there exists $\sigma\in L_T$ such that $\dist_T(g,\sigma) \leq r_T$.
\end{definition}
We are left with showing that almost every $T\in V_2$ is correct.

\begin{claim}
	\label{claim:correct vert}
	$\Pr_{T \sim \Pi_2}[T \text{ is correct}]\geq 1-\beta$.
\end{claim}

\begin{proof}
	We look at the bipartite sampler $(X(V_2,V_1),\Pi_{2,1})$. Let
	$f:V_1\rightarrow [0,1]$ equal $1$ if $z_S = g_{|S}$ and $0$ otherwise.
	We know that
	$\E_{S \sim \Pi_1}[f(S)] \ge \frac{3}{4}\gamma > 3\alpha$. Let
	\[Bad = \sett{ T \in V_2 }{ \E_{S \sim T}[f(S)]< \frac{3}{4}\gamma - \alpha}.\]
	As $(X(V_2,V_1),\Pi_{2,1})$ is an $(\alpha,\beta)$ sampler, $\Pr_{T \sim \Pi_2}[T \in Bad] \le
	\beta$.
	
	For every $T \not\in Bad$,
	\[ \E_{S \sim T}[z_S = g_{|S}] \ge  \frac{3}{4}\gamma -\alpha \ge \frac{\gamma}{2}. \]
	From \prettyref{lem:well-sep}, for every $T$ such that the above equation holds, the list $L_T$ contains $\sigma$ such that $\dist_T(\sigma,g)\leq r_T$.
\end{proof}

\subsection{Creating a UG constraint graph}
The algorithm creates a constraint graph $(G_C=(V_2,E_C),W)$ which is a two-step random walk on the weighted bipartite sampler graph $(X(V_2,V_1),\Pi_{2,1})$. For each edge $(T_1,T_2)_S$ in $G_C$, the algorithm creates a matching between the lists $L_{T_1},L_{T_2}$.
In this section we prove that for every radius $r=2\alpha_0 10^i$, and for most edges $(T_1,T_2)_S\textsl{}$ for which $r_{T_1}=r_{T_2}$, the algorithm matches the list element closest to $g$ in $L_{T_1}$ to the list element closest to $g$ in $L_{T_2}$.

\begin{definition}[Correct assignment and constraint]
	\label{def:correct_list_i} Let $a:V_2\rightarrow [\ell]$ be the assignment that for every $T\in V_2$ assigns the element in $L_T$ closest to $g|_T$. Ties are broken arbitrarily. We say $\pi$ is \emph{correct} on $e=(T_1,T_2)_S \in E$ if $\dist_{T_1}({\cind{T_1}},g)\leq  r_{T_1}$, $\dist_{T_2}({\cind{T_2}},g)\leq r_{T_2}$ and $\pi_e(\cind{T_1})=\cind{T_2}$.
\end{definition}

\begin{definition}
	\label{def:correct_subset}
	Let $T$ be a correct vertex. For $S \in V_1$, $S \subseteq T$, we say that the pair $(S,T)$ is \emph{correct} if
	\begin{itemize}
		\item
		$\dist_S(\cind{T},g)  \leq  r_T+\alpha_0$, and,
		\item
		For every $\sigma \in L_T,\sigma \neq  \cind{T}$ we have  $\dist_S(\sigma,g) > R_T-r_T-\alpha_0$.
	\end{itemize}
\end{definition}

We prove that if $T$ is a correct vertex, then for almost all of $S\subset T, S\in V_1$, the pair $(S,T)$ is correct.
\begin{claim}
	\label{claim:correct half edge}
	Suppose $T \in V_2$ is correct. Then
	$\Pr_{S\sim(\Pi_1|\Pi_2=T)}[(S,T)\text{ is correct}]\geq 1- \ell\beta_0$.
\end{claim}

\begin{proof}
	Fix a correct $T\in V_2$, and denote $X_{|T} = (U_T,T,E)$ (i.e. $U_T$ is all $S\in V_1$ such that $S\subset T$).
	
	For every $\sigma \in L_T$, let
	\[ Bad_\sigma = \sett{S\in U_T}{\abs{\dist_S(\sigma,g)-\dist_T(\sigma,g)}>\alpha_0}. \]
	
	As $X$ is double sampler, $X_{|T}$ is an $(\alpha_0,\beta_0)$ sampler and for every $\sigma\in \Sigma_0^T$, $\Pr_{S \sim (\Pi_1|\Pi_2=T)} [S \in Bad_\sigma] \leq \beta_0$.
	
	For a correct $T$, $\dist_T(a(T),g)\leq r_T$ and for every other $\sigma'\in L_T,\sigma'\neq \cind{T}$, we have
	\[ \dist_T(\sigma',g)\geq \dist_T(\sigma',\cind{T}) - \dist_T(\cind{T},g) \geq R_T-r_T.\]
	This implies that for $S\notin \cup_{\sigma\in L_T}Bad_\sigma$,
	\begin{itemize}
		\item $\dist_S(a(T),g)\leq \dist_T(a(T),g) + \alpha_0 \leq r_T+\alpha_0$.
		\item For every $\sigma'\in L_T, \sigma'\neq a(T)$, we have $\dist_S(\sigma',g)\geq \dist_T(\sigma',g)- \alpha_0 \geq R_T-r_T-\alpha_0$.
	\end{itemize}
	So the pair $(S,T)$ is correct for every $S\notin \cup_{\sigma\in L_T}Bad_\sigma$, taking a union bound over the probability of $S\in Bad_\sigma$ for every $\sigma\in L_T$ finishes the proof.
\end{proof}

\begin{lemma}
	\label{lem:correct const}
	For $T_1,T_2 \in V_2$ and $S\in V_1, S\subset T_1\cap T_2$ let $\pi_{(T_1,T_2)_S}$ denote
	the constraint for the edge $(T_1,T_2)_S$. If $T_1,T_2,(S,T_1),(S,T_2)$ are correct and $r_{T_1} = r_{T_2}=r$, then $\pi_{(T_1,T_2)_S}$ is correct.
\end{lemma}

\begin{proof}
	First notice that $\pi_{(T_1,T_2)_S}$ can match $\cind{T_1}$  to $\cind{T_2}$, because \[\dist_S(\cind{T_1},\cind{T_2}) \leq \dist_S(\cind{T_1},g) + \dist_S(\cind{T_2},g) \leq r_{T_1}+r_{T_2}+2\alpha_0=2(r+\alpha_0).\]
	Next, observe that no other list element can be matched to either $\cind{T_1}$ or $\cind{T_2}$. To see that consider a pair $(\sigma,\cind{T_2})$ for some $\cind{T_1} \neq \sigma \in L_{T_1}$.  Then,
	\begin{eqnarray*}
		\dist_S(\sigma,\cind{T_2}) & \ge &
		\dist_S(\sigma,g)-\dist_S(\cind{T_2},g) \\
		& \ge & R_{T_1}-r_{T_1}-\alpha_0- (r_{T_2}+\alpha_0)\\
		& \ge & R_{T_1}-2(r+\alpha_0) > 2(r+\alpha_0),
	\end{eqnarray*}
	which holds because $R_T =5r, R_T \geq 10\alpha_0$, then $R_T> 4(r+\alpha_0)$.
	Hence $\pi_{(T_1,T_2)_S}$ matches $\cind{T_1}$ to $\cind{T_2}$ and only to $\cind{T_2}$ and vice versa.
\end{proof}

\subsection{Finding a large expanding UG constraint subgraph}
The starting point of this section is the constraint graph $G_C$ which has unique constraints. In this section we show that the algorithm finds an induced subgraph of $G_C$ which is an expander, and that almost all of its constraints are correct.

Recall that $V^{(i)}$ is the set of all $T \in V_2$ with $r_T=2\alpha_0 10^i$. The decoding algorithm goes over all $i$ such that $V^{(i)}$ is not too small, and finds an expander subgraph $G^{(i)}$ of $G_C$ in which all vertices are in $V^{(i)}$. In this section we show that there is at least one $i$ in which the expander graph $G^{(i)}$ is correct.

Denote $\mu(V^{(i)}) = \Pr_{T \sim \Pi_2}[T\in V^{(i)}]$. Let $E_C^{(i)}\subset E_C$ be all edges in $G_C$ such that both endpoints have radius $r=2\alpha_0 10^i$.

\begin{lemma}
	\label{lem:constraint prop}
	Denote $\eta = \beta+ \ell\beta_0$.
	There exists $i\in  \{0,\dots \ell\}$, such that:
	
	\begin{itemize}
		\item $\mu(V^{(i)})  \geq \frac{1}{2(\ell+1)}$, and,
		
		\item
		$\Pr_{(T_1,T_2)_S \in E_C^{(i)}} [\pi_{(T_1,T_2)_S} \text{ is correct}] \geq 1-4\eta.$
	\end{itemize}
\end{lemma}

\begin{proof}
	For every $i$ let
	\begin{eqnarray*}
		p & = &  \Pr_{(S,T) \sim (\Pi_1,\Pi_2)}[T,(S,T)\text{ are correct} ]\\
		p^{(i)} & = & \Pr_{(S,T) \sim (\Pi_1,\Pi_2)}[T,(S,T)\text{ are correct} | T\in V_2^{(i)}].
	\end{eqnarray*}
	Then by  \prettyref{claim:correct vert} and \prettyref{claim:correct half edge}
	\begin{eqnarray*}
		p & = & \Pr[T,(S,T)\text{ are correct} ]
		~ = ~  \Pr[T \mbox{ is correct}] \cdot  \Pr[(S,T)\text{ is correct} ~|~T  \mbox{ is correct}]\\
		&  \ge & (1-\beta) \cdot (1-\ell \beta_0) \ge 1-\beta-\ell \beta_0=1-\eta.
	\end{eqnarray*}
	
	We claim:

	\begin{claim}
		There exists $i\in\{0,\dots\ell\}$ with $\mu(V^{(i)})  \geq \frac{1}{2(\ell+1)}$  and $p^{(i)}\geq 1-2\eta$.
	\end{claim}
	
	\begin{proof}
		Assume towards contradiction that no such $i$ exists.
		Let  \[ q = \sum_{i:\mu(V^{(i)})< \frac{1}{2(\ell+1)}} \mu(V^{(i)}). \]
		Therefore $q < \frac{1}{2}$ and
		
		\begin{eqnarray*}
			p & = & \sum_{i\in[\ell]}\mu(V^{(i)})p^{(i)} =
			\sum_{i:\mu(V^{(i)})\geq \frac{1}{2(\ell+1)}}\mu(V^{(i)})p^{(i)} + \sum_{i:\mu(V^{(i)})< \frac{1}{2(\ell+1)}}\mu(V^{(i)})p^{(i)}   \\
			& < & (1-2\eta)(1-q) + q \\
			&=& 1 -2\eta(1-q) < 1-\eta,
		\end{eqnarray*}
		which contradicts $p \ge 1-\eta$.
	\end{proof}
	
	Now, fix $i$ such that $\mu(V^{(i)})\geq \frac{1}{2(\ell+1)}$ and $p^{(i)}\geq 1-2\eta$. Let $(V^{(i)}E_C^{(i)}),W^{(i)}$ be the induced graph on $V^{(i)}$. We prove the second item,
	\begin{align*}
	\Pr_{(T_1,T_2)_S \sim E_C^{(i)}}& [\text{ one of } T_1,T_2,(T_1,S)(T_2,S) \text{ is incorrect }]\\ \le&
	\Pr_{(S,T_1) \sim (\Pi_1,\Pi_2)} [T_1 \mbox{ or } (T_1,S) \mbox { is incorrect}~|~T_1 \in V_2^{(i)}] \\
	&+ \Pr_{(S,T_2) \sim (\Pi_1,\Pi_2)}[T_2 \mbox{ or } (T_2,S) \mbox { is incorrect}~|~T_2 \in V_2^{(i)}] \\
	\le& 2(1-p^{(i)}) \le 4\eta.
	\end{align*}
	The inequality holds because the weight of an edge $(T_1,T_2)_S \sim E_C^{(i)}$ is the probability of picking $S\sim\Pi_1$, and  $T_1,T_2\sim (\Pi_2|\Pi_1\S)$ independently, given that $r_{T_1}=r_{T_2}=2\alpha_0 10^i$.

	The second item follows by \prettyref{lem:correct const}, because whenever all of $T_1,T_2$, $(T_1,S)$ and $(T_2,S) $ are correct, we have that $\pi_{(T_1,T_2)_S}$ is correct, i.e. 
	\[ \Pr_{(S,T_1,T_2) \sim E_C^{(i)}} [\pi_{(T_1,T_2)} \text{ is incorrect }] \leq 4\eta. \]
\end{proof}

We have $\mu(V^{(i)})\geq \frac{1}{2(\ell+1)} > 10\sqrt{\max\set{\alpha,\beta}}$.
By \prettyref{thm:induced-expander}, \prettyref{alg:expanding-subgraph} returns a subset $U^{(i)}\subset V^{(i)}$ such that
\begin{itemize}
	\item $\mu_{G_C}(U^{(i)})\geq \frac{\mu(V^{(i)})}{4} \ge \frac{1}{8(\ell+1)}$, and,
	\item $\lambda_2(G^{(i)})\leq\frac{99}{100}$.
\end{itemize}
Furthermore, since $\mu(U^{(i)})\geq \frac{1}{4}\mu(V^{(i)})$, from \prettyref{claim:sampler mixing lemma},
\[ \Pr_{(T_1,T_2)\sim E_C} [T_1,T_2 \in U^{(i)} ~|~T_1,T_2\in V^{(i)}]\geq \frac{1}{20}. \]
We bound the probability of an incorrect edge:
\begin{eqnarray}
\label{eq:Vi}
\nonumber
\Pr  [\pi_{(T_1,T_2)} \text{ is incorrect} | T_1,T_2\in U^{(i)}]
& \le &
\frac{\Pr [\pi_{(T_1,T_2)} \text{ is incorrect} | T_1,T_2\in V^{(i)}]}{\Pr [T_1,T_2 \in U^{(i)} ~|~T_1,T_2\in V^{(i)}]}\\
&\leq & 20\cdot 4\eta,
\end{eqnarray}
where correct means correct with respect to $g$, as in the entire proof.

At this point we have found an expanding subgraph $G^{(i)}$ of $G_C$, such that almost all of its edges are correct with respect to $g$. 
\subsection{Solving the unique constraints}
For every graph $G^{(i)}$, the decoding algorithm runs  \prettyref{alg:list-ug} on $G^{(i)}$ and outputs a list of assignments $\mathcal{L}^{(i)}$. Then, the decoding algorithm takes every assignment in $b\in \mathcal{L}^{(i)}$ and transforms it into a string $x\in\Sigma_0^n$. In this section we prove that there exists one such $x$ which approximates $g$.

Let $i\in \{0,\dots\ell\}$ be the index promised from \prettyref{lem:constraint prop}, and let $L_{out}^{(i)}\subset L_{out}$ be the subset of the output list created by the algorithm when running on $G^{(i)}$.
\begin{claim} \label{claim:lout-correct}
	With high probability, there exist $x\in L_{out}^{(i)}$ such that \[ \dist(g,x)\leq r_T+80\eta(c^{\ell}+1)+\alpha+8\beta(\ell+2).\]
\end{claim}
\begin{proof}
	Recall $a:V_2 \to [\ell]$ is the assignment which assigns each vertex $T\in V_2$ the list element closest to $g$. We prove the claim by showing there exists $b\in\mathcal{L}^{(i)}$ which is close to $a$ on $U^{(i)}$. Then we prove that with high probability, the algorithm generates from $b$ a string $x\in\Sigma_0^n$ which approximates $g$, details follows.
	
	By (\ref{eq:Vi}), with probability $1-80\eta$ a random edge $(T_1,T_2)_S\in G^{(i)}$ is correct, that is $\pi_{(T_1,T_2)_S}(a(T_1)) = a(T_2)$. This means that the assignment $a$ has value at least $1-80\eta$ on the unique games instance of $G^{(i)}$.
	The parameters of the double sampler promise that $\eta = \beta+\ell\beta_0$ is small enough to satisfy $80\eta c^{\ell+1}<1$, so
	\prettyref{thm:pre UG} guarantees that \prettyref{alg:list-ug} outputs an assignment $b:U^{(i)} \to [\ell]$ in $\mathcal{L}^{(i)}$ such that
	\begin{eqnarray}\label{eq:dist_from_a}
	\Pr_{T \sim U^{(i)}}[a(T) \ne b(T)] & \le & c^{\ell}80\eta,
	\end{eqnarray}
	where $T \sim U^{(i)}$ is the probability to pick a vertex $T$ according to the weights of $G^{(i)}$.
	
	We show that with high probability, the decoding algorithm on $b$ outputs $x$ which approximates $g$.
	For every $j\in [n]$, the decoding algorithm picks a random $T\sim U^{(i)}$ which contains $j$, and sets $x_j = b(T)_j$. For each $j$, let $err(j)$ be the probability that $x_j$ is decoded incorrectly,
	\begin{eqnarray*}
		err(j) & =&  \Pr_{T\sim U^{(i)}|j\in T} [b(T)_j \neq g_j].
		\label{eq:err_Pi}
	\end{eqnarray*}
	If there is some $j\in[n]$ such that no $T\in U^{(i)}$ contains it, we define $err(j)=1$. $\E_{j\sim\Pi_0}[err(j)]=\E_{j\in[n]}[err(j)]$ is the probability of a random coordinate to be decoded incorrectly ($\Pi_0$ is uniform over $[n]$) and our goal is to bound it.
	
	Choosing a random $T\sim U^{(i)}$ then a random $j\in T$ results in a weighted distribution over $V_0$, denote this distribution by $\Pi_0^{(i)}$.
	We show two things:
	\begin{itemize}
		\item $\E_{j\sim \Pi_0^{(i)}}[err(j)]   \le  r_T+80\eta(c^{\ell}+1)$ on \prettyref{claim:assignment-dist}, and,
		\item $\E_{j\in[n] }[err(j)] \leq \E_{j\sim \Pi_0^{(i)}}[err(j)] + \alpha+8\beta(\ell+1)$ on \prettyref{claim:distribution-dist},
	\end{itemize}
	together this implies $\E_{j\in[n]}[err(j)]   \le  r_T+80\eta(c^{\ell}+1)+\alpha+8\beta(\ell+1)$.
	
	To finish the proof, we need to show that with high probability $\dist(x,g)$ is small. The algorithm chooses for each $j\in [n]$ a random set $T\ni j$ independently, so we can apply a Chernoff tail bound,
	\[ \Pr\Brac{\Pr_{j\in[n]}[w_j\neq g_j] > (r_T+80\eta(c^{\ell}+1)+\alpha+8\beta(\ell+1)) + \beta} \leq e^{-\beta^2 n},\]
	and get that with probability $\exp(-n)$ the string $x$ satisfies $\dist(x,g)\leq r_T+80\eta(c^{\ell}+1)+\alpha+8\beta(\ell+2)$.
	
\end{proof}
The parameters satisfy $r_T\leq\alpha_0 10^{\frac{8}{\gamma}}$, $\ell = \frac{8}{\gamma}$ and $\eta=\beta+\ell\beta_0$, which results in $\epsilon =10^{\frac{8}{\gamma}} \cdot \alpha_0+80(\beta+ {\frac{8}{\gamma}}\beta_0) \cdot (c^{\frac{8}{\gamma}}+1)+\alpha+8\beta(\frac{8}{\gamma}+2)$. In the proof of \prettyref{lem:constraint prop} we use the fact that  $\mu(V^{(i)})\geq \frac{1}{2(\ell+1)} > 10\sqrt{\max\set{\alpha,\beta}}$, which holds for $\gamma > \max \{200\sqrt{\alpha},200\sqrt{\beta}\}$. In order for the well-separated list-decoding algorithm to succeed, we require that that $\gamma > 10\sqrt{\beta_0}$. The unique games algorithm from \prettyref{claim:lout-correct} requires $80(\beta+ {\frac{8}{\gamma}}\beta_0) \cdot c^{\frac{8}{\gamma}+1} < 1$. The $\alpha,\beta,\alpha_0,\beta_0$ which satisfy the conditions of \prettyref{thm:main} satisfy all these requirements.

The algorithm uses randomness in the preprocessing step and in the final step above. In each time the success probability is $1-\exp(-n)$, so the total success probability is also $1-\exp(-n)$.

We remark that our algorithm uses as black box the unique games algorithm of \cite{MakarychevM2010}. Their algorithm is randomized, and can be derandomized without changing the algorithm parameters, as they explain in their paper.

To finish the proof, we are left with proving the two claims.
\begin{claim} \label{claim:assignment-dist}
	$\E_{j\sim \Pi_0^{(i)}}[err(j)]   \le  r_T+80\eta(c^{\ell}+1)$.
\end{claim}

\begin{proof}
	$\E_{j\sim \Pi_0^{(i)}}[err(j)]$ is the probability the following experiment fails: pick $j \sim \Pi_0^{(i)}$ and $T\sim U^{(i)}$ which contains $j$,  and check whether $b(T)_j =g_j$. This is the same distribution as picking $T\sim U^{(i)}$ and then a uniform $j\in T$. Hence,
	\begin{eqnarray*}
		\E_{j\sim \Pi_0^{(i)}}[err(j)] & = & \E_{T\sim U^{(i)}, j\in T}[b(T)_j\neq g_j] \\
		& \le & \Pr_{T \sim U^{(i)}}[\dist_T(b(T),g)\le r_T]  \cdot r_T+\Pr_{T \sim U^{(i)}}[\dist_T(b(T),g) > r_T]  \\
		& \le &
		r_T +(c^{\ell}+1)80\eta.
	\end{eqnarray*}
	The last inequality is true because
	\begin{align*}
	\Pr_{T \sim U^{(i)}}[\dist_T(b(T),g)> r_T]  &\le
	\Pr_{T  \sim U^{(i)}}[\dist_T(a(T),g)> r_T] + \Pr_{T \sim U^{(i)}}[a(T) \ne b(T)]  \\
	&\le  80\eta + c^{\ell}80\eta. \tag{by \prettyref{eq:Vi} and \prettyref{eq:dist_from_a}}
	\end{align*}
	
\end{proof}

Next we prove:
\begin{claim} \label{claim:distribution-dist}
	$ \E_{j\in[n] }[err(j)] \leq \E_{j\sim \Pi_0^{(i)}}[err(j)] + \alpha+8\beta(\ell+1)$.
\end{claim}

\begin{proof}
	Let $Bad \subseteq V_2$ be the set of $T$ for which
	$$\Abs{\E_{j \sim (\Pi_0|\Pi_2=T)}[ err(j)]-\E_{j \in[n] }[err(j)]} \ge \alpha.$$
	By \prettyref{claim:sampler}, the graph $X(V_2,V_0)$ is an $(\alpha,\beta)$ sampler and $\Pr_{T \sim \Pi_2}[T \in BAD] \le \beta$ (recall that $\Pi_0$ is uniform, so $\E_{j \in[n] }[err(j)] = \E_{j \sim\Pi_0}[err(j)]$). Since $\Pr_{T\sim\Pi_2}[T\in U^{(i)}]\geq \frac{1}{8(\ell+1)}$, $\Pr_{T \sim U^{(i)}} [T \in BAD]\leq 8(\ell+1)\beta$.
	
	We have
	\begin{eqnarray*}
		\E_{j\sim \Pi_0^{(i)}}[err(j)]  & = &
		\E_{T\sim U^{(i)}}\Brac{\E_{j \in T} [err(j)]} \\
		& \ge &
		\Pr_{T \sim U^{(i)}} [T \notin BAD]\E_{T\sim U^{(i)}\setminus BAD}\Brac{\E_{j \in T}[err(j)]} \\
		& \ge &
		(1-8\beta(\ell+1))(\E_{j\in [n]} [err(j)]-\alpha),
	\end{eqnarray*}
	where the last inequality is because the term is a convex combination of elements that are within $\alpha$ of the common number $\E_{j\in[n]}[err(j)]$.
\end{proof}

\section{High Dimensional Expanders yield Double Samplers}\label{sec:doublesampler}
\def\D{{\cal D}}
In this section we describe how to construct double samplers from high dimensional expanders, proving \prettyref{thm:doublesampler}. In a nutshell, we take the high dimensional expanders constructed by Lubotzky, Samuels and Vishne \cite{LubotzkySV2005-exphdx}, and let the layers of the double sampler be $V_0=X(0),V_1=X(a),V_2=X(b)$ for appropriately chosen $0<a<b$, and put edges for inclusion of subsets. Below we give some minimal background on high dimensional expanders and prove that this construction is indeed a double sampler.

A $d$-dimensional complex $X$ is given by a collection of $(d+1)$-subsets of a ground set $[n]$, called $d$-faces. For each $i < d$ we define a distribution $\D_{i}$ over $(i+1)$-subsets as follows: choose a $d$-dimensional face uniformly and then remove $d-i$ elements from this set at random. A set of size $i+1$ that has positive probability is called an $i$-face of the complex, and we denote the collection of $i$-faces by $X(i)$. We let $\D_d$ denote the uniform distribution on the top faces and remark that even though $\D_d$ is uniform, $\D_i$ need not be uniform, since some $i$-faces can be contained in more $d$-faces than others\footnote{A very recent work \cite{FriedgutI2020} shows how to construct regular high dimensional expanders, giving rise to uniform $\D_i$ for all $i$, and with parameters very similar to those of \cite{LubotzkySV2005-exphdx}. This was not available when this manuscript was completed yet it gives rise to perfectly regular double samplers which simplify some of the work done here.}. We also set $X(0)=[n]$ to be the ground set.

For every $s \in X(i)$ for $0\leq i \leq d-2$, we define the graph $X_s$:
\begin{itemize}
	\item The vertices are $x \in X(0)$ such that $s \cup \{x\} \in X(i+1)$ is a face (clearly $x\notin s$).
	\item Two vertices $x,y$ are connected by an edge if $s\cup\{x,y\}\in X(i+2)$.
	\item The weight of the edge $\{x,y\}$ is
	\[
	w_s(x,y) := \Pr_{t \sim \D_{i+2}}[t = (s \cup \{x,y\}) \,\mid t\, \supset s].
	\]
	
\end{itemize}

The edge weights $w_s$ define a marginal weight distribution on the vertices in $X_s$. Explicitly, for $x\in X_s$,
\[
w_s(x) = \sum_{y \sim x} w_s(x,y) = \Pr_{t \sim \D_{i+1}}[t \supset s \cup \{x\} \mid t \supset s] = \Pr_{(u,v) \sim (\D_{i+1},\D_{i})}[u = s \cup \{x\} \mid v = s].
\]
The graph $X_s$ is called the $1$-skeleton of the link of $s$ in literature on high dimensional expanders.

There are several different definitions of high dimensional expansion. For our purposes, the most relevant is the one-sided spectral expansion,

\begin{definition}[Spectral high dimensional expander (HDX)]
	A $d$-dimensional complex $X$ is said to be a {\em $\gamma$-spectral high dimensional expander}  if for every $i\leq d-2$, and every face $s\in X(i)$ the graph $X_s$ is an expander with $\lambda(X_s) \leq \gamma$.
\end{definition}
Where $\lambda(G)$ is the second largest eigenvalue in $G$, \emph{not} in absolute value.
In a previous version of this paper we relied on a slightly stronger definition of {\em two-sided} high dimensional expansion, but it turns out that the above (one-sided) definition suffices, and in fact gives us a slightly cleaner result.

Lubotzky, Samuels and Vishne~\cite{LubotzkySV2005-exphdx} constructed an explicit family of Ramanujan complexes.
\begin{theorem}[LSV]\label{thm:LSV}
	For every prime $q\in\mathbb{N}$ and dimension $d\in\mathbb{N}$ there is a sequence of $d$-dimensional complexes $\set{X_n}_n$ which are $\frac{1}{\sqrt{q}}$-spectral high dimensional expanders. The vertex set of $X_n$ has size $q^{cn}$ for all large enough $n\in\mathbb{N}$, for some constant $1<c \leq d^2$. $X_n$ is constructible in time $\poly(n)$ and satisfies the following,
	\begin{itemize}
		\item {\bf Bounded degree:} For each $i<d$ each face $s\in X_n(i)$ is contained in at most $D=q^{d^2}$ $d$-faces.
		\item {\bf Uniform top and bottom:} The distribution $\D_d$ on the top $d$-faces is uniform, and the distribution $\D_0$ on the vertices is uniform as well (but $\D_i$ isn't uniform for $0<i<d$).
	\end{itemize}
\end{theorem}

Let $X$ be a $d$-dimensional simplicial complex. Fix $d= m_2 -1> m_1-1 \ge 0$ and  define a graph as follows. Let $V_2 = X(m_2-1)$, $V_1 = X(m_1-1)$, and $V_0=X(0)$ and look at the inclusion graph $(V_2,V_1,V_0)$ together with the distribution $\Pi_1=\D_{m_1-1}$ on $V_1$ and uniform distributions on $V_2,V_0$.

We first prove that the above construction is a spectral version of double sampler, and then use this lemma to prove \prettyref{thm:doublesampler}.
\begin{lemma}[Spectral version of double sampler]\label{lem:spectral-ds}
	Let $X$ be a $d$-dimensional complex which is a $\gamma$-spectral high dimensional expander, and let $d = m_2 -1> m_1-1 \ge 0$.
	Let  $V_2 = X(m_2-1)$, $V_1 = X(m_1-1)$ and $V_0=X(0)$. Let $G_{2,1}= X(V_2,V_1)$ and $G_{1,0}= X(V_1,V_0)$ be the weighted bipartite inclusion graphs between $(V_2,V_1)$ and $(V_1,V_0)$ respectively, where the distributions over $V_2,V_1,V_0$ are the distributions $\D_2,\D_1,\D_0$ of $X$.
	The following spectral bounds hold,
	\begin{itemize}
		\item  $\lambda_2^\bip(G_{1,0})^2 \le {1/m_1} + O(m_1\gamma)$, and
		\item  $\lambda^\bip_2(G_{2,1})^2 \le  {{m_1}/{m_2}} + O(m_1m_2\gamma).$
	\end{itemize}
\end{lemma}

\begin{proof}
	The lemma follow essentially by combining the analysis of Dinur and Kaufman \cite{DinurK2017} of random walks on high dimensional expanders together with \cite{KaufmanO2020} that shows that the bounds in \cite{DinurK2017} hold also for the more general one-sided expansion case (which is what we defined above), we elaborate next.
	
	For every $i\in [d]$, let $A_i$ be the normalized adjacency matrix of the weighted bipartite graph between $X(i)$ and $X(i+1)$ (see \prettyref{sec:expanders}). The matrix $A_i$ can also be viewed as the random-walk operator, moving from an $i$-face to a random $(i+1)$-face containing it. In \cite[Theorem 5.4]{KaufmanO2020} the authors defined the operator $M^{+}_i\in\mathbb{R}^{{X(i)}\times {X(i)}}$, corresponding to a two-step random walk called the ``non-lazy upper walk''. That is, starting at an $i$-face, go to a random $(i+1)$-face $r$ containing it, and then to an $i$-face contained in $r$ that {\em isn't} the face you started with.
	
	The operator $A_i$  satisfies the following identity with $M^+_i$,
	\[ A_i^tA_i = \frac 1 {i+2} Id + \frac {i+1}{i+2} M^+_i\,.
	\]
	One can take this to be an explicit definition of $M_i^+ := \frac {i+2}{i+1}A_i^tA_i - \frac 1 {i+1}Id$. The bound proven in \cite[Theorem 5.4]{KaufmanO2020}  is that the second eigenvalue of $M_i^+$ is at most $(\frac{i}{i+1}+O(i\cdot \gamma))^{1/2}$, which means that $\lambda_2(A_iA_i^t) = \frac {i+1}{i+2} +O(i\cdot \gamma)$.
	We refer the reader to \cite{DinurK2017} to see more details on $M^+_i$ and how the relation between operators is derived.
	
	The normalized adjacency matrix of the graph $G_{2,1}$ can also be described by the operator $A_{m_1}A_{m_1+1}\cdots A_{m_2}$, i.e. by moving in a random-walk fashion from $m_1$-face to $(m_1+1)$-face to $(m_1+2)$-face and so on all the way to $m_2$-face. Composing the corresponding linear operators, we get that the second eigenvalue is
	\begin{align*}
	\lambda_2^\bip(G_{2,1})\leq& \Paren{\frac{m_1+1}{m_1+2}\cdot \frac{m_1+2}{m_1+3}\cdots \frac{m_2-1}{m_2} }^{1/2} + O(\gamma m_1 m_2) \\=& \Paren{\frac{m_1+1}{m_2}}^\frac{1}{2}O(\gamma m_1 m_2), 
	\end{align*}
	where the error term comes from taking all error terms for every $m_1\leq i\leq m_2$.
	
	The calculation in the case of $G_{1,0}$ is identical, only we go over all $0\leq i\leq m_1$.
	
\end{proof}
We are now ready to prove  \prettyref{thm:doublesampler}, which we restate for convenience.

\begin{reptheorem}{thm:doublesampler}{(restated)}
	For every $\alpha,\beta,\alpha_0,\beta_0>0$ there exist constants  $m_1,m_2,D\in \mathbb{N}$ such that  $m_1,m_2 = \poly(\frac{1}{\alpha\beta\alpha_0\beta_0}), D = \exp(\poly(\frac{1}{\alpha\beta\alpha_0\beta_0}))$, such that there is a family of explicitly constructible double samplers $(X_n,W_n)$ for infinitely many $n\in \mathbb{N}$ satisfying
	\begin{itemize}
		\item  $X_n=(V_2,V_1,V_0)$ is an inclusion graph, where $|V_0|=n$, $V_i \subseteq \binom{V_0}{m_i}$  for $i=1,2$.
		\item $X_n$ is an $((\alpha,\beta),(\alpha_0,\beta_0))$  double sampler.
		\item $|V_1|,|V_2| \le D \cdot n$.
		\item The distributions $\Pi_0,\Pi_2$ are uniform and the distribution $\Pi_1$ has irregularity at most $D$.
		\item For each $m\in \mathbb{N}$ there is some $n\in [m,Dm]$ such that the complex $X_n$ on $n$ vertices is constructible in time $poly(n)$.
	\end{itemize}
\end{reptheorem}

\begin{proof}
	We construct the double sampler from the LSV high dimensional expander promised by \prettyref{thm:LSV}.

	We start by choosing the parameters of the high dimensional expander, we choose $\gamma< 1/(m_1m_2)^2$ small enough so that the term $O(m_1m_2\gamma)$ is negligible with respect to $m_1/m_2$ and $1/m_1$.
	We choose $m_1,m_2$ so that $2/m_1 < \min(\alpha^2{\beta},\alpha_0^2{\beta_0}, ) $ and $2m_1/m_2 < \alpha^2{\beta}$. We set the dimension $d=m_2-1$ and define $D=\gamma^{-2 d^2}$. Summarizing:
	\begin{itemize}
		\item $m_1 = \max(1/2\alpha^2{\beta}, 1/2\alpha_0^2{\beta_0})$
		\item $m_2 = \frac{1}{2}m_1/\alpha^2{\beta}$
		\item $D \le \exp(\poly(m_2)) \le \exp(\poly(1/\alpha,1/\beta, 1/\alpha_0,1/\beta_0))$.
	\end{itemize}
	
	Let $X'$ be the $d$-dimensional $\gamma$-spectral high dimensional expander promised by \prettyref{thm:LSV}, with $\abs{X'(0)}=n\in [n',Dn']$. The theorem states it can be constructed in $\poly(n)$ time. Let $X=(V_2,V_1,V_0)$ be the double sampler defined by $V_2 = X'(d),V_2 = X'(m_1)$ and $V_0=X'(0)$, with the distribution $\Pi_1 = \mathcal{D}_{m-1}$ and uniform distributions on $V_2,V_0$.
	
	From \prettyref{thm:LSV}, each $s\in X'(0)$ is in at most $D$ $d$-faces, which means that the size of $V_2,V_1$ is bounded by $D n$. Moreover, the distributions $\mathcal{D}_d,\mathcal{D}_0$ are uniform, and $\mathcal{D}_{m_1}=\Pi_1$ has irregularity at most $D$ (as each $S\in X'(V_1)$ is contained in at most $D$ $d$-faces).
	
	\prettyref{lem:spectral-ds} proves that the bipartite graph $G_{2,1}$ between $V_2$ and $V_1$ is a spectral expander $\lambda_2^\bip(G_{2,1})^2 \le  {{m_1}/{m_2}} + O(m_1m_2\gamma)$. Our choice of $m_1,m_2,\gamma$  promises that $\lambda^\bip_2(G_{2,1})^2\leq \frac{1}{2}\alpha\beta^2$. According to \prettyref{claim:sampling}, $G_{2,1}$ is an $(\alpha,\beta)$-sampler.
	
	We are left with proving the local sampling property, i.e. showing that for each $T\in V_2$, the graph $X_{|T}$ is an $(\alpha_0,\beta_0)$ sampler.
	$X_{|T}$ is the bipartite graph whose one side is all subsets of $T$ of size $m_1$ and whose other side is all of the elements of $T$. Each subset is connected to its $m_1$ elements, so it is bi-regular and has uniform distribution. The claim on the eigenvalue follows either by invoking again the HDX machinery, or by a more direct argument, from considering the two-step walk and noticing it is a convex combination of the  identity matrix with probability $1/m_1\cdot (m_2-m_1)/(m_2-1)$ and the all-ones matrix (normalized) with remaining probability. Our choice of $m_1,m_2$ promises that $\lambda(X_{|T})\leq \frac{1}{2}\alpha_0^2\beta_0$, so by  \prettyref{claim:sampling} is it an $(\alpha_0,\beta_0)$ sampler.
\end{proof}

\section{Relation Between Sampler and Expansion}\label{sec:expanding subset}
In this section we look at sampler and expander graphs, show when a bipartite expander is also a sampler, and how can sampler derive an expander graph. The main difference between a sampler graph and an expander is that expanders are a \emph{worst case} definition, whereas samplers allow exceptions from the expansion requirements.

\subsection{From spectral gap to a sampler}
We show that every bipartite expander graph is also a sampler graph, and calculate the relation between the expansion and the sampler parameters.

We use a variant of the expander mixing lemma from \cite{DinurK2017}, \prettyref{claim:EML} below, to deduce the sampler property from the spectral gaps. The proof of this claim is very similar to the proof of the expander mixing lemma.
\begin{claim}{\cite[Proposition 2.8]{DinurK2017}}\label{claim:EML}
	Let $(G=(U,V,E),W)$ be a weighted bipartite graph with edge weights $W=\{w_e\}_{e\in E}$,let $f:V\to [0,1]$ and $g:U\to[0,1]$. Then
	$$ \left|\E_{(u,v) \sim E} [ f(v)g(u) ] - \E[f]\E[g] \right| \le \lambda_2^{\bip}(G) \sqrt{\E [f] \E [g]}.
	$$
	Where the expectations are over the weights of the edges and vertices.
\end{claim}

\begin{claim}\label{claim:sampling}
	A weighted bipartite graph $(G=(U,V,E),W)$ with $\lambda = \lambda_2^{\bip}(G)$ is an $(\alpha, \frac {2\lambda^2}{\alpha^2})$ sampler.
	In other words, to get an $(\alpha,\beta)$ sampler, it suffices to take a graph with $\lambda < \frac{1}{2}\alpha\sqrt \beta $.\end{claim}
\begin{proof}
	Let $f:V\to[0,1]$ have $\E[f]= \eta$. Let $A$ be the set of vertices that see too little of $f$
	\[ A = \sett{ u\in U}{ \E_{v\sim u}[f(v)] < \eta-\alpha}, \]
	recall that $v\sim u$ is a random neighbor of $u$.
	Similarly, let $B$ be the set of vertices that see too much of $f$,
	\[ B = \sett{ u\in U}{ \E_{v\sim u}[f(v)] >  \eta+\alpha }. \]
	We will show $\Pr[A]+\Pr[B] \le 2\lambda^2 \eta/\alpha^2$. Write
	\[ (\eta + \alpha)\Pr[B] \le \E_{(u,v)\sim W}[ f(v) 1_B(u)] \le \E[f]\Pr[B] + \lambda \sqrt{\E[f]\Pr[B]}
	\] where the first inequality is by definition of $B$ and the second inequality is relying on \prettyref{claim:EML}. Dividing both sides by $\sqrt{\E[f]\Pr[B]}$ and rearranging, we get
	$\frac{\Pr[B]}{\E[f]} \le \lambda^2/\alpha^2$ so $\Pr[B] \le \lambda^2 \eta/\alpha^2 $.
	
	Similarly for $A$, by \prettyref{claim:EML}
	\[|\E[f]\Pr[A] - \E_{(u,v)\sim W}[ f(v) 1_A(u)] |\le \lambda \sqrt{\E[f]\Pr[A]}
	\] so
	\[\E[f]\Pr[A] - \lambda \sqrt{\E[f]\Pr[A]}
	\le \E_{(u,v)\sim W}[ f(v) 1_A(u)] \le  (\eta - \alpha)\Pr[A]\]
	and again we get $\Pr[A] \le \lambda^2\eta/\alpha^2$.
	
	The function $f:V\rightarrow[0,1]$ has a maximum value $1$, so $\E[f]=\eta\leq 1$ and $\Pr[A],\Pr[B] \le \frac{\alpha^2\eta}{\lambda^2}\leq \frac{\alpha^2}{\lambda^2}$.
	If we want an $(\alpha,\beta)$ sampler, we choose $\lambda < \frac{1}{2}\alpha\sqrt \beta $ to get $2\frac{\alpha^2}{\lambda^2} < \beta$.
\end{proof}

\subsection{From sampler to spectral gap}\label{sec:induced-expander}
A sampler graph is not necessarily an expander, an expander graph doesn't have even a single disconnected vertex, whereas a sampler graph can tolerate a small number of less connected vertices. Nevertheless, we prove in \prettyref{thm:induced-expander} that the two-step random walk over a sampler graph contains a large expander. In fact, we prove that if $G$ is the two-step random walk over a sampler graph (see \prettyref{sec:G2} for the definition), then every large set of vertices in $G$ contains an expanding subgraph.

We restate the theorem for convenience.

 \begin{theorem}[\prettyref{thm:induced-expander} restated]
 	Let $\alpha, \eta, \beta \in (0,1)$ be constants such that $\alpha,\beta < \frac{\eta^2}{100}$.
 	Let $(G_{samp} = (V_2,V_1,E_{s}), W_S)$ be an $(\alpha,\beta)$ sampler.
 	Let $(G=(V_2,E),W)$ be the two-step walk of $G_{samp}$.
 	Then for every set $A \subseteq V_2$ with $\mu_{G}(A) =\eta$, there exists a set $B \subseteq A$ such that:
 	\begin{itemize}
 		\item $\mu_G(B) \ge \frac{\eta}{4}$.
		
 		\item Let $G_B$ be the induced graph of $G$ on $B$ with the same edge weights, then $\lambda_2(G_B) \le \frac{99}{100}$.
 	\end{itemize}
 	Furthermore, given $A$ \prettyref{alg:expanding-subgraph} finds such set $B$ in polynomial time in $\abs{V}$.
 \end{theorem}

We first present the algorithm, then prove its correctness. The idea of \prettyref{alg:expanding-subgraph} is based on \cite{DinurG2018}. In the algorithm, we gradually remove sparse cuts form $G_B$ until reaching an expanding subgraph. We find the sparse cuts using the proof of Cheeger inequality, which is constructive. The proof of the theorem uses the fact the large sets in $G$ expand, see \prettyref{claim:sampler mixing lemma}.

\begin{algorithm}[Finding an expanding subgraph]\label{alg:expanding-subgraph}
	
	The algorithm receives a graph $(G=(V,E),W)$, and a  subset $A\subset V$. The output is a subset $B\subset A$.
	
	\begin{description}
		\item[Initialization]: Set $i=0$, $A_0=A$ and let $(G_0=(A_0,V_0),W_0)$ be the subgraph induced by $A_0$ with the same edge weights $W$.
		
		\item[Graph Improvement]: While $\lambda_2(G_i)\geq \frac{99}{100}$:
		\begin{enumerate}
			\item Find a cut $(U_i,A_i\setminus U_i)$ in $G_i$ such that $\mu_{G_i}(E(U_i,A_i\setminus U_i))\leq\sqrt{2(1-\lambda(G_i))}\mu_{G_{i}}(U_i)$. Let $U_i$ be the smaller part of the cut, i.e. $\mu_{G_i}(U_i)\leq\frac{1}{2}$. See \cite{Chung2005} for the algorithm.
			\item Set $A_{i+1} = A_i\setminus U_i$ and let $(G_{i+1}=(A_{i+1},E_{i+1}),W_{i+1})$ be the subgraph induced by $A_{i+1}$ with the same edge weights $W$.
			\item Increase $i$ to $i+1$.
		\end{enumerate}
		\item[Output:] $B=A_i$.
	\end{description}
\end{algorithm}

\begin{proof}
	We start with the runtime of the algorithm. Each iteration in the loop runs an algorithm for finding a sparse cut, which takes polynomial time. The number of iterations in the loop is bounded by $\abs{A}$, because for each iteration $i$, $A_{i+1}\subsetneq A_i$. Therefore, the total runtime of the algorithm is polynomial in $\abs{G}$.
	
	Let $l$ be the number of steps performed by the algorithm. The output of the algorithm is always expanding, i.e. $G_l$ satisfies $\lambda_2(G_l)\leq \frac{99}{100}$. It remains to show that $\mu_G(A_l)\geq\frac{\eta}{4}$.
	
	Assume towards contradiction that $\mu_G(A_l)< \frac{\eta}{4}$, this implies that $\mu_G(U_1\cup\dots\cup U_{l-1})\geq \frac{3}{4}\eta$. The sets $U_i$ are always set to be the smaller part in the partition $(U_i,A_i\setminus U_{i})$, so $\mu_{G_i}(U_i)\leq \frac{1}{2}$. In the graph $G$, this implies that \[ \mu_G(U_i)= \Pr_{u\sim W}[u\in U_i]
	\leq \Pr_{u\sim W}[u\in A_0] \Pr_{u\sim W}[u\in U_i | u\in A_0]
	\leq \eta \frac{1}{2}.\]
	Therefore, there must be $j\in [l-1]$ such that $\mu_G(U_1\cup\dots\cup U_{j})\in [\frac{\eta}{4},\frac{3\eta}{4}]$. Denote $U=U_1\cup\dots\cup U_{j}$.
	
	We show a contradiction by upper bounding and lower bounding the fraction of edges between $U$ and $A_0\setminus U$ in $G$.
	
	\paragraph{Lower Bound:} From the variant of the expander mixing lemma on $G$, \prettyref{claim:sampler mixing lemma}
	\[ \mu_G(E(A_0\setminus U,U)) \geq (\mu_G(A_0\setminus U) - \beta)(\mu_G(U)-\alpha) \geq (\frac{\eta}{4} - \beta)(\frac{\eta}{4}-\alpha),\]
	where we used the fact that $\mu_G(U),\mu_G(A_0\setminus U)\in \Brac{\frac{\eta}{4},\frac{3\eta}{4}}$.
	\paragraph{Upper Bound:} Let $h=\sqrt{2(1-\frac{99}{100})}$, in particular $h\leq \frac{1}{20}$.
	We upper bound the cut $(U,A_0\setminus U)$ in $G$ by showing that it is contained in the union of all of the cuts used by the algorithm. More explicitly, we show that $E(U,A_0\setminus U) \subset \bigcup_{i=1}^{j}E(U_i, A_i\setminus U_i)$.
	
	Let $e=(u,v)$ be an edge in $E(U,A_0\setminus U)$. Then there exists some $i\leq j$ such that $u\in U_i$. By definition, $A_i = A_0\setminus (\cup_{t<i}U_t)$, so
	$A_0\setminus U \subset A_i\setminus U_i$, and $v\in A_i\setminus U_i$.
	
	The cuts $E(U_i, A_i\setminus U_i)$ are disjoint, so the inclusion implies
	\begin{align}
	\mu_G(E(U,A_0\setminus U))
	\leq \sum_{i=1}^j\mu_G(E(U_i,A_{i}\setminus U_i)). \label{eq:up_bound}
	\end{align}
	
	Each cut $(U_i,A_i\setminus U_i)$ is a sparse cut in $G_i$, and satisfies $\mu_{G_{i}}(E(U_i,A_i\setminus U_i)) \leq h\mu_{G_{i}}(U_i)$. 	We can translate this inequality into an inequality with $\mu_G$ instead of $\mu_{G_i}$, by using the fact that $G_i$ is an induced subgraph of $G$. Using \prettyref{claim:trans measure} we get,
	\begin{align*}
	\mu_G(E(U_i,A_i\setminus U_i))\leq h\mu_{G_{i}}(U_i) \leq  h\mu_G(E(U_i,A_i)).
	\end{align*}
	Using this bound in \prettyref{eq:up_bound}, we get
	\[ \mu_G(E(U,A_0\setminus U))
	\leq \sum_{i=1}^j h\mu_G(E(U_i,A_i)) \leq h\mu_G(E(U,A_0)) .\]
	The last inequality holds since $A_i\subset A_0$ for every $i$, and $U$ is partitioned into $U_1,\ldots ,U_j$. This implies that
	$\bigcup_{i=1}^j E(U_i,A_i)\subset E(U,A_0)$, and that the sets in the union are disjoint.
	
	We use again the variant of the expander mixing lemma on the graph $G$, \prettyref{claim:sampler mixing lemma},
	\[ \mu_G(E(U,A_0)) \leq \mu_G(U)\Paren{\mu_G(A_0) + \alpha} + \beta \leq \frac{3\eta}{4}(\eta + \alpha) + \beta. \]
	Which means that
	\[ \mu_G(E(U,A_0\setminus U)) \leq h\Paren{\frac{3\eta}{4}(\eta + \alpha) + \beta}. \]
	Since $h<\frac{1}{20}$ and $\alpha,\beta\leq\frac{\eta^2}{100}$, we reach a contradiction.
\end{proof}

We are left with proving the ``translation'' between $\mu_G$ and $\mu_{G_i}$.

\begin{claim}\label{claim:trans measure}
	Let $(G=(V,E),W)$ be a weighted graph, and let $(G' = (V',E'),W')$ be an induced subgraph inheriting the weights of $G$. Then for every $V''\subset V',E''\subset E'$:
	\[ \mu_{G'}(E'') = \frac{\mu_G(E'')}{\mu_G(E')}, \quad \mu_{G'}(V'') = \frac{\mu_G(E(V'',V'))}{\mu_G(E')}. \]
\end{claim}
\begin{proof}
	By definition (recall \prettyref{sec:vertex weight}), $\mu_{G'}(E'')$ is the probability to pick a random edge from $E''$ when picking a random edge in $G'$. The weights in $G'$ are the same as in $G$.
	\[ \mu_{G'}(E'') = \Pr_{e\sim W'}[e\in E''] = \Pr_{e\sim W}[e\in E'' | e\in E'] = \frac{\Pr_{e\sim W}[e\in E'']}{\Pr_{e\sim W}[e\in E']} =  \frac{\mu_G(E'')}{\mu_G(E')}. \]
	For vertex weights, the weight of every vertex is the sum of its adjacent edges. The weight of a vertex $v\in V''$ in $G'$ is the sum of its adjacent edges in $G'$.
	\begin{align*}
	\mu_{G'}(V'') =& \Pr_{(u,v)\sim W'}[u\in V''] \\
	=& \Pr_{(u,v)\sim W}[u\in V'' | (u,v)\in E'] \\
	=& \Pr_{(u,v)\sim W}[(u,v)\in E(V'',V') | (u,v)\in E'] \\
	=&  \frac{\mu_G(E(V'',V'))}{\mu_G(E')}.
	\end{align*}	
\end{proof}

\section{List-Decoding of Unique Games over Expanders}
\label{sec:list ug}

The unique games algorithm takes a solvable unique games instance and outputs a single solution. We want a list-decoding algorithm, so we need a list of all possible solutions. We do so by running the unique games algorithm of \cite{MakarychevM2010} multiple times, and removing the solution after each time. We restate the theorem for convenience.

 \begin{theorem}[\prettyref{thm:pre UG} restated]
 	Let $(G=(V,E),W)$ be a weighted undirected graph with
 	$\lambda_2(G)\leq\frac{99}{100}$.
 	Let $\set{\pi_e}_{e \in E}$ be unique constraints over the edges of $G$, with $\ell$ labels.
	
 	Then, there is an absolute constant $c>1$ and a polynomial time algorithm, \prettyref{alg:list-ug}, that on input $(G=(V,E),W),\set{\pi_e}_{e \in E}$ outputs a list of assignments $L=\set{a^{(1)},\dots, a^{(t)}}$, with $a^{(i)}:V \to [\ell]$.
 	The list satisfies that for every assignment $a:V\rightarrow[\ell]$ that satisfies $1-\eta$ of the constraints for $\eta <c^{-\ell-1}$,
 	there exists $a^{(i)}\in L$ that satisfies $\Pr_{v\sim W}[a(v) = a^{(i)}(v)]\geq 1-\eta c^{\ell}$.
 \end{theorem}

The constant $c$ is derived from the constant of the unique games algorithm from  \cite{MakarychevM2010}, see \prettyref{thm:unique games}. For $c',C'$ the constants in \prettyref{thm:unique games}, we set $c=\max\set{\frac{100}{c'}, 101(1+50C')}$.

\begin{algorithm}[List decoding unique games]\label{alg:list-ug}
	The algorithm receives a weighted constraint graph $G=(V,E), W=\set{w_e}_{e\in E}, \set{\pi_e}_{e\in E}$ and returns a list of assignments $L=\set{a^{(1)},\dots a^{(t)}}$, $a^{(i)}:V\rightarrow[\ell]$.
	\begin{description}
		\item[Initialization]: Set $i=1$, and set $\pi^{(1)}_e = \pi_e$ for every $e\in E$.
		\item[Solving unique constraints]: Repeat
		\begin{enumerate}
			\item Use the unique games algorithm from \cite{MakarychevM2010} (see \prettyref{thm:unique games}) on the graph $G$ with constraints $\set{\pi^{(i)}}_{e\in E}$.
			\item If the algorithm didn't return a solution, quit the loop.
			\item Otherwise, let $a^{(i)}:V\rightarrow[\ell-i+1]$ to be the solution.
			\item Let $\set{\pi^{(i+1)}}_{e\in E}$ be the constraints after removing $a^{(i)}$ (see details after the algorithm). For every edge $e$, $\pi^{i+1}(e):[\ell-i+1] \rightarrow[\ell-i+1]$.
			\item Set $i = i+1$ and repeat.
		\end{enumerate}
		\item[Output]: Output $L=a^{(1)},\ldots, a^{(i-1)}$, written as assignments from $V$ to $[\ell]$.
	\end{description}
	Removing the assignment $a$ from $\pi:[j]\rightarrow[j]$, getting $\pi':[j-1]\rightarrow[j-1]$ is done as follows:
	\begin{itemize}
		\item For every vertex $v$, reorder the elements such that $a(v)=j$.
		\item If $a$ satisfies $\pi$, i.e. $\pi(j) = j$, then $\pi'$ is equal to $\pi$ restricted to $[j-1]$.
		\item Otherwise, there exist $i,l\neq j$ such that $\pi(l) = j$ and $ \pi(j) = i$. Set $\pi'(l) = i$, and the rest is identical to $\pi$.
	\end{itemize}
\end{algorithm}
Transforming $a^{(j)}:V\rightarrow [\ell - j +1]$ to $a^{(j)}:V\rightarrow [\ell]$ is done by reversing the permutations done in the iterations of the loop. On each iteration of the loop, the algorithm removes the previous solution from all of the constraints by reordering the elements in $[\ell]$. In the output, the algorithm translates back each solution to be $a^{(j)}:V\rightarrow[\ell]$ by reversing the order changes.

Before proving \prettyref{thm:pre UG}, we prove the following simple claim. If $a,a'$ are two assignments satisfying almost all of the constraints in an expander graphs, then they must be either almost identical or completely different.
\begin{claim} \label{claim:equal or dis}
	Let $G$ be a graph with $\lambda_2(G)\leq\frac{99}{100}$, and let $a,a':V\rightarrow[\ell]$ be two assignments satisfying $1-\eta$ and $1-\eta'$ of the constraints in $G$. Then, either $\Pr_{v\in V}[a(v) = a'(v)]\geq 1-50(\eta'+\eta)$ or $\Pr_{v\in V}[a(v) = a'(v)]\leq 50(\eta' + \eta)$.
\end{claim}
\begin{proof}
	Let $D\subset E$ be the set of disagreeing vertices,
	\[ D = \sett{v\in V}{a(v)\neq a'(v)}. \]
	The constraints in $G$ are unique, so for every edge $(v_1,v_2)\in E$ if both $a,a'$ satisfies the edge constraint, and $a(v_1) = a'(v_1)$, then it must be that $a(v_2) = a'(v_2) = \pi_{(1,2)}(a(v_1))$.
	
	Therefore, if an edge $(v_1,v_2)$ has $v_1\notin D,v_2\in D$,  it is not possible that both $a,a'$ satisfy it. This gives a bound on the cut $D,V\setminus D$,
	\begin{align*}\label{eq:sparse disagree cut}
	\mu(E(D,V\setminus D))\leq\eta + \eta'.
	\end{align*}
	The second largest eigenvalue of $G$ is at most $\frac{99}{100}$, so its edge expansion is at least $\frac{1}{50}$. By Cheeger inequality.
	\[ \mu(E(D,V\setminus D)) \geq \frac{1}{50}\min\set{\mu(D),\mu(V\setminus D)}. \]
	This means that $\min\set{\mu(D),\mu(V\setminus D)} \leq 50(\eta+\eta')$, which finishes the proof.
\end{proof}

\begin{proof}[Proof of \prettyref{thm:pre UG}]
	Let $a:V\rightarrow[\ell]$ be an assignment satisfying $1-\eta$ of the constraints of $(G,W)$.
	For $c',C'$ the constants in \prettyref{thm:unique games}, we set $c=\max\set{\frac{100}{c'}, 101(1+50C')}$.
	
	Denote by $t$ the number of solutions in $L$, when running on $(G,W),\set{\pi_e}_{e\in E}$. For every $i\in[t]$, recall $\set{\pi_e^{(i)}}_{e\in E}$ are the constraints used in the $i$th step of the algorithm runtime. Let $\eta_i$ be the fraction of the constraints in the $i$th round unsatisfied by the assignment $a$:
	\[ \eta_i = \Pr_{(u,v)\sim E}[a(u)\neq \pi^{(i)}_{u,v}(a(v))]. \]
	
	In the following claim we show that if $a^{(i)}$ is very different than $a$, then after removing $a^{(i)}$ the assignment $a$ still satisfies a large fraction of the new constraints $\pi^{(i+1)}$.

	\begin{claim}
		If $\eta_i\leq \frac{100}{c}$ and
		$\Pr_u[a(u) = a^{(i)}(u)]\leq \frac{1}{2}$, then $\eta_{i+1}\leq c\eta_{i}$.
	\end{claim}
	\begin{proof}
		The assignment $a$ on $(G,W)$ with constraints $\pi^{(i)}$, satisfies $1-\eta_i$ of the constraints. If $\eta_i \leq\frac{100}{c}$, then the unique games algorithm, \prettyref{thm:unique games}, outputs an assignment $a^{(i)}$ satisfying at least $1-50C'\eta_i$ of the constraints.
		
		The constraints $\pi^{(i+1)}$ are created from $\pi^{(i)}$ by removing the labels of $a^{(i)}$. For every edge $(u,v)$, if $a,a^{(i)}$ differs on \emph{both} endpoints of the edge, then removing $a^{(i)}$ won't ``ruin'' the constrain for $a$.

		This gives us a bound on $\eta_{i+1}$,
		\begin{align*}
		\eta_{i+1} =& \Pr_{(u,v)\sim E}[a(u)\neq \pi^{(i+1)}_{u,v}(a(v))]\\
		\leq& \Pr_{(u,v)\sim E}[a(u)\neq \pi^{(i)}_{u,v}(a(v))] + \Pr_{(u,v)\sim E}[a(u) = a^{(i)}(u) \vee a(v) = a^{(i)}(v)]\\
		\leq&  \eta_{i}+ 100(\eta_{i} + 50C'\eta_{i})\leq c\eta_i,
		\end{align*}
		where the last inequality is by applying \prettyref{claim:equal or dis} on the assignments $a,a^{(i)}$. The constant $c$ such that $c>101(1 + 50C')$.
	\end{proof}
	
	We now show that $L=\set{a^{(i)},\dots a^{(t)}}$ contains an assignment close to  $a$. Assume towards contradiction that for all $i\in[t]$, $\Pr_u[a(u) = a^{(i)}(u)]\leq \frac{1}{2}$, then by the above claim
	\[ 		\eta_t =\Pr_{(u,v)\sim E}[a(u)\neq \pi^{(t)}_{u,v}(a(v))]
	\leq c \eta_{t-1}
	\leq c^2 \eta_{t-2}\leq \cdots
	\leq c^{t}\eta.  \]
	
	Since $\eta <c^{-\ell-1}$, $\eta c^\ell \leq \frac{1}{c}\leq \frac{c'}{100}$ the unique games algorithm should have outputted a solution and not stopped at $t$, reaching a contradiction.
	
	Therefore, there must be $i\in[t]$ such that $\Pr_u[a(u) = a^{(i)}(u)]> \frac{1}{2}$ and for each $i'<i$, $\Pr_u[a(u) = a^{(i')}(u)]< \frac{1}{2}$. The claim above implies that $\eta_i \leq c^{i-1}\eta$. The unique games algorithm outputs $a^{(i)}$ which satisfies at least $1-50C'\eta_i \geq 1-50C'\eta_i$ of the constraints, so by \prettyref{claim:equal or dis}
	\[ \Pr_u[a(u) \neq a^{(j)}(u)]\leq 50(\eta_i+ 50C\eta_i)\leq c^i \eta_i \leq c^{\ell}\eta. \]
	which finishes the proof.
\end{proof}

\bibliographystyle{prahladhurl}

\bibliography{list_dec_with_DS-bib}

\appendix
\section{A Unique Games Algorithm Over Weighted Graphs}\label{app:weighted ug}
Our starting point is the following theorem from \cite{MakarychevM2010}.
\begin{theorem}[Theorem 10, \cite{MakarychevM2010}]
	There exists a polynomial time approximation algorithm that
	given a $1-\delta$ satisfiable instance of unique games on a $d$-regular expander graph
	$G$ with $\frac{\delta}{\lambda_G}\leq c$, the algorithm finds a solution of value
	\[ 1-C\frac{\delta}{h_G}, \]
	where $c$ and $C$ are some positive absolute constants, $\lambda_G$ is the laplacian second smallest eigenvalue and $h_G$ is the edge expansion.
\end{theorem}

In this section we show that the theorem holds also for non-regular weighted graph.
\begin{theorem}[Weighted unique games] \label{thm:unique games}
	There exists a polynomial time approximation algorithm that
	given a $1-\delta$ satisfiable instance of unique games on a weighted expander graph $(G,W)$ such that $\frac{\delta}{\lambda_G}\leq c$, the algorithm finds a solution of value
	\[ 1-C\frac{\delta}{h_G}, \]
	where $c$ and $C$ are some positive absolute constants, $\lambda_G$ is the laplacian second smallest eigenvalue and $h_G$ is the edge expansion.
\end{theorem}
The laplacian smallest non-zero eigenvalue of a graph $G$ is the eigenvalue gap $(1-\lambda_2(G))$, where $\lambda_2$ is defined in \prettyref{sec:expanders}.

We prove \prettyref{thm:unique games} by following the algorithm and proof in \cite{MakarychevM2010} and modifying the parts which are different in the case of a non-regular weighted graph.

The algorithm in \cite{MakarychevM2010} starts from defining an SDP relaxation of the unique games instance. We follow their algorithm and do the same, only in our SDP the target function uses the edge weighted $w_{u,v}$.

For each vertex $u\in V$ and label $i\in[\ell]$ we define a vector $u_i$ of length $t$.
\begin{definition}[SDP relaxation]\label{def:SDP relaxation}
	Minimize:
	\[ \frac{1}{\omega}\sum_{(u,v)\in E}w_{u,v}\sum_{i\in[\ell]}\Norm{u_i - v_{\pi_{u,v}(i)}}^2 \]
	Subject to
	\begin{align}
	\forall u\in V, i\neq j \in [\ell], \quad &\inprod{u_i}{u_j} = 0 \label{eq:orthogonal}\\
	\forall u\in V, \quad & \sum_{i\in[\ell]}\Norm{u_i}^2 = 1 \label{eq:normalization}\\
	\forall u,v,x\in V, i,j,l\in[\ell] \quad &\Norm{u_i - x_l}^2\leq \Norm{u_i - v_j}^2 + \Norm{v_j - x_l}^2 \label{eq:triangle 1}\\
	\forall u,v\in V, i,j\in[\ell] \quad &\Norm{u_i - v_j}^2\leq \Norm{u_i}^2 + \Norm{v_j}^2 \label{eq:triangle 2}\\
	\forall u,v\in V, i,j\in[\ell] \quad &\Norm{u_i}^2\leq \Norm{u_i - v_j}^2 + \Norm{v_j}^2	  \label{eq:triangle 3}
	\end{align}
	Where $\omega = \sum_u w_u$.
\end{definition}
An integral solution sets for each $u\in V$ a label $i\in [\ell]$. It translates into vectors by setting $u_i = \boldsymbol{1}\frac{1}{\sqrt{t}}$, and for each $j\neq i, u_j = \boldsymbol{0}$, where $\boldsymbol{1,0}$ are the all $1$ and all $0$ vectors, respectively. Each integral solution satisfies all of the constraints, so the SDP value is at least the value of the unique games instance.
The algorithm of Makarychev and Makarychev solves the above SDP, and then rounds the SDP solution to get an integral solution with high value. To prove the correctness of the algorithm, we prove that the output of the rounding algorithm has value at least $1-\frac{C}{h_G}\delta$.

Before presenting the rounding algorithm, we define the the earthmover distance, similarly to \cite{AroraKKSTV2008,MakarychevM2010}.
\begin{definition}
	For every two sets of orthogonal vectors $\set{u_i}_{i\in[\ell]},\set{v_i}_{i\in[\ell]}$ let
	\[ \Delta(\set{u_i}_{i\in[\ell]},\set{v_i}_{i\in[\ell]}) = \min_{\tau\in\mathcal{S}_\ell}\Set{\sum_{i\in[\ell]}\Norm{u_i - v_{\tau({i})}}^2}, \]
	where $\mathcal{S}_\ell$ are all permutation over $\ell$ elements.
\end{definition}
For two vertices $u,v$, with the sets of orthogonal vectors $\{u_i\}_{i\in[\ell]},\{v_i\}_{i\in[\ell]}$, a small distance $\Delta(\{u_i\}_{i\in[\ell]},\{v_i\}_{i\in[\ell]})$ means that the sets of vectors are correlated (using a permutation $\tau$), i.e. for every vector $u_i$ there is a vector $v_j$ which is close to it.

For an SDP solution $\set{u_i}_{u\in V,i\in[\ell]}$, we denote by $\Delta(u,v)$ the earthmover distance between the vectors of $u$ and the vectors of $v$.

Arora et al.~\cite{AroraKKSTV2008} showed that on $d$-regular unweighted graphs, the SDP solution has a small average earthmover distance. More explicitly, that for the SDP solution $\set{u_i}_{u\in V,i\in[\ell]}$, the expression $\E_{u,v\in V}[\Delta(\set{u_i}_{i\in [\ell]},\set{v_i}_{i\in [\ell]})]$ is small. The proof in \cite{AroraKKSTV2008} has a lemma and a corollary, the lemma is general for any SDP solution and is not related to the graph. The corollary uses the graph regularity but can be easily modified to hold for weighted graphs as well.
\begin{lemma}[Lemma 2.2 in \cite{AroraKKSTV2008}]\label{lem:vec normalization}
	For every positive even integer $q$ and every SDP solution
	$\set{u_i}_{u\in V,i\in[\ell]}$, there exists a set of vectors $\set{\boldsymbol{V}_u}_{u\in V}$
	that for every pair $u,v\in V$,
	\[ \frac{1}{q}\Norm{\boldsymbol{V}_u - \boldsymbol{V}_v}^2\leq \frac{1}{\ell}\Delta(u,v)\leq 2\Norm{\boldsymbol{V}_u - \boldsymbol{V}_v}^2 + O\Paren{2^{-\frac{q}{2}}}. \]
\end{lemma}

We prove the following corollary, it is the same as the corollary in \cite{AroraKKSTV2008} only for weighted graph. The proof of the corollary is also almost the same.
\begin{corollary}
	For every constant $R\in (0,1)$, there exists a positive $c>0$ such that for any $1-\delta$ satisfiable instance of unique games on $G$, if $\frac{\delta}{\lambda_G}<c$, then
	\[ \E_{u,v\in V}[\Delta(u,v)]\leq R, \]
	For $u,v$ distributed according to their weight, and $\lambda_G=1-\lambda(G)$ the second smallest eigenvalue of the normalized laplacian.
\end{corollary}
\begin{proof}
	By \prettyref{claim:lap ev for vectors}, the second smallest eigenvalue of the laplacian of $G$ can also be represented by
	\begin{align}\label{eq:sec ev}
	\lambda_G  = \min_{\set{z_u}_{u\in V}}\frac{\E_{(u,v)\in E}[\Norm{z_u-z_v}^2]}{\E_{u,v\in V}[\Norm{z_u-z_v}^2]}~,
	\end{align}
	where $\set{z_u}_{u\in V}$ is a set of vectors, one for every vertex, and the expectation is done according to the edge and vertex weights in $G$.
	
	\begin{align*}
	\E_{u,v\in V}[\Delta(u,v)] \leq& 2\ell\E_{u,v\in V}\Brac{\Norm{\boldsymbol{V}_u - \boldsymbol{V}_v}^2} + \ell O\Paren{2^{-\frac{q}{2}}} \tag{by \prettyref{lem:vec normalization}}\\
	\leq& \frac{2\ell}{\lambda_G}\E_{(u,v)\in E}\Brac{\Norm{\boldsymbol{V}_u - \boldsymbol{V}_v}^2} +\ell O\Paren{2^{-\frac{q}{2}}} \tag{by \prettyref{eq:sec ev}}\\
	\leq&\frac{2q\ell}{\lambda_G} \E_{(u,v)\in E}\Brac{\Delta(u,v)}+\ell O\Paren{2^{-\frac{q}{2}}} \tag{by \prettyref{lem:vec normalization}} \\
	\leq& \frac{2q\ell}{\lambda_G}\epsilon  +\ell O\Paren{2^{-\frac{q}{2}}}.\tag{by the SDP solution}
	\end{align*}
	Taking large enough $q$ such that $\ell O\Paren{2^{-\frac{q}{2}}} < \frac{R}{2} $ and $c<\frac{R}{4q\ell}$ ($R,\ell,q$ are all constants), we finish the proof.
\end{proof}

We present the rounding algorithm.
The only difference between our rounding and the rounding in \cite{MakarychevM2010} is that in our case the initial vertex $u$ is picked according to its weight. The proof of correctness is also very similar.

The input is an SDP solution $\set{u_i}_{u\in V,i\in[\ell]}$, the output is an assignment $a:V\rightarrow[\ell]$.
\begin{description} \label{desc:UG solve}
	\item[Initialization]:
	\begin{enumerate}
		\item Pick a random vertex $u\in V$ according to the vertex weights $w_u$.
		\item Pick a random label $i\in[\ell]$, each with probability $\Norm{u_i}^2$.
		\item Pick a random number $t\in [0,\Norm{u_i}^2]$.
		\item Pick a random $r\in[R,2R]$.
		\item Obtain vectors $\set{\tilde{u}_i}_{u\in V,i\in[\ell]}$.
	\end{enumerate}
	\item[Labels Assignment]: For every $v\in V$:
	\begin{enumerate}
		\item Let $S_v = \sett{p\in[\ell]}{\Norm{v_p}^2\geq t, \Norm{\tilde{u}_i - \tilde{v}_p}^2\leq r}$.
		\item If $S_v = \set{p}$, assign the label $p$ to $v$. Else, assign an arbitrary one.
	\end{enumerate}
\end{description}
The vectors $\set{\tilde{u}_i}_{u\in V,i\in[\ell]}$ are a normalized version of the vectors  $\set{{u}_i}_{u\in V,i\in[\ell]}$ that are  promised from \prettyref{lem:tilde-vectors}.
This lemma appears as lemma 1 in \cite{MakarychevM2010}, and is actually proven in \cite{ChlamtacMM2006}. It is a general claim about vectors normalization, and is not related to any graph, therefore it holds for the solution of the SDP on the weighted graph as well.
\begin{lemma}[Lemma 1 from \cite{MakarychevM2010},  proven in \cite{ChlamtacMM2006}.]\label{lem:tilde-vectors}
	For every SDP solution $\set{u_i}_{u\in V,i\in[\ell]}$, there exists a set of vectors $\set{\tilde{u}_i}_{u\in V,i\in[\ell]}$ satisfying the following properties:
	\begin{enumerate}
		\item Triangle inequalities: for every $u,v,w\in V$ and labels $i,j,l\in[\ell]$:
		\[ \Norm{\tilde{u}_i - \tilde{v}_j} + \Norm{\tilde{v}_j - \tilde{w}_l} \leq \Norm{\tilde{u}_i - \tilde{w}_l}.\]
		\item For every $u,v\in V,i.j\in[\ell]$,
		\[ \inprod{\tilde{u}_i}{\tilde{v}_j} = \frac{\inprod{u_i}{v_j}}{\max\set{\Norm{u_i}^2,\Norm{v_j}^2}}. \]
		\item For all non-zero vectors $u_i$, $\Norm{\tilde{u}_i} = 1$.
		\item For every $u\in V,i\neq j\in[\ell]$, $\inprod{\tilde{u}_i}{\tilde{u}_j}=0$.
		\item For every $u,v\in V,i,j\in[\ell]$,
		\[ \Norm{\tilde{v}_j - \tilde{u}_i}\leq \frac{2\Norm{v_j-u_i}}{\max\set{\Norm{u_i}^2,\Norm{v_j}^2}}. \]
	\end{enumerate}
	The set of vectors $\set{\tilde{u}_i}_{u\in V,i\in[\ell]}$ can be obtained in polynomial time.
\end{lemma}

We prove the correctness of the algorithm. Let $(G=(V,E),W=\set{w_{u,v}}_{(u,v)\in E})$ and $\set{\pi_{u,v}}_{(u,v)\in E}, \st \pi_{u,v}:[\ell]\rightarrow[\ell]$ be a unique games instance on a weighted graph, that has a solution satisfying  $1-\delta$ fraction of the constraints (where the fraction is weighted). Then the SDP solution $\set{u_i}_{u\in V,i\in[\ell]}$ also has value at most $\delta$.

We start from a few definitions.
\begin{definition}
	Let $\tau_{x,v}$ be the partial mapping from $[\ell]$ to $[\ell]$ which maps $p$ into $q$ if $\Norm{\tilde{v}_p - \tilde{x}_q}\leq 4R$.
\end{definition}
The function is well defined, because $\set{\tilde{u}_i}_{u\in V,i\in[\ell]}$ are orthogonal and satisfy the triangle inequality. It is not possible that $\Norm{\tilde{v}_p - \tilde{x}_q}\leq 4R$ and $\Norm{\tilde{v}_p - \tilde{x}_{q'}}\leq 4R$, as it implies that $\Norm{\tilde{x}_q - \tilde{x}_{q'}}\leq 8R$, but $\tilde{x}_q,\tilde{x}_{q'}$ are orthogonal.
\begin{definition}
	Let $X = \sett{x\in V}{\abs{S_x} = 1}$.
\end{definition}

In the proof we use the following claims from \cite{MakarychevM2010}. The claims are unrelated to the graph structure, and holds for weighted graphs as well.
\begin{enumerate}
	\item 	If $p\in S_v$ and $q\in S_x$ with non-zero probability for the same initial vertex and label, then $\tau_{v,x}(p) = q$.
	\item $\abs{S_v}\leq 1$.
	\item If $S_v = \set{p}$, then $S_x = \set{\tau_{x,v}(p)}$ or $S_w = \emptyset$.
	\item For every choice of initial vertex $u$, every $v\in V,p\in[\ell]$, $\Pr_{t,r}[S_v = \set{p}]\leq\Norm{v_p}^2$.
\end{enumerate}

We reprove a weighted variant of the following lemmas from \cite{MakarychevM2010}. The Lemmas and proofs are very similar to those in \cite{MakarychevM2010}, the main difference is that in the weighted case the distribution is over the weights of the edges and vertices.
\begin{lemma}[a variant of Lemma 5 in \cite{MakarychevM2010}]
	If $\frac{\epsilon}{\lambda_G}<c$, then $\E[\mu(X)]\geq\frac{1}{4}$.
\end{lemma}
\begin{proof}
	Suppose $u$ is the initial vertex, then for every $v\in V$ we express the probability of $v\in X$ using $\Delta(u,v)$.
	
	For every label $p\in[\ell]$, if $\exists q\in[\ell]$ such that $\Norm{u_q - v_p}^2\leq \frac{R}{2}\Norm{v_p}^2$, then
	\[ \Norm{\tilde{v}_p - \tilde{u}_q}\leq \frac{2\Norm{u_q - v_p}}{\max\set{\Norm{u_q}^2,\Norm{v_p}^2}}\leq R\leq r, \]
	which implies $\tau_{u,v}(p) = q$.
	In this case, if $q$ is the initial label and $t\leq \Norm{v_p}^2$, then $S_v=\set{p} $ which implies $v\in X$. Therefore,
	\begin{align}
	\Pr_{i,t}[S_v = \set{p}]\geq& \Pr_{i,t}[i=q \wedge t\leq\Norm{v_p}^2] \\=&\Norm{u_q}^2\min\set{1,\frac{\Norm{v_p}^2}{\Norm{u_q}^2}} = \min\set{\Norm{u_q}^2,\Norm{v_p}^2} \geq \frac{1}{2}\Norm{v_p}^2.
	\end{align}
	The last inequality is by the triangle inequality, using the fact that $\Norm{u_q - v_p}^2\leq \frac{R}{2}\Norm{v_p}^2$.
	
	Going over all possible labels $p$ for $v$:
	\begin{align*}
	\Pr[v\in X] =& \Pr_{i,t,r}[\exists p \st S_v = \set{p}]\\
	=&\sum_p \Pr_{i,t,r}[S_v = \set{p}]\\
	\geq& \sum_{p \st \exists q,\Norm{u_q - v_p}^2\leq \frac{R}{2}\Norm{v_p}^2} \frac{1}{2}\Norm{v_p}^2\\
	\geq& \sum_{p} \frac{1}{2}\Norm{v_p}^2 - \sum_{p \st \forall q, \Norm{u_q - v_p}^2> \frac{R}{2}\Norm{v_p}^2} \frac{1}{2}\Norm{v_p}^2\\ \tag{since for all $q, \Norm{v_p}^2 <\frac{2}{R}\Norm{v_p-u_q}^2$}
	\geq& \frac{1}{2} - \frac{1}{2}\sum_p\frac{2}{R}\min_q\set{\Norm{v_p - u_q}^2}\\
	=&\frac{1}{2}-\frac{1}{R}\Delta(u,v).
	\end{align*}
	
	By the earthmover distance lemma, $\E_{u,v\in V}[\Delta(u,v)]<R$, when $u,v$ are distributed according to their weight in the graph, so
	\begin{align*}
	\E[\mu(X)] =& \sum_{u,v\in V}\mu(u)\mu(v)\Pr[v\in X | u \text{ initial vertex}] \\
	\geq& \sum_{u,v\in V}\mu(u)\mu(v)\Paren{\frac{1}{2} - \frac{1}{R}\Delta(u,v)}\\
	\geq& \frac{1}{2} - \frac{1}{R}\E_{u,v\in V}[\Delta(u,v)] \\
	\geq& \frac{1}{2}- \frac{1}{4}.
	\end{align*}
\end{proof}
\begin{corollary}
	\[ \Pr[\mu(X)\geq \frac{1}{8}]\geq\frac{1}{8} .\]
\end{corollary}

\begin{lemma}[a variant of Lemma 7 in \cite{MakarychevM2010}]
	\[ \E[\mu(X,V\setminus X)]\leq \frac{6\delta}{R}. \]
\end{lemma}
\begin{proof}
	Fix $u\in V$ the initial vertex, we bound the probability of $v\in X,x\notin X$ by $\frac{6}{R}\sum_p\Norm{v_p -x_{\pi_{x,v}(p)} }^2$.
	
	If $v\in X,x\notin X$, then $S_v = \set{p},S_x = \emptyset$. Let $q = \pi_{x,v}(p)$.
	Since $S_v = \set{p}$, then $\Norm{v_p}^2\geq t,\Norm{\tilde{u}_i - \tilde{v}_p}^2\leq r, i = \tau_{u,v}(p)$. One of the two cases must happen
	\begin{enumerate}
		\item $\Norm{x_q}^2 < t$.
		\item $\Norm{x_q}^2 \geq t$, $\Norm{\tilde{x}_q - \tilde{u}_i}^2 > r$.
	\end{enumerate}
	We sum over all $p$ the probability that these events occur (each $p$ has a  $q =\pi_{x,v}(p)$).
	\begin{align*}
	\Pr_{i,t,r}[1]
	\leq& \sum_p \Pr[i = \sigma_{v,u}(p)]\Pr[\Norm{x_q}^2 < t \leq \Norm{v_p}^2 | i=\sigma_{v,u}(p)]\\
	\leq& \sum_p\Norm{u_{\sigma_{v,u}(p)}}^2\frac{\Norm{v_p}^2 - \Norm{x_q}^2}{\Norm{u_{\sigma_{v,u}(p)}}^2}\\
	\leq& \sum_p \Paren{\Norm{v_p}^2 - \Norm{x_q}^2}.
	\end{align*}
	\begin{align*}
	\Pr_{i,t,r}[2] = & \sum_p \Pr[i = \sigma_{v,u}(p)]\Pr[t\leq \Norm{v_p}^2]\Pr[\Norm{\tilde{u}_i - \tilde{v}_p}^2\leq r < \Norm{\tilde{x}_q - \tilde{u}_i}^2 | i =\sigma_{v,u}(p)]\\
	\leq& \sum_p \Norm{u_{\sigma_{v,u}(p)}}^2\frac{\Norm{v_p}^2}{\Norm{u_{\sigma_{v,u}(p)}}^2}\frac{\Norm{\tilde{x}_q - \tilde{u}_i}^2 - \Norm{\tilde{u}_i - \tilde{v}_p}^2}{R} \tag{triangle inequality}\\
	\leq&\sum_p \Norm{v_p}^2\frac{1}{R}\Norm{\tilde{v}_p - \tilde{x}_q}^2\\
	\leq&\sum_p \Norm{v_p}^2\frac{1}{R} \frac{2\Norm{v_p - x_q}^2}{\max\set{\Norm{v_p}^2,\Norm{x_q}^2}} \\
	\leq& \sum_p \frac{2}{R}\Norm{v_p - x_q}^2.
	\end{align*}
	Therefore, for every edge $(v,x)$, $\Pr[(v,x)\in E(X,V\setminus X)]\leq \sum_p (1+\frac{2}{R})\Norm{v_p - x_q}^2$.
	
	The expected value of the cut:
	\begin{align*}
	\E[\mu(E(X,V\setminus X))] =& \frac{2}{\omega}\sum_{(v,x)\in E}w_{v,w}\Pr[(v,x)\in E(X,V\setminus X)]\\
	\leq&\frac{2}{\omega }\sum_{(v,x)\in E}w_{v,w}\frac{3}{R}\Norm{v_p - x_q}^2 \tag{SDP value $\leq \delta$}\\
	\leq&\frac{6\delta}{R}.
	\end{align*}
	
\end{proof}

\begin{lemma}
	[a variant of Lemma 8 in \cite{MakarychevM2010}]
	If $\delta < \min\set{c_R\lambda_G,\frac{h_G R}{1000}}$ then with probability at least $\frac{1}{16}$, $\mu(X)\geq 1-\frac{100\delta}{h_G R}$.
\end{lemma}
\begin{proof}
	By the definition of $h_G$, $\mu(E(X,V\setminus X))\geq h_G \min\set{\mu(X),\mu(V\setminus X)}$, which implies
	\begin{align*}
	\frac{6\delta}{R} \geq& \E[\mu(E(X,V\setminus X))]\\
	\geq& h_G \E[\min\set{\mu(X),\mu(V\setminus X)}].
	\end{align*}
	We get that $\E[\min\set{\mu(X),\mu(V\setminus X)}]\leq \frac{6\delta}{h_G R}$, by Markov inequality
	\[ \Pr[\min\set{\mu(X),\mu(V\setminus X)} \leq \frac{100\delta}{h_G R}] \geq 1-\frac{1}{16}.\]
	
	We also know that $\Pr[\mu(X)\geq\frac{1}{8}]\geq \frac{1}{16}$, so with probability at most $1/16$, the set $V\setminus X$ is large $\mu(V\setminus X)\leq \frac{100 \delta}{h_G R}$.
\end{proof}
The following lemma is independent of the graph, so the proof in \cite{MakarychevM2010} holds here as well.
\begin{lemma}[Lemma 9 in \cite{MakarychevM2010}]
	For every edge $(v,x)\in E$,
	\[ \Pr[v,x\in X, (v,x) \text{isn't satisfied}]\leq 4\delta_{v,x}, \]
	for $\delta_{v,x} = \frac{1}{2}\sum_{i\in[\ell]}\Norm{v_i - x_{\pi_{x,v}(i)}}^2$.
\end{lemma}

And we are ready to prove the theorem, the proof is almost identical to the proof in \cite{MakarychevM2010}.
\begin{proof}[proof of \prettyref{thm:unique games}]
	We show that the randomized algorithm described above solves the UG instance with constant probability. It can then easily be derandomized.
	
	The algorithm solves the SDP, then runs the rounding algorithm. If $\mu(X)\geq 1-\frac{100\delta}{h_G R}$, it outputs the labelling, else it fails.
	
	Suppose the algorithm doesn't fail, then by definition
	\[ \mu(E(X,X))\geq 1-\frac{100\delta}{h_G R}, \]
	as $\mu(V\setminus X)\leq \frac{100\delta}{h_G R}$).
	
	The expected fraction of violated constraints inside $X$ is at most,
	\begin{align*}
	\frac{2}{\omega}\sum_{(v,x)\in E}w_{v,x}4\delta_{v,x}\leq& \frac{2}{\omega}\sum_{(v,x)\in E}w_{v,x}4\sum_{p\in [\ell]}\Norm{v_p - x_{\pi_{x,v}(p)}}^2\leq 64\delta.
	\end{align*}
	Therefore with constant probability the algorithm outputs a solution satisfying $1-64\delta-\frac{100\delta}{h_G R}$ of the constraints.
\end{proof}
\subsection{Eigenvalue proof}
\begin{claim}\label{claim:lap ev for vectors}
	Let $G = (V,E)$ be a weighted graph with weights $\set{w_{u,v}}_{(u,v)\in E}$, and let $\mathcal{L}$ be the normalized laplacian matrix of $G$,
	\[ \mathcal{L}_{v,u} = \begin{cases}
	1 \quad& \text{if } u=v\\
	-\frac{w_{u,v}}{\sqrt{w_u w_v}} \quad &\text{if } (u,e)\in E\\
	0 \quad &\text{else}
	\end{cases}, \]
	where $w_u = \sum_{v \st (u,v)\in E}w_{u,v}$.
	The second smallest eigenvalue of the laplacian corresponds to
	\[ \lambda_2  = \min_{\set{z_u}_{u\in V}}\frac{\E_{(u,v)\sim w}[\Norm{z_u-z_v}^2]}{\E_{u,v\sim V}[\Norm{z_u-z_v}^2]}~. \]
	Where $\set{z_u}_{u\in V}$ is a set of vectors, $\forall u, z_u\in \mathbb{R}^t$.
\end{claim}
\begin{proof}
	We define a new matrix $\mathcal{L}'\in\mathbb{R}^{\abs{V}t\times\abs{V}t}$, which is composed of $t\times t$ scalar matrix blocks, i.e. for every $u,v\in V$, the matrix $\mathcal{L}'_{u,v}$ is a $t\times t$ scalar matrix, $\mathcal{L}'_{u,v}=I^{t\times t}\mathcal{L}_{u,v}$.
	Formally, we denote each row and column by two indices $u\in V,i\in [t]$ and
	\[ \mathcal{L}'_{(u,i),(v,j)} = \begin{cases}
	\mathcal{L}_{u,v} \quad &\text{if } i=j \\
	0 \quad &\text{else}
	\end{cases}. \]
	
	$\mathcal{L}$ has a single eigen value $0$, the new matrix $\mathcal{L}'$ has $t$ eigenvalues $0$.  One eigenvectors basis for the nullspace is $y^1,\dots y^t\in \mathbb{R}^{\abs{V}t}$,  $y^l_{u,j} = \begin{cases}
	y_u \quad &l=j\\ 0 \quad &\text{else}\end{cases}$, for $y$ the eigenvector of $\mathcal{L}$.
	
	The spectrum of $\mathcal{L}'$ is identical to the spectrum of $\mathcal{L}$, only each eigenvalue repeats $t$ times. Therefore the second largest eigenvalue of $\mathcal{L}$ is equal to the $t+1$ eigenvalue of $\mathcal{L}'$, and is equal
	\begin{align}
	\lambda_2 = \min_{x\in\mathbb{R}^{t\abs{V}}}\Set{\frac{\inprod{x}{\mathcal{L}'x}}{\inprod{x}{x-y^1\inprod{x}{y^1}-\dots -y^t\inprod{x}{y^t}}}}.
	\end{align}
	
	The numerator equals:
	\begin{align*}
	\inprod{x}{\mathcal{L}'x} =& \sum_{u,v\in V,i,j\in[t]}x_{u,i}\mathcal{L}'_{(u,i),(v,j)}x_{v,j}\\
	=&\sum_{u,v\in V,i\in [t]}x_{u,i}\mathcal{L}_{u,v}x_{v,i}\\
	=& \sum_{u\in V,i\in[t]}x_{u,i}^2 - 2\sum_{(u,v)\in E,i\in[t]}\frac{w_{u,v}}{\sqrt{w_u w_v}}x_{u,i}x_{v,i}\\
	=&\sum_{u\in V}\Norm{x_u}^2 - 2\sum_{(u,v)\in E}\frac{w_{u,v}}{\sqrt{w_u w_v}}\inprod{x_u}{x_v}.
	\end{align*}
	Where $x_u$ is the length $t$ vector containing $x_{u,i}$ for $i\in[t]$.
	
	The denominator:
	\begin{align*}
	\inprod{x}{x-y^1\inprod{x}{y^1}-\dots -y^t\inprod{x}{y^t}} =&\inprod{x}{x}-\Paren{\inprod{x}{y^1}}^2-\cdots \Paren{\inprod{x}{y^t}}^2\\
	=& \sum_{u\in V,i\in[t]}x_{u,i}^2 - \sum_{l\in[t]}\Paren{\inprod{x}{y^l}}^2\\
	=&\sum_{u\in V,i\in[t]}x_{u,i}^2 - \sum_{l\in[t]}\sum_{u,v\in V,i,j\in t}x_{u,i}y^l_{u,i}x_{v,j}y^l_{v,j}\\
	=&\sum_{u\in V,i\in[t]}x_{u,i}^2 - \sum_{l\in[t]}\sum_{u,v\in V}x_{u,l}y_{u}x_{v,l}y_{v}\\
	=&\sum_{u\in V}\Norm{x_u}^2 - \sum_{u,v\in V}y_u y_v\inprod{x_u}{x_v}\\
	=&\sum_{u\in V}\Norm{x_u}^2 - \sum_{u,v\in V}\frac{\sqrt{w_u w_v}}{\omega}\inprod{x_u}{x_v}.
	\end{align*}
	We write the expectations explicitly:
	\begin{align*}
	\E_{(u,v)\sim w}[\Norm{z_u-z_v}^2] =& \frac{2}{\omega}\sum_{(u,v)\in E}w_{u,v}\inprod{z_u-z_v}{z_u-z_v}\\
	=& \frac{2}{\omega}\sum_{(u,v)\in E}w_{u,v}(\Norm{z_u}^2+\Norm{z_v}^2-2\inprod{z_u}{z_v})\\
	=& \frac{2}{\omega}\sum_{u\in V}w_u \Norm{z_u}^2 - \frac{4}{\omega}\sum_{(u,v)\in E}w_{u,v}\inprod{z_u}{z_v}.
	\end{align*}
	
	\begin{align*}
	\E_{u,v\sim V}[[\Norm{z_u-z_v}^2] =& \frac{1}{\omega^2}\sum_{u,v\in V}w_u w_v \inprod{z_u-z_v}{z_u-z_v} \\
	=& \frac{1}{\omega^2}\sum_{u,v\in V}w_u w_v (\Norm{z_u}^2+\Norm{z_v}^2-2\inprod{z_u}{z_v}) \\
	=& \frac{1}{\omega^2}\sum_{u\in V}2\omega w_u \Norm{z_u}^2 - \frac{2}{\omega^2}\sum_{u,v\in V}w_u w_v \inprod{z_u}{z_v}
	\end{align*}
	For every $u\in V,i\in [t]$ let $x_{u,i} = \sqrt{w_u}z_{u,i}$,
	\[ \inprod{x}{\mathcal{L}'x} = \frac{\omega}{2}\E_{(u,v)\sim w}[\Norm{z_u-z_v}^2], \]
	\[ \inprod{x}{x-y^1\inprod{x}{y^1}-\dots -y^t\inprod{x}{y^t}} = \frac{\omega}{2}\E_{u,v\sim V}[[\Norm{z_u-z_v}^2]. \]
	The factor of $\frac{\omega}{2}$ cancels out, and the minimum value is not affected by the multiplication in $\sqrt{w_u}$, as it is taken over all vectors in $\mathbb{R}$.
\end{proof}

\end{document}